\RequirePackage{etex}
\documentclass[10pt]{article}
\usepackage[utf8]{inputenc}
\usepackage[a4paper,left = 0.8in,right = 0.8in, top = 1in, bottom = 1in]{geometry}
\usepackage{amsmath}
\usepackage{amsthm}

\usepackage{bm}
\usepackage{cancel}
\usepackage{amssymb}
\usepackage{amssymb}
\usepackage{amsmath,amssymb}
\usepackage{bbm}
\usepackage[natbibapa]{apacite}
\newcommand{\shuffle}{\mathbin{\sqcup\mkern-3mu\sqcup}}
\usepackage{dsfont}
\usepackage{commath}
\usepackage{mathtools}
\usepackage{float}
\usepackage{stmaryrd}
\usepackage{cases}
\usepackage{url}

\usepackage{enumitem,linegoal}

\usepackage{subcaption}
\usepackage{graphicx}
\usepackage{color}
\usepackage{algorithm}

\usepackage[dvipsnames]{xcolor}
\usepackage{hyperref}
\definecolor{Chelsea}{RGB}{3,70,148}
\hypersetup{colorlinks,linkcolor=Chelsea,citecolor=Chelsea,urlcolor = Chelsea}
\usepackage[capitalise]{cleveref}

\theoremstyle{plain}
\newtheorem{theorem}{Theorem}[section]
\newtheorem*{theorem*}{Theorem}
\newtheorem{lemma}[theorem]{Lemma}
\newtheorem{corollary}[theorem]{Corollary}
\newtheorem{proposition}[theorem]{Proposition}
\newtheorem{assumption}[theorem]{Assumption}
\newtheorem{conjecture}[theorem]{Conjecture}
\theoremstyle{definition}
\newtheorem{definition}[theorem]{Definition}
\newtheorem*{definition*}{Definition}
\theoremstyle{remark}
\newtheorem{example}{Example}[section]
\newtheorem{remark}{Remark}[section]
\DeclareMathOperator{\Ss}{\mathbb{S}}
\DeclareMathOperator{\SSS}{\mathbf{S}}
\DeclareMathOperator{\XX}{\mathbb{X}}
\DeclareMathOperator{\XXX}{\mathbf{X}}

\usepackage{fancyhdr}

\usepackage{authblk}
\title{\vspace{-10mm} 
Unbiased Rough Integrators and No Free Lunch in Rough-Path-Based Market Models
\vspace{-3mm}}

\date{\today}
\author{Tomoyuki Ichiba\footnote{Email: {\tt ichiba@pstat.ucsb.edu}. Affiliation: {\tt Department of Statistics and Applied Probability, University of California, Santa Barbara, California, 93106-3110}},  Qijin Shi\footnote{Email: {\tt qijin@ucsb.edu}. Affiliation: {\tt Department of Statistics and Applied Probability, University of California, Santa Barbara, California, 93106-3110}}}

\begin{document}
\maketitle
\vspace{-1cm}

\vspace{0mm}
\begin{abstract}

Built to generalise classical stochastic calculus, rough path theory provides a natural and pathwise framework to model continuous non-semimartingale assets. This paper investigates the capacity of this framework to support frictionless continuous No-Free-Lunch markets à la Kreps-Yan. We establish a ``Rough Kreps-Yan" theorem, which links a No Controlled Free Lunch (NCFL) condition to the unbiasedness of the driver of the price process as a rough integrator. The central work of this paper is a classification of these unbiased rough integrators with respect to different classes of controlled paths as integrands, under some assumptions. As the admissible strategies are enlarged from Markovian-type portfolios to signature-type and adaptedly scaled signature-type portfolios, the admissible random rough paths collapse first to Gaussian-Hermite rough paths, and ultimately to the It\^o rough path lift of a standard Brownian motion, up to a time change. Notably, simple strategies do not appear in the theory. This implies that within our framework, continuous frictionless markets based on rough path theory are inevitably constrained to the classical semimartingale paradigm, clarifying the limits of this approach. Our framework covers $\alpha-$H\"older continuous rough paths for $\alpha>0$ arbitrarily small in the tensor algebra setting.
\end{abstract}
\vspace{2mm}
	
\\ \vspace{-0mm}



\vspace{8mm}

\tableofcontents


\section{Introduction}

The Fundamental Theorem of Asset Pricing (FTAP) elevates semimartingales to the canonical class of price processes by equating frictionless no–arbitrage with the existence of an equivalent (local) martingale measure (cf. \cite{DS94} and \cite{DS98}). This classical paradigm is inseparable from It\^o calculus, which furnishes stochastic integrals and change–of–variables formulae for semimartingale drivers and underpins modern pricing and hedging. In parallel, modern data and modelling practice motivate market noise with memory and multiscale structure (cf.\ \cite{Con01} and \cite{DiM07}), often outside the semimartingale class. Rough path theory (Lyons and subsequent developments; cf.\ \cite{Lyo98} and \cite{FH14}) furnishes a pathwise calculus for continuous such irregular signals, extending integration and differential equations beyond semimartingales and enabling models with low regularity and memory effects.

Any discussion on no-arbitrage hinges on (i) a sufficiently large set of admissible strategies and (ii) a valid stochastic integral to define the gain process. The classical paradigm builds the integral from simple, predictable strategies and a scaling limit, naturally tied to the semimartingale world. Since different interpretations of the gain process result in different meanings of the self-financing condition of a trading portfolio, and thus different interpretations of no-arbitrage, a basic question arises: can the gain process, defined via an alternative stochastic integral, support a frictionless no-arbitrage theory beyond semimartingales, and where are its limits?

Indeed, four structural features make the rough integral a natural and compelling alternative to It\^o’s integral in financial modelling:
\begin{itemize}
    \item \emph{Consistency with the It\^o calculus.}  Using the It\^o rough lift of a standard Brownian motion, rough integrals and rough differential equations recover the classical It\^o calculus (resp.\ the Stratonovich lift recovers the Stratonovich calculus) (cf.\ \cite{FH14} Chapter 3, 5 and 9).
    
    \item \emph{Stability.} The rough integral map $(\XXX,Y)\mapsto\int Y\,\mathrm{d}\XXX$ is jointly continuous with respect to rough path/controlled path metrics (cf. \cite{FH14} Section 4.4). In particular, $Y_n\rightarrow 0$ in the controlled paths metric implies $\int Y_n \mathrm{d}\XXX \rightarrow 0$ for any rough path driver $\XXX$. This mirrors the ``simple strategies \(\to\) limit" paradigm in the Bichteler–Dellacherie good integrator theory (cf. \cite{Bit81} and \cite{BSV11_Good_Integrator}) in a pathwise manner without assuming semimartingales.
    \item \emph{Model-free accounting.} Rough integral is defined purely pathwise, so self-financing and rebalancing become pathwise accounting identities. The choice of lift only fixes the deterministic compensation (the bracket and higher renormalisation terms), without invoking any probabilistic model a priori.
    \item \emph{A wide class of portfolios.} The integrands in rough integrals, the controlled paths, naturally accommodate state-dependent, as well as history-dependent (signature-type) strategies while keeping the gain functional linear in $Y$ and local in time. By the universal approximation of path signatures (cf. \cite{HBS24} and \cite{CPS25}), linear functionals of signatures approximate broad classes of non-anticipative, path-dependent strategies arbitrarily well in various settings, yielding a rich class of portfolios.
\end{itemize}

Taken together—and in continuity with classical stochastic calculus—these observations motivate the central question of this paper: \emph{Does Rough Path theory, and in particular
rough integration to define the gain process, yield a coherent, frictionless no-arbitrage framework beyond the It\^o-diffusion/semimartingale setting? If so, what are its structural limits?}

Another key feature of rough-path theory is that a given stochastic process admits infinitely many rough-path lifts, if there exists one, and different lifts yield different rough integrals. The non-uniqueness is already evident for Brownian motion: one can construct distinct It\^o and Stratonovich lifts so that the resulting rough integral coincides exactly with the classical It\^o or Stratonovich integral. In the regime where sample paths are a.s.\ $\alpha$–H\"older for some $\alpha\in(\tfrac13,\tfrac12)$ (which includes Brownian motion), the choice of the rough path lift corresponds one-to-one to a $C^{2\alpha}-$renormalisation function, which is usually referred to as the rough bracket in rough paths literature; for the It\^o lift this bracket coincides with the quadratic variation, making It\^o the natural choice to suppress spurious arbitrage.  As we will illustrate in Section \ref{sec:Classification}, by letting $\alpha$ be small enough, one can choose up to $\lfloor\frac{1}{\alpha}\rfloor-1$ different renormalisation functions in the rough path lift. We call them renormalisation terms since they will all appear in an It\^o-type change of variable formula for rough integrals. In other words, rough path theory provides a very rich class of bona fide notions of integrals against stochastic processes to define the gain process of a portfolio. This reframes the central problem of this paper into two precise questions:
\begin{center}
    \emph{What class of continuous stochastic processes can serve as the driving random rough path for an arbitrage-free market? For such a process, what must its specific rough path lift be?}
\end{center}

Two previous related works provide some hints. The first construction was made for fractional Brownian motions (fBm) with Hurst parameter $H\in(\frac{1}{3}, \frac{1}{2}]$ in \cite{QX18}, where an It\^o-type rough path lift was crafted such that state-dependent arbitrages are excluded if the gain process is interpreted as a rough integral. However, as pointed out therein, this is theoretically unsatisfactory since fBm itself is non-Markovian and admits memories for $H\neq \frac{1}{2}$; restricting strategies to be purely state-dependent leaves the most important question unanswered. Yet another relevant discussion was made in \cite{DJ23} for rough differential equations driven markets. Their settings contain proportional transaction cost, for which only portfolios of bounded variation were allowed to ensure a well-defined transaction cost process, so that the integral used therein was essentially still a Riemann-Stieltjes integral rather than a rough integral.

The main contribution of this paper is to provide an answer to the questions above. Roughly speaking, continuous frictionless rough-path-based market models with non-semimartingale drivers are ruled out in our setting.

\subsection{Contribution and outline}
We fix a finite time interval $[0,T]$ and consider the scenario of continuous trading with an early liquidation allowed, i.e., no other jumps except for the final cash out. Unlike in the FTAP, we understand no-arbitrage in this paper as a controlled version of the No Free Lunch (NFL) condition à la Kreps-Yan. It is well known (cf. \cite{DS06} Theorem 5.2.2) that when simple strategies (jumps) are allowed, the (NFL) is equivalent to the existence of an equivalent local martingale measure for the price process, which enforces it to be a semimartingale. Simple strategies can be built into our theory as piecewise controlled paths, but we shall not include them for the sake of continuous trading. In our continuous trading scenario, different classes of strategies (understood as holding processes and liquidation times) induce different mathematical interpretations of our No Controlled Free Lunch (NCFL) condition, which is defined in Definition \ref{def:NCFL} as:

\begin{definition*}[No controlled free lunch]
    Let $\mathcal{H}$ be a linear and convex set of random controlled paths containing 0 and let $T \in{\mathfrak{T}}_T$ be a set of stopping times bounded by the terminal time $T$. For $p\in (1,+\infty]$, a rough market model $\mathbf{S}$ (cf. Definition \ref{def:Rough Market Model}) is said to satisfy the condition of no controlled free lunch (NCFL) of order $p$ within $\mathcal{H}\times  {\mathfrak{T}}_T$ if the closure $\overline{C}^p(\mathcal{H},  {\mathfrak{T}}_T)$ of the gain cone $C^p(\mathcal{H},  {\mathfrak{T}}_T)$ induced by the trading strategies within $\mathcal{H}\times  {\mathfrak{T}}_T$, taken with respect to the weak-star topology on $L^p(\Omega, \mathcal{F}, \mathbb{P})$, satisfies
    \begin{equation}
        \overline{C}^p(\mathcal{H},  {\mathfrak{T}}_T)\cap L^p_+(\Omega, \mathcal{F}, \mathbb{P})=\{0\}.
    \end{equation}
\end{definition*}

We translate the (NCFL) market condition into a purely probabilistic property of the price integrator via the ``Rough Kreps-Yan Theorem" in Theorem \ref{RoughKYTheorem} as:

\begin{theorem*}[Rough Kreps-Yan Theorem]
   Let $p\in (1, +\infty],\;q\in [1,+\infty)$ with $\frac{1}{p}+ \frac{1}{q}=1$ and let $\mathcal{H}$ be a linear and convex subset of controlled portfolios containing $0$ and ${\mathfrak{T}}_T$ any set of stopping times bounded by the terminal time $T$ and containing $T$. A rough market $\mathbf{S}$ satisfies \textnormal{(NCFL)} of order $p$ within $\mathcal{H}\times {\mathfrak{T}}_T$, if and only if there exists an equivalent measure $\mathbb{Q}$ such that 
    \begin{equation}\label{eq:Rough Kreps-Yan in Intro}
        \frac{\mathrm{d}\mathbb{Q}}{\mathrm{d}\mathbb{P}}\in L^q_+ \quad \text{and} \quad \mathbb{E}_{\mathbb{Q}}\Big[\int_0^{\tau} Y_t \mathrm{d} \mathbf{S}_t\Big] =0 \quad\forall Y\in \mathcal{H},\;\forall \tau\in {\mathfrak{T}}_T.
    \end{equation}
\end{theorem*}

We then call the random rough path $\SSS$ an $\mathcal{H}\times {\mathfrak{T}}_T-$unbiased rough integrator under $\mathbb{Q}$, since integrating against $\SSS$ always yields a centred process. More concretely, we define in Definition \ref{def:Unbiased RI}:

\begin{definition*}[Unbiased rough integrator]
    Let $\XXX$ be an adapted random rough path on a filtered probability space $(\Omega, \mathcal{F}, (\mathcal{F}_t)_{t\in [0,T]}, \mathbb{P})$ and $\mathcal{H}$ be a class of adapted random piecewise a.s. by-$\XXX$-controlled paths. Moreover, let ${\mathfrak{T}}_T$ be any set of stopping times bounded by the terminal time $T$ and containing $T$. We say that $\XXX$ is an $\mathcal{H}\times {\mathfrak{T}}_T-$unbiased rough integrator, if 
    \begin{equation}
        \mathbb{E}\Big[\int_0^\tau Y_t\mathrm{d} \XXX_t\Big]=0\quad\quad \forall Y\in \mathcal{H},\; \forall \tau\in {\mathfrak{T}}_T.
    \end{equation}
    We also use the notation $\mathcal{H}-$unbiased rough integrator for ${\mathfrak{T}}_T =\{T\}$ for simplicity.
\end{definition*}

Our main questions then reduce to finding unbiased integrators with respect to different $\mathcal{H}$ and ${\mathfrak{T}}_T$. Since each coordinate of $\SSS$ describes the rough price process of one asset, we discuss only $1-$dimensional unbiased rough integrators without loss of generality. Our approach is rather fundamental and covers $\alpha-$H\"older continuous rough paths for $\alpha>0$ arbitrarily small. We first give a combinatorial classification of all $1-$dimensional rough paths via the complete Bell polynomials, and give a one-to-one correspondence between rough paths and the tuple of the underlying path with renormalisation functions in Proposition \ref{Rough Paths via Bell Polynomials}, with which we develop an It\^o type change of variable formula in Theorem \ref{Ito type formula} for every $1-$dimensional rough path.

We first consider rough noises, defined in Definition \ref{def:rough noise}, as the integrators to model the market noise in the price process without loss of generality. Roughly speaking, rough noises are the rough path lifts of centred stochastic processes with deterministic renormalisation functions, as in Proposition \ref{Rough Paths via Bell Polynomials}. Otherwise, we consider an adapted time change of the random rough paths using the adapted and monotone renormalisation terms (cf. Subsection \ref{subsec:Random renorm}). Then, for a fixed rough noise $\XXX$, define
\begin{equation}
        \mathcal{H}_{\XXX}^{\text{Pol}}:= \Big\{ \Big(P(X_t), DP(X_t),..., D^{(k-1)}P(X_t)\Big):  \; \forall P \;\text{polynomial}\Big\}.
    \end{equation}
as the set of polynomials-induced controlled path and ${\mathfrak{T}}_T:=\{T\}$. We then have the classification  of $\mathcal{H}^{\mathrm{Pol}}-$unbiased rough integrator in Theorem \ref{thm:Pol-unbiased rough integrators}:
\begin{theorem*}[$\mathcal{H}^{\text{Pol}}-$unbiased rough integrator]\label{thm:Pol-unbiased rough integrators in Intro}
    Let $\XXX$ be a rough noise, renormalised by bounded variation paths $G^2,...,G^k$ for some $k\geq 2$. Then, $\XXX$ is an $\mathcal{H}^{\mathrm{Pol}}-$unbiased rough integrator if and only if it is a Gaussian-Hermite rough path as in Definition \ref{Unbiased Rough Paths}.
\end{theorem*}

Roughly speaking, a Gaussian-Hermite rough path is such that the underlying process must have Gaussian marginals, and different levels of the rough path lift must be given by Hermite polynomials applied to the path increments and the variance increments. Now we enlarge the set of controlled paths to include minimal path-history-dependence (a single past state):
\begin{equation}
    \mathcal{H}^{\mathrm{pSig}}_{\XXX}:=\{\prescript{n}{s}{Y}: \; \forall s\in [0,T] \text{ and }\forall n\in \mathbb{N}\},
\end{equation}
where $\prescript{n}{s}{Y}:=(\prescript{n}{s}{Y}^{(1)},...,\prescript{n}{s}{Y}^{(k)})$ is the piecewise controlled path induced by the signature component path $\prescript{n}{s}{Y}^{(1)}_t: = \XX^n_{s,t}\cdot \mathbbm{1}_{t\geq s}$ as in Proposition \ref{prop:Sig as CP}. By allowing an early exit time by ${\mathfrak{T}}_T:=[0,T]$, we must further tighten the choice of unbiased rough integrator to the Hermite rough path lift (cf. Definition \ref{Unbiased Rough Paths}) of a Chen-Hermite almost Brownian motion (cf. Definition \ref{def:CHABM}). More concretely, let
\begin{equation}
    \mathcal{H}:=\text{Span}( \mathcal{H}^{\text{Pol}} \cup \mathcal{H}^{\text{pSig}}) \quad\quad \text{and} \quad\quad {\mathfrak{T}}_T:=[0,T],
\end{equation}
Theorem \ref{Thm:Main Theorem} states:
\begin{theorem*}[$\mathcal{H} \times {\mathfrak{T}}_T$-unbiased rough integrator]\label{Thm:Main Theorem in Intro}
    Let $\XXX$ be a rough noise, renormalised by bounded variation paths $G^2,...,G^k$. Then, $\XXX$ is an $\mathcal{H} \times {\mathfrak{T}}_T$-unbiased rough integrator if and only if it is the Hermite rough path lift of a deterministic time change of a Chen-Hermite almost Brownian motion.
\end{theorem*}

A Chen-Hermite almost Brownian motion (cf. Definition \ref{def:CHABM}) is a class of centered Gaussian marginal processes that have the same covariance structure as the standard Brownian motion, and we conjecture that they are indeed the same (cf. Conjecture \ref{conjecture}). However, the sharper collapse would indeed happen if we further enlarge the set of admissible continuous trading portfolios to encode more path-dependence. We define:
\begin{equation}
    \mathcal{H}^{\mathrm{AdpSig}}_{\XXX}:=\{\eta \cdot \prescript{n}{s}{Y}: \; \forall s\in [0,T], \; \forall \eta\in L^\infty(\mathcal{F}_s) \text{ and }\forall n\in \mathbb{N}\},
\end{equation}
and call such portfolios adaptedly scaled piecewise signature-induced portfolios. Consequently, we define
\begin{equation}
\mathcal{H}^\# :=\operatorname{Span}\!\big(\mathcal{H}^{\mathrm{Pol}}\cup \mathcal{H}^{\mathrm{AdpSig}}\big).
\end{equation}
We now state the main theorem of this paper in Theorem \ref{Thm:Main Sharp Theorem}:
\begin{theorem*}[$\mathcal{H}^\# \times {\mathfrak{T}}_T$-unbiased rough integrator]\label{Thm:Main Sahrp Theorem in Intro}
    Let $\XXX$ be a rough noise, renormalised by bounded variation paths $G^2,...,G^k$. Then, $\XXX$ is an $\mathcal{H}^\# \times {\mathfrak{T}}_T$-unbiased rough integrator if and only if it is the It\^o rough path lift of a deterministic time change of a standard Brownian motion.
\end{theorem*}

The implications of our theory for financial modelling with rough paths and rough integrals are straightforward (cf. Subsections \ref{subsec:RKYTheorem} and \ref{subsec:RDE models}). For rough-noise-induced Bachelier-type and RDE models, \((\mathrm{NCFL})\) within Markovian-type strategies forces Gaussian-Hermite rough paths under an equivalent change of measure. Passing to signature-type and then adaptedly scaled signature-type portfolios progressively narrows the admissible noises, ultimately forcing a time-changed Brownian motion together with its Itô rough path lift, up to an equivalent change of measure. Thus, frictionless viability collapses to the classical semimartingale paradigm once strategies are sufficiently rich. Notably, simple strategies do not occur in this mechanism.

\paragraph{Connections to literature}
Our work builds upon several distinct, yet interconnected, streams of literature. The challenge of building no-arbitrage markets beyond the semimartingale paradigm has been explored from various approaches, including market frictions and alternative integral constructions \cite{DHP00}, \cite{DK00}, \cite{EV03}, \cite{Che03}, \cite{BH05}, \cite{Guasoni}, \cite{BSV08}, \cite{QX18}, \cite{DJ23}. Rough–path ideas have increasingly entered finance and econometrics; representative applications include \cite{ABBC21}, \cite{ACLP23}, \cite{ALP24}, \cite{FT24}, \cite{BFGJ24}, \cite{BPS25}, \cite{BBFP25}, \cite{AI25}. In particular, the (RIE) property, which clarifies when a rough integral reduces to a classical limit of Riemann-Stieltjes sum, was discussed in \cite{PP16} and \cite{ALP24}. Non-exhausting applications of path signatures in financial modelling can be found in \cite{ASS20}, \cite{LNP20}, \cite{CGS23}, \cite{DT23}, \cite{BHRS23}, \cite{AGH24}, \cite{CM24}, \cite{CGMS25} and \cite{AG25}. On a technical level, our classification of one-dimensional rough paths is closely related to prior work on the structure of geometric and non-geometric rough paths, often described in the language of branched rough paths \cite{Kel12}, \cite{HK15}, \cite{BC19}.

\paragraph{Organisation.} The rest of this paper is organised as follows: In Section \ref{sec:Preliminaries}, we review the essentials of rough path theory and, in particular, rough integral against $\alpha-$H\"older rough paths with $\alpha>0$ arbitrarily small. In Section \ref{sec:Classification}, we study $1-$dimensional rough paths via a combinatorial classification, derive an It\^o type change of variable formula for all $1-$dimensional rough paths and construct the Hermite-Gaussian rough path lift for a $1-$dimensional stochastic process with Gaussian marginals. In Section \ref{sec:Unbiasedness}, we classify unbiased rough integrators for different sets of controlled paths. In Section \ref{sec:market}, we define formally rough market models, derive the ``Rough Kreps-Yan" Theorem, and derive the (NCFL) consequences of the classifications of unbiased rough integrators for RDE-driven models. Finally, we give an explicit arbitrage in geometric rough market models in Section \ref{sec:Arbitrage}. Conclusions and some outlook on future work are made in Section \ref{sec:Conclusion}. Appendix \ref{app:Proofs} contains some proofs, while Appendix \ref{app:Consistecy} covers some consistency results between controlled paths and rough paths to understand the gain process of an RDE model.


\section{Preliminaries on rough path theory}\label{sec:Preliminaries}

In this section, we give a review of the essentials of rough path theory. Subsection \ref{subsec:notation} contains the notations that we will use throughout this paper. Subsection \ref{subsec:RP} reviews rough paths, controlled paths and rough integrals. Subsection \ref{subsec:Examples of CP} lists examples of controlled paths. Finally, we explain the meaning of the solution to a rough differential equation (RDE) driven by an $\alpha-$H\"older rough path for $\alpha\in (\frac{1}{3}, \frac{1}{2}]$ in Subsection \ref{subsec:RDE}.
\subsection{Basic notations}\label{subsec:notation}
Let $(\mathbb{R}^d,|\cdot|)$ be standard Euclidean space and let $A_1 \otimes \cdots \otimes A_k$ denote the tensor product of vectors $A_1,..., A_k \in \mathbb{R}^d$, that is, the $d$-dimensional hypermatrix of order $k$ with $(i_1,..., i_k)$-component given by $[A_1 \otimes \cdots \otimes A_k]^{i_1... i_k}=A_1^{i_1}\cdot ... \cdot A_k^{i_k}$ for $1 \leq i_1,..., i_k \leq d$, where $A^i_j$ denotes the $i-$th component in the vector $A_j\in \mathbb{R}^d$ for $i=1,...,d$ and $j=1,...,k$. 

The tensor algebra $T(\mathbb{R}^d)$ and the truncated up-to-level-$k$ tensor algebra $T^{(k)}(\mathbb{R}^d)$ are defined as
$$
T(\mathbb{R}^d):= \bigoplus_{i=0}^\infty (\mathbb{R}^d)^{\otimes^{i}}\quad\quad \text{and}\quad\quad T^{(k)}(\mathbb{R}^d):= \bigoplus_{i=0}^k(\mathbb{R}^d)^{\otimes^{i}}.
$$
We equip $T^{(k)}(\mathbb{R}^d)$ with the Euclidean norm, but note that all norms on this space are equivalent since it is finite-dimensional.

The space of continuous paths $X:[0, T] \rightarrow \mathbb{R}^d$ is given by $C([0, T] , \mathbb{R}^d)$, and $\|X\|_{\infty,[0, T]}$ denotes the supremum norm of $X$ over the interval $[0, T]$. We use the notation $TV_{[s,t]}(X)$ for the total variation of the path $X$ on the subinterval $[s,t]$, and drop the interval in notation if it is clear in context. For the increment of a path $X:[0, T] \rightarrow \mathbb{R}^d$, we use the standard shorthand notation
$$
X_{s, t}:=X_t-X_s, \quad \text { for } \quad(s, t) \in \Delta_{[0, T]}:=\left\{(u, v) \in[0, T]^2: u \leq v\right\} .
$$

For any partition $\mathcal{P}=\left\{0=t_0<t_1<\cdots<t_N=T\right\}$ of an interval [$0, T$], we denote the mesh size of $\mathcal{P}$ by $|\mathcal{P}|:=\max \left\{\left|t_{k+1}-t_k\right|: k=0,1, \ldots, N-1\right\}$.

Let $W$ be an Euclidean space, we define for a two-parameter function $\XX: \Delta_{[0,T]}\rightarrow W$ and $\gamma>0$ its $\gamma$-H\"older coefficient as
$$
||\XX||_{\gamma}:=\sup_{(s,t)\in \Delta_{[0,T]}}\frac{|\XX_{s,t}|}{|t-s|^\gamma},
$$
and say it is $\gamma$-H\"older continuous if $||\XX||_\gamma <\infty$. Note that if $\XX$ represents the increments of a non-constant path, it cannot be $\gamma$-H\"older continuous for any $\gamma>1$. However, general two-parameter functions can have H\"older exponent strictly larger than 1. (E.g., $\XX_{s,t}:= (t-s)^2$). For $\alpha\in (0,1]$, we say a continuous path $X \in C([0,T], \mathbb{R}^d)$ is $\alpha$-H\"older continuous if its increments function is $\alpha$-H\"older continuous. We say $X$ has H\"older exponent $\alpha$ if it is $\alpha-$H\"older continuous but not $\beta-$H\"older continuous for any $\beta>\alpha$; and similarly, we say $X$ has H\"older exponent $\alpha^-$ if it is $\beta-$H\"older continuous for any $\beta<\alpha$ but not $\alpha-$H\"older continuous.

Let $(E,\|\cdot\|)$ be a normed space and let $f, g: E \rightarrow \mathbb{R}$ be two functions. We shall write $f \lesssim g$ or $f \leq C g$ to mean that there exists a constant $C>0$ such that $f(x) \leq C g(x)$ for all $x \in E$. Note that the value of such a constant may change from line to line, and that the constants may depend on the normed space, for example, through its dimension or regularity parameters.

Finally, given two vector spaces $U, V$, we write $\mathcal{L}(U , V)$ for the space of linear maps from $U$ to $V$. For $x\in \mathbb{R}$, we write $\lfloor x\rfloor: = \sup \{n\leq x: n\in \mathbb{Z} \}$.


\subsection{Rough paths, controlled paths and rough integral}\label{subsec:RP}
We now review $\alpha$-H\"older rough paths for $\alpha\in (0,1]$, (piecewise) controlled paths and the corresponding rough integrals. To the authors' best knowledge, the proofs of the results in this subsection in standard textbooks \cite{FV10} and \cite{FH14}, only cover the case when $\alpha\in (\frac{1}{3}, \frac{1}{2}]$. However, they can be canonically generalised to our settings, and we develop them in Appendix \ref{app:Proofs} for the convenience of the readers, only if necessary.

\begin{definition}[Rough paths]\label{Rough Paths}
     Let $\alpha\in (0, 1]$, we fix $k:=\lfloor \frac{1}{\alpha}\rfloor$. We say that $\mathbf{X}=$ $\left(1, \mathbb{X}^{1}, \ldots, \mathbb{X}^{k}\right): \Delta_T \rightarrow T^{(k)}(\mathbb{R}^d)$ is an $\alpha$-H\"older rough path in $\mathbb{R}^d$ if Chen's relation (\cite{Che57}) holds, i.e., for all $0\leq s\leq u\leq t\leq  T$ and all $i=0,...,k$ we have:
\begin{equation}\label{Chen's relation}
        \XX^{i}_{s,t}=\sum_{j=0}^{i}\XX^{j}_{s,u}\otimes \XX^{i-j}_{u,t}, 
\end{equation}
where we use the convention $\XX^0 = 1$ and for all $(s,t)\in \Delta_{[0,T]}$ and all $i=1,\cdots , k$ we have
\begin{equation}\label{Holder regularity}
\left\|\XX_{s, t}^{i}\right\| \lesssim(t-s)^{i\alpha },
\end{equation}
where $||\cdot||$ denotes the corresponding Euclidean norm.
\end{definition}

The object $\XXX$ is related to a path $X$ in the sense that $\XX^{1}$ represents exactly the increment of a path since when $i=1$, \eqref{Chen's relation} reads exactly $\XX^{1}$ is additive. In the following, we shall say that $\XXX$ is a rough path lift of a path $X$ if $\XX^{1}_{s,t} =X_{s,t}$ for all $0<s,t<T$.

For those who are familiar with rough paths in continuous semimartingales settings (e.g., Brownian motions, It\^o diffusions), note that if we set $\alpha\in (\frac{1}{3}, \frac{1}{2}]$, we recover the commonly-used definition for continuous rough paths in literature, where Chen's relation reads $\XX^{2}_{s,t} = \XX^{2}_{s,u}+ \XX^{2}_{u,t} + X_{s,u}\otimes X_{u,t}$ for all $0<s\leq u\leq t<T$. This is already sufficient for the analysis of most continuous semimartingales, since Brownian motions and It\^o diffusions with regular drift and volatility coefficients have almost surely H\"older exponent $\frac{1}{2}^-$. Moreover, they can be almost surely lifted to rough paths via It\^o integrals. (See e.g., \cite{FV10}). However, in this paper, we will develop our theory for arbitrary $\alpha \in (0,1]$, since the goal of this paper is to detect whether rough path theory can possibly go beyond It\^o calculus for the sake of financial modelling.

\begin{remark}
    One shall have the intuition that the terms $\XX^{i}_{s,t}$, also called postulated signatures of the path $X$, are designed to represent some postulated values for the (a priori ill-posed) iterated integrals 
    \begin{equation}
        ``\int_s^t \int_s^{u_1} \cdots\int_s^{u_i}\mathrm{d}X_{u_{i-1}}\cdots\mathrm{d}X_{u_2} \mathrm{d}X_{u_1}"  \text{ for }i=1,...k
    \end{equation} up to the level $k=\lfloor\frac{1}{\alpha}\rfloor$. Indeed, one can check that if $X$ is of bounded variation (or has H\"older exponent strictly larger than $\frac{1}{2}$, resp.), then the iterated Riemann-Stieltjes integrals (H\"older integrals, resp.) yield a well-defined $\alpha$-H\"older rough path for any $\alpha\in(0,1]$. 
\end{remark}

While \eqref{Chen's relation} does capture the most basic (additivity) property that one expects any decent theory of integration to respect, it does not imply any form of integration by parts/chain rule. Indeed, one can also postulate such a condition on the underlying rough path as below.

For notational convenience, we think of components of some fixed rough path increment $\mathbf{X}_{s, t} \in T^{(k)}(\mathbb{R}^d)$ as being indexed by words $w$ of length at most $k$ with letters in the alphabet $\{1, \ldots, d\}$, i.e., $w=w_1\cdots w_n$ with $n\leq k$ and $w_1,...,w_k\in \{1,...,d\}$. Given a word $w=w_1 \cdots w_n$, the corresponding component $\mathbf{X}^w$, which we also write as $\langle\mathbf{X}, w\rangle$, is then interpreted as the postulated $n$-fold integral
\begin{equation}
    \left\langle\mathbf{X}_{s, t}, w\right\rangle:= \XX^{n; w_1...w_n }_{s,t}
\end{equation}
where $\XX^{n}_{s,t}\in (\mathbb{R}^d)^{\otimes n}$ is viewed as a hypermatrix and $\XX^{n; w_1...w_n }_{s,t}$ denotes the $(w_1,...,w_n)-$component of the hypermatrix $\XX^{n}_{s,t}$. We also use the convention 
\begin{equation}
    \langle \XXX_{s,t}, \emptyset \rangle:=\XX_{s,t}^0=1.
\end{equation}

In order to describe the constraints imposed on these iterated integrals by the chain rule, we define the shuffle product $\shuffle$ between two words as the formal sum over all possible ways of interleaving them. For example, one has

\begin{equation}
    a \shuffle x=a x+x a, \quad a b \shuffle x y=a b x y+a x b y+x a b y+a x y b+x a y b+x y a b,
\end{equation}
with the empty word acting as the neutral element. With this notation at hand, we define:

\begin{definition}[Geometric rough paths]\label{Geometric Rough Paths}
    Let $\XXX$ be as in Definition \ref{Rough Paths}. $\XXX$ is called geometric if 
    \begin{equation}\label{Shuffle Product}
        \left\langle\mathbf{X}_{s, t}, w\right\rangle\cdot \left\langle\mathbf{X}_{s, t}, v\right\rangle = \left\langle\mathbf{X}_{s, t}, w\shuffle v\right\rangle.
    \end{equation}
    for all $0<s,t<T$ and all words $w,v$ such that the sum of their lengths is no larger than $k$.
\end{definition}

\begin{remark}\label{rem:GRP}
    It was already remarked in \cite{Ree58} that chain rules would imply \eqref{Shuffle Product} for iterated Riemann-Stieltjes integrals of paths of bounded variation. See also \cite{FH14}, Section 2.2 for a simpler version when $k=2$, where the \eqref{Shuffle Product} becomes essentially $\int_s^t X^i_{u}\mathrm{d} X^j_u + \int_s^t X^j_{u}\mathrm{d} X^i_u = X^i_t\cdot X^j_t - X^i_s\cdot X^j_s$ for every $i,j=1,...,d$. This illustrates the integration by parts/chain rule nature of the constraints \eqref{Shuffle Product}.
\end{remark}

\begin{remark}
    We mention that what we call geometric rough paths is more often referred to as weakly geometric rough paths in rough path literature. There is only a very minor difference between these two notions, but one barely distinguishes them in practice. See \cite{FH14} Section 2.2 for a discussion on this topic. For brevity, we drop the adjective ``weakly” and simply write ``geometric”.
\end{remark}

Let us review some classical rough paths before proceeding further. 
\begin{example}[Classical examples]
Assume $X\in C^{\alpha}([0, T], \mathbb{R})$ is a 1-dimensional $\alpha-$H\"older continuous path. One can check that $\XX^{i}_{s,t}:= \frac{1}{i!}\cdot X_{s,t}^i$ yields a well-defined and unique geometric rough path lift of $X$. As we shall classify in Section \ref{sec:Classification}, every 1-dimensional non-geometric rough path (cf. Definition \ref{Rough Paths}) can be understood in the form of the geometric lift with some higher order perturbation functions.

Moving to the $d-$dimensional case for $d>1$, things get trickier since generally speaking, one does not have natural choices for the values to postulate for the cross integral, e.g., $\XX^{2;31}$, $\XX^{4; 1212}$, etc. However, we regain the possibility of random rough path lifts once we enter a stochastic setting. For semimartingales, It\^o (or Stratonovich resp.) integrals lift the sample paths to well-defined non-geometric (or geometric resp.) rough paths almost surely. Gaussian processes with sufficiently regular covariance functions admit a canonical geometric rough paths lift (cf. \cite{FH14} Chapter 10), including fractional Brownian motions (fBm) with Hurst parameter $H>\frac{1}{4}$ (cf. \cite{CQ02}), while geometric rough path lift can be constructed in non-canonical ways for fBm with $H\leq\frac{1}{4}$ (cf. \cite{Unt10}, \cite{NT11}). Finally, continuous Markov processes with regular generators can be lifted to rough paths via a construction involving Dirichlet forms (cf. \cite{FV10}, Chapter 16).

Nevertheless, we make the crucial remark that given any path or stochastic process, the rough path lift is never unique as long as there exists one. We shall explain this and classify the 1-dimensional case in Section \ref{sec:Classification} and \ref{sec:Unbiasedness}. One of our objectives in this paper is to study which is the (or whether there exists a) proper rough path lift given a price process (possibly beyond semimartingales), for the sake of certain no-arbitrage conditions.
\end{example}

    The spirit of rough path theory is that if one can fake those (a priori ill-posed) iterated integrals via postulated values, then one can practically use them to develop pathwise analysis that iterated integrals, if well-defined, can develop. This includes defining integrals against irregular paths, which cannot be done via classical pathwise theory (e.g. Riemann-Stiltjes integral, Young's integral). Now we may review the notion of rough integrals.

Given a rough path $\XXX$, one can define the integral against it for certain classes of integrands, i.e., the controlled paths. Roughly speaking, the controlled paths are paths whose future infinitesimal increments locally depend on the reference path $X$ via a Taylor-like expansion with $X$ and certain derivative-like objects.

\begin{definition}[Controlled paths]\label{Controlled Paths}
    Let $\alpha \in (0,1]$, $k:= \lfloor\frac{1}{\alpha}\rfloor$ and $\XXX$ an $\alpha-$H\"older rough path in $\mathbb{R}^d$ as above. Let $W$ be another Euclidean space. We define a $W$-valued controlled path $Y:=(Y^{(1)}, Y^{(2)},..., Y^{(k)})$ as a ${\alpha}-$H\"older continuous path taking values in the space $\bigoplus_{i=0}^{k-1} \mathcal{L}((\mathbb{R}^{d})^{\otimes^i}, W)$ such that for all $s,t\in [0,T]:$
    \begin{equation}\label{eq:Controlled Path}
        |Y^{(i)}_{s,t} - \sum_{j=i+1}^k Y^{(j)}_s \XX^{j-i}_{s,t}|\lesssim |t-s|^{(k+1-i)\cdot \alpha} \quad\quad \forall i=1,...,k-1, 
    \end{equation}
    where we used the canonical isomorphism $\mathcal{L}((\mathbb{R}^d)^{\otimes^{j-1}}, W)\otimes (\mathbb{R}^d)^{\otimes^{(j-i)}}\cong \mathcal{L}((\mathbb{R}^d)^{\otimes^{i-1}}, W)$ via $f\otimes w \mapsto (v \mapsto f(v\otimes w))$ in writing $Y^{(j)}_s \XX^{j-i}_{s,t}$. We also say that $Y$ is controlled by $\XXX$. 
    
    We call $Y:=(Y^{(1)}, Y^{(2)},..., Y^{(k)})$ a piecewise controlled path, if there exists a finite partition $0=t_0<t_1<\cdots<t_n=T$ such that the restriction of $Y$ on each subinterval $[t_i, t_{i+1}]$ is controlled by the restriction of $\XXX$ on the same subinterval and that $Y$ is continuous.
\end{definition}

\begin{remark}
We make the point that in defining piecewise controlled paths, we are intentionally not allowing jumps. Otherwise, simple integrands will be allowed. Since the main goal of this paper is to study whether the framework of rough paths extends some kind of no-arbitrage condition out of the classical semimartingale setting, we do not wish to consider simple integrands in the rest of this paper, since they immediately drive us back to the semimartingale setting by the classical Kreps-Yan theorem. 
\end{remark}

We shall give some classical examples of controlled paths in Subsection \ref{subsec:Examples of CP} below. We are now ready to integrate a controlled path against the reference rough path $\XXX$ as a limit of compensated Riemann-Stieltjes sum.

\begin{proposition}[Rough integral]\label{Rough Integral}
    Let $\alpha\in (0,1]$. Assume $\XXX$ is an $\alpha$-H\"older rough path in $\mathbb{R}^d$ and $Y$ is an associated $\mathcal{L}(\mathbb{R}^d, V)$-valued controlled path as above for some Euclidean space $V$. Then
    \begin{equation}\label{Compensated Riemann Sum}
        \int_s^t Y_u \mathrm{d}\XXX_u := \lim_{|\mathcal{P}|\downarrow 0}\sum_{[t_n, t_{n+1}]\in \mathcal{P}}\sum_{i=1}^k Y^{(i)}_{s,t_n} \XX^{i}_{t_n, t_{n+1}}     
    \end{equation}
    is well-defined for any $0\leq s, t\leq T$, where we use the canonical isomorphism $\mathcal{L}((\mathbb{R}^d)^{\otimes^{i-1}}, \mathcal{L}(\mathbb{R}^d, V))$ $\cong \mathcal{L}((\mathbb{R}^d)^{\otimes^i}, V)$ in writing $Y^{(i)}_{t_n} \XX^{i}_{t_n, t_{n+1}}$. The $k-$tuple $(\int_0^\cdot Y_u\mathrm{d}\XXX_u, Y^{(1)},..., Y^{(k-1)})$ is again a $V$-valued controlled path.
    
    The above construction then extends to the rough integral of a piecewise controlled path $Y$ against the rough path $\XXX$ in a unique additive way. The $k-$tuple $(\int_0^\cdot Y_u\mathrm{d}\XXX_u, Y^{(1)},..., Y^{(k-1)})$ is again a $V$-valued piecewise controlled path with respect to the same partition.
\end{proposition}
\begin{proof}
    The extension from integrals of controlled paths to integrals of piecewise controlled paths is obvious. The proof of the existence of the former relies on Gubinelli's seminal sewing lemma \cite{Gub04} and is essentially a generalisation of the proof in the standard case where $\alpha\in (\frac{1}{3}, \frac{1}{2}]$, see e.g. \cite{FH14} Theorem 4.10. For completeness, we present a proof in Appendix \ref{App:Rough Integral}
\end{proof}

\begin{remark}
    We note that a more widely used framework for defining (non-geometric) rough paths of lower regularity ($\alpha< \frac{1}{3}$) is the so-called branched rough paths (cf. \cite{Gub10}), where one needs to postulate values to rooted-tree-indexed integrals, such as
    $\int_s^t (X_{s,u}\int_s^u X_{s,v}dX_v)dX_u.$ Then one can use a tree-indexed expansion to define controlled paths and then a corresponding limit of rooted-tree-indexed compensated Riemann sum as the rough integral. An advantage of this approach is that applying a smooth enough function $F$ to the underlying path would always yield a controlled path, a property that the non-geometric rough paths in our settings do not automatically admit without any additional constraints, as will be illustrated in Example \ref{Markovian Integrand}.

    We choose not to adopt this framework for two primary reasons. First, our analysis is self-contained within the tensor algebra setting for all regularities $\alpha>0$ that we consider, and already covers the classical settings of It\^o/Stratonovich calculus for standard Brownian motion. Second, our core arguments in Sections \ref{sec:Classification} and \ref{sec:Unbiasedness} depend on a novel connection between the seminal Marcinkiewicz's theorem (\cite{Mar39}) and our novel combinatorial classification of one-dimensional rough paths with exponential polynomials, which leads to a particularly tractable Itô-type formula. Developing similar results in the context of branched rough paths will be of much heavier algebraic notations. Whether the same results in this paper hold for branched rough paths lies beyond the scope of this paper.
    
\end{remark}

\subsection{Examples of controlled paths}\label{subsec:Examples of CP}
Now, let us briefly review some classical classes of controlled paths, which will be later used to understand trading portfolios in rough market models. Throughout this subsection, we fix an $\alpha-$H\"older continuous rough path $\XXX$ and let $k:=\lfloor\frac{1}{\alpha} \rfloor.$
\begin{example}[Markovian-type controlled paths]
    Assume first $\XXX$ is geometric, or non-geometric but $\alpha>\frac{1}{3}$. It is well known (cf. \cite{FH14} Lemma 4.1 and Section 7.6) that $F(t, X_t)$ is a controlled path for $F\in C^{1,k+1}$. Moreover, suppose $(A_t)_{t\in[0, T]}$ is an additional Lipschitz-continuous information path in some Euclidean space $\mathbb{R}^{d_A}$, e.g., the running average $t \mapsto \int_0^t S_u^i \mathrm{~d} u$ for some (or all) $i=1, \ldots, d$. Similar arguments yield that $F(A_t, S_t)$ is a controlled path for $F\in C^{1,k+1}$. See \cite{ALP24} Section 3.2 for a proof for the case of $\alpha\in (\frac{1}{3}, \frac{1}{2}]$.
They are thus well-defined as controlled paths. We call them Markovian-type controlled paths, due to their nature that the current state depends only on the current information, but not on the path history.
For general non-geometric $\XXX$ with $\alpha<\frac{1}{3}$, some additional regularity requirements for the non-geometric renormalisation terms encoded in the rough path $\XXX$ are needed to ensure the well-definedness of Markovian-type controlled paths. The regularity requirements are rather loose. We delay a detailed discussion, including proof, until Section \ref{sec:Classification}. (Cf. Example \ref{Function of Geometric Rough Paths as Controlled Paths}, Example \ref{Markovian Integrand} and Remark \ref{rem:Function of Info and Path as CP}.)
\end{example}

Markovian-type controlled paths are restrictive to describe trading portfolios, in particular, if we want to allow price processes beyond semimartingales (say, with memories) in the rough path settings. We introduce two ways to include path-dependence in controlled paths.

\begin{example}[Extended signature components as controlled paths]\label{exp:Sig as CP}
    It is well-known (cf. \cite{LV07}) that every rough path uniquely extends to a signature, which is an object in $T^{(\infty)}(\mathbb{R}^{d+1})$ that plays the role of the set of iterated integrals of the paths at every level. Indeed, such constructions can be made via rough integrals, and as a consequence, every component in the extended signature is a well-defined controlled path. We make this rigorous via the following proposition.
\end{example}

\begin{proposition}[Signature as controlled paths]\label{prop:Sig as CP}
    Let $\XXX:=(1,X,...,\XX^k)$ be an $\alpha-$H\"older continuous rough path with $k:=\lfloor \frac{1}{\alpha}\rfloor$ and $w=w_1\cdots w_n$ be a word with letters in the alphabet $\{0,...,d\}$. If $n\leq k$, the following is a controlled path:
    \begin{equation}\label{eq:Sig as CP}
        \Big(\langle \XXX_{s,\cdot}, w \rangle, \langle \XXX, w^{-1} \rangle \otimes e_{w_n}, \langle \XXX, w^{-2} \rangle \otimes e_{w_{n-1}}\otimes e_{w_n},\cdots, \; \langle \XXX, \emptyset \rangle \otimes e_{w_1}\otimes\cdots \otimes e_{w_n},0,...,0\Big)
    \end{equation}
    for any $s\in[0,T]$, where $w^{-i}:=w_1\cdots w_{n-i}$ for $i=1,...,n$ and $\{e_0,...,e_d\}$ is the canonical basis of $\mathbb{R}^{d+1}\cong \mathcal{L}(\mathbb{R}^{d+1}, \mathbb{R})$. Inductively, for $n>k$, we may extend $\XXX$ via
    \begin{equation}\label{eq:Sig Portfolio int}
        \langle \XXX_{s,t}, w \rangle: = \int _s^t \Big(\langle \XXX_{s,u}, w^{-1} \rangle \otimes e_{w_n},\cdots,\langle \XXX_{s,u}, w^{-k} \rangle \otimes e_{w_{n-k+1}}\otimes\cdots \otimes e_{w_n}\Big)\mathrm{d} \XXX_u
    \end{equation}
    where the integrand is a controlled path for all words of finite length with letters in the alphabet $\{0,...,d\}$. We call them the signature components of the rough paths $\XXX$.
\end{proposition}
\begin{proof}
     See Appendix \ref{App:Sig as CP}.
\end{proof}

\begin{example}[Path-dependent functionals as controlled paths]\label{exp:Functional as CP}
    In general contexts of financial modelling, another common way to treat path-dependent and adapted trading strategies is via the so-called non-anticipative path functionals. We refer readers to \cite{Dup19} and \cite{CF10} for details. The connection between non-anticipative path functionals and controlled paths is that for $\alpha\in (\frac{1}{3}, \frac{1}{2}]$, one can show that for a sufficiently regular non-anticipative path functional, applying it to any $\alpha-$H\"older continuous rough path yields a controlled path in a slightly weaker sense. One can still define a rough integral against it via the sewing lemma (cf. \ref{Sewing Lemma}). See \cite{Ana19} Section 4.2 and \cite{ALP24} Section 3.3 for more discussions on this topic.
\end{example}


\subsection{Rough differential equations (RDEs)}\label{subsec:RDE}
Yet another fruitful development in rough path theory is the so-called rough differential equation (RDE), which provides robust pathwise methodologies for describing differential dynamics driven by irregular signals. Notably, the theory of RDE recovers the It\^o/Stratonovich SDEs by setting the underlying rough path as the It\^o/Stratonovich of a Brownian motion.

We recall the definition of the solution to an RDE driven by a $\alpha-$H\"older continuous rough path for $\alpha\in (\frac{1}{3}, \frac{1}{2}]$. We refer to \cite{FH14} Chapter 8 for more details. For simplicity, we use the notation $(Y,Y')$ for controlled paths in this subsection, where $Y'$ is usually referred to as the Gubinelli derivative in literature.

\begin{definition}[Solution to RDE]\label{def:RDE}
    Let $\XXX:=(X, \XX)$ be an $\alpha-$H\"older continuous rough path and $f$ a $C^2-$function. A controlled path $(Y, Y')$ is a solution to the RDE
\begin{equation}\label{eq:RDE}
    \mathrm{d} Y_t=f\left( Y_t\right) \mathrm{d} \XXX _t, \quad Y_0=y_0
\end{equation}
if for all $t \in[0, T]$ we have that $Y$ satisfies the integral equality
\begin{equation}\label{1-18}
    Y_t=y_0+\int_0^t (f\left(Y\right)_s,  Df\left(Y_s\right)Y'_s )\mathrm{d} \XXX_s
\end{equation}
where it is straightforward to verify that $(f\left(Y\right)_s,  Df\left(Y_s\right)Y'_s )$ is a controlled path if so is $(Y,Y')$.
\end{definition}

We refer to \cite{Lyo98} and \cite{FH14} Chapter 8 for the well-posedness of RDEs. For our purpose, we assume that there always exists a unique solution to the RDE \eqref{eq:RDE} in the rest of this paper. As we shall point out in Subsection \ref{subsec:RDE models}, it will be essential to be able to integrate against the solutions to RDEs (as controlled paths) again to understand the gain process of RDE market models. This is solved via some consistency results between controlled paths and rough paths; we list them in Appendix \ref{app:Consistecy}. We mention that the same results therein cannot be repeated for $\alpha-$H\"older continuous non-geometric rough paths for $\alpha\leq \frac{1}{3}$. To restore similar results, one needs to introduce the notion of branched rough paths (cf. \cite{Gub10}).


\section{One-Dimensional Rough Paths}\label{sec:Classification}

In this section, we study the structure of one-dimensional rough paths. We first give a combinatorial classification of all 1-dimensional rough paths in Subsection \ref{subsec:classification}. Roughly speaking, every 1-dimensional rough path can be characterised by the underlying path and some renormalisation terms with a graded structure. With the renormalisation terms, we then derive an It\^o-type formula for every 1-dimensional rough path in Subsection \ref{subsec:Ito}, which will be crucial in Section \ref{sec:Unbiasedness}, for classifying unbiased rough integrators for different sets of controlled paths as integrands. Via the Hermite polynomials, we then construct the Gaussian-Hermite rough paths for $1-$dimensional processes with Gaussian marginals. The Gaussian-Hermite rough paths admit an initial unbiased rough integrator property for the set of sufficiently regular functions of the underlying Gaussian process, which mimics the classical Skorokhod integral in a purely pathwise manner. We work in the general setting of $\alpha\in (0,1].$

 
\subsection{Classification of one-dimensional rough paths}\label{subsec:classification}
It is well known that given a 1-dimensional $\alpha-$H\"older path $X$ on $[0, T]$, its only geometric rough path lift is given by polynomials of the path's increments
\begin{equation}
    \XX^i_{s,t} : = \frac{1}{i!}\cdot(X_t - X_s)^i,\quad\quad\forall i=1,...,\lfloor\frac{1}{\alpha}\rfloor,
\end{equation}
whereas all other 1-dimensional rough paths are non-geometric. For $\alpha\in (\frac{1}{3}, \frac{1}{2}]$, the non-geometric part is uniquely determined by the so-called rough bracket (cf. \cite{FH14} Exercise 2.11 and Definition 5.5), a $C^{2\alpha}$ function, which plays a similar renormalising role as the quadratic variation for semimartingales. We now extend such a type of fact for $\alpha>0$ arbitrarily small, employing the complete Bell polynomials (cf. \cite{Bel34}), which can be viewed as a generalisation of Hermite polynomials.

We first recall the definition and the explicit form of the ordinary complete Bell polynomials.
\begin{definition}[Complete Bell polynomials; \cite{Bel34}]\label{BellPoly}
    For any $k,n\in \mathbb{N}_0$, the $n-$th ordinary complete Bell polynomial of degree $k$, $P^{(k)}_n$, is defined as the abstract polynomial coefficients, of $k$ indeterminates, of the following Taylor expansion around $x=0$:

\begin{equation}
    \exp(\sum^k_{m=1}a_m\cdot x^m) =: \sum_{n=0}^\infty P^{(k)}_n(a_1,...,a_k)\cdot x^n,     
\end{equation}
in the sense that the above equation holds for every $a_1,...a_k\in \mathbb{R}.$
\end{definition}

One checks immediately by setting $k=2$, $a_1 = u$, an indeterminate, and $a_2 = -\frac{1}{2}$, the probabilist's Hermite polynomials $\mathrm{He}_n$ are recovered via:
\begin{equation}
    \mathrm{He}_n(u) = n!\cdot P^{(2)}_n(u, -\frac{1}{2}) \quad \forall n\in \mathbb{N}_0.
\end{equation}
Moreover, the first few ordinary complete Bell polynomials are given by 

\begin{equation}\label{eq:First Bell polynomials}
    \begin{aligned}
        & P^{(k)}_0 = 1, \quad \forall k\geq 0;\\
        & P^{(k)}_1 = a_1, \quad \forall k\geq 1;\\
        & P^{(k)}_2 = \frac{a_1^2}{2}+ a_2, \quad \forall k\geq 2;\\
        & P^{(k)}_3 = \frac{a_1^3}{6}+ a_1a_2+ a_3, \quad \forall k\geq 3;\\
        &\quad\quad\quad\quad...
    \end{aligned}
\end{equation}

More explicitly, the ordinary complete Bell polynomials have the following formula:

\begin{lemma}\label{lem:Bell}
    For all $k,n\in \mathbb{N}_0$, we have:
\begin{equation}\label{Bell Polynomial Formula}
    P^{(k)}_n (a_1,...a_k) = \sum_{\substack{p_1,..., p_k\geq0 \\ p_1+2p_2+\cdots+k\cdot p_k = n}}\prod_{m=1}^k\frac{a_m^{p_m}}{p_m !}.
\end{equation}
\end{lemma}
\begin{proof}
    See Appendix \ref{App:Bell}.
\end{proof}

We are now able to detail the classification. Fix $\alpha\in (0,1]$.

\begin{proposition}\label{Rough Paths via Bell Polynomials}
    Fix $k: = \lfloor \frac{1}{\alpha}\rfloor$. Let $X\in C^{\alpha}([0, T], \mathbb{R})$ and $G^2, ..., G^k$ be functions such that $G^i \in C^{i\alpha}([0, T], \mathbb{R})$ $\forall i =2,...,k.$ (In particular, when $\alpha>\frac{1}{2},$ we have $k=1$ and no $G^i$ appears.) Then $\XXX : [0,T]^2 \rightarrow \mathbb{R}$ defined by:
    \begin{equation}\label{eq:prop4.1-1}
        \XXX_{s,t} := (\XX^{0}_{s,t}, \XX^{1}_{s,t},..., \XX^{k}_{s,t})
    \end{equation}
    where
    \begin{equation}\label{eq:prop4.1-2}
       \XX^{i}_{s,t} : = P^{(k)}_i(X_{s,t},\; G^2_{s,t},\;..., \;G^k_{s,t})\quad\quad \forall i=0,...,k
    \end{equation}
    is a well-defined $\alpha-$H\"older rough path. Conversely, for any $\alpha-$H\"older rough path 
    \begin{equation}\label{eq:prop4.1-3}
        \mathbf{Y}=(1, Y, \mathbb{Y}^2..., \mathbb{Y}^k),
    \end{equation}
    there exists functions $F^2,... F^k$ such that $F^i \in C^{i\alpha}([0, T], \mathbb{R})$ $\forall i =2,...,k$ and 
     \begin{equation}\label{eq:prop4.1-4}
       \mathbb{Y}^{i}_{s,t} : = P^{(k)}_i(Y_{s,t},\; F^2_{s,t},\;..., \;F^k_{s,t})\quad\quad \forall 0\leq s, t\leq T\quad \text{and}\quad\forall i=0,...,k.
    \end{equation}
    In the following, we shall call $F^2,...,F^k$ the renormalisation terms encoded in the rough path $\mathbf{Y}$.
\end{proposition}
\begin{proof}
    We let $\XXX$ be as given in \eqref{eq:prop4.1-1} and \eqref{eq:prop4.1-2}, and use the convention $G^1=X\in C^{\alpha}$. The graded H\"older regularity condition in \eqref{Holder regularity} follows from the RHS of \eqref{Bell Polynomial Formula} that $p_1+2p_2+...+kp_k = n$ and that $G^{i\alpha}\in C^{i\alpha}$. Moreover, notice that
    \begin{equation}
        \exp{(\sum_{m=1}^k G^m_{s,t}\cdot x^m)} = \exp{(\sum_{m=1}^k G^m_{s,u}\cdot x^m)}\cdot \exp{(\sum_{m=1}^k G^m_{u,t}\cdot x^m)}.
    \end{equation}
    By comparing the coefficients in the Taylor expansions around 0 of the above equality, one obtains:
    \begin{equation}\label{eq:Bell to Chen}
        P^{(k)}_i (G^1_{s,t}, ...,G^k_{s,t})= \sum_{j=0}^i P^{(k)}_j (G^1_{s,u}, ...,G^k_{s,u})\cdot P^{(k)}_{i-j} (G^1_{u,t}, ...,G^k_{u,t})\quad\quad \forall 0\leq s,t \leq T
    \end{equation} 
    which recovers the desired Chen's relation in \eqref{Chen's relation}.

    Now, conversely, let $\mathbf{Y}$ be an $\alpha-$H\"older rough path. We construct the desired functions $F^i$ for $i=2,...,k$ via induction on $i$. That \eqref{eq:prop4.1-4} is satisfied for $i=0,1$ follows from the definition of $\mathbf{Y}$ in \eqref{eq:prop4.1-3}. We now assume there exists $F^i\in C^{i\alpha}([0,T], \mathbb{R})$ satisfying \eqref{eq:prop4.1-4} for all $i< j$ for some $j = 2,...,k$ and aim to construct $F^j\in C^{j\alpha}([0,T], \mathbb{R})$ satisfying \eqref{eq:prop4.1-4} as well. We note from \eqref{Bell Polynomial Formula} that the polynomial $P^{(k)}_j$ is essentially a polynomial of $j$ indeterminates, and thus we may indeed make this assumption without knowing $F^j,..., F^k$. We define the polynomials
    \begin{equation}
        \tilde{P}^{(k)}_j (a_1,..., a_k) := P^{(k)}_j(a_1,...,a_k) - a_j,
    \end{equation}
    and then set
    \begin{equation}\label{eq:Recover perturbation}
        F^j_{s,t}:=\mathbb{Y}^{j}_{s,t} -  \tilde{P}^{(k)}_j (Y_{s,t}, F^1_{s,t},..., F^{j-1}_{s,t},0,...,0) \quad\quad \forall 0\leq s,t\leq T.
    \end{equation}
    Again because the polynomial $P^{(k)}_j$ depends only on the first $j$ indeterminates, it is then obvious that \eqref{eq:prop4.1-4} is satisfied at level $j$ with the constructed $F^2,..., F^j$ and then any choices of $F^{j+1},...,F^k$. Moreover, by the expression of $\tilde P^{(k)}_j$ in Lemma \ref{lem:Bell} and the induction hypothesis $F^i\in C^{i\alpha}$ for $i<j$, we have
\begin{equation}
    \big\lvert \tilde P^{(k)}_j\big(Y_{s,t},F^1_{s,t},\dots,F^{j-1}_{s,t},0,\dots,0\big) \big\rvert
\lesssim |t-s|^{j\alpha} \quad\quad\forall 0\leq s,t\leq T.
\end{equation}
Hence, together with the rough path regularity assumption that $\mathbb{Y}^j_{s,t} \lesssim |t-s|^{j\alpha}$, we have
\begin{equation}
    |F^j_{s,t}|\lesssim |t-s|^{j\alpha}\quad\quad\forall 0\leq s,t\leq T .
\end{equation}
We now only need to show the additivity of $F^j$, i.e., $F^i_{s,t} = F^i_{s,u} + F^i_{u,t} $ for all $0\leq s\leq u\leq t\leq T.$ Similar to \eqref{eq:Bell to Chen}, we have
    \begin{equation}
         P^{(k)}_j (a_1+b_1,..., a_j+b_j, 0,...,0)= \sum_{l=0}^j P^{(k)}_l (a_1,...,a_j,0,...,0)\cdot P^{(k)}_{j-l} (b_1,...,b_j,0,...,0),
    \end{equation}
    for all $a_i,b_i\in \mathbb{R}$ for $i=1,...,j$. Again thanks to Lemma \ref{lem:Bell}, by deleting the terms involving $a_j$ and $b_j$, we obtain
    \begin{equation}\label{eq:Prop4.1-finally}
    \begin{aligned}
        \tilde{P}^{(k)}_j (a_1+b_1,..., a_{j-1}+b_{j-1}, 0,...,0)&= \sum_{l=1}^{j-1} P^{(k)}_l (a_1,...,a_{j-1},0,...,0)\cdot P^{(k)}_{j-l} (b_1,...,b_{j-1},0,...,0)\\
        & \quad + \tilde{P}^{(k)}_j (a_1,..., a_{j-1}, 0,...,0)+ \tilde{P}^{(k)}_j (b_1,..., b_{j-1}, 0,...,0),
    \end{aligned}
    \end{equation}
    for all $a_i,b_i\in \mathbb{R}$ for $i=1,...,j$. Plugging in \eqref{eq:Recover perturbation}, $a_1 = Y_{s,u}$, $b_1 = Y_{u,t}$ and $a_i = F^i_{s,u}$, $b_i = F^i_{u,t}$ for $i=2,...,j-1$ into \eqref{eq:Prop4.1-finally}, we obtain:
    \begin{equation}
        \mathbb{Y}^j_{s,t} - F^j_{s,t} = \Big\{\sum_{l=1}^{j-1} \mathbb{Y}^l_{s,u}\cdot \mathbb{Y}^{j-l}_{u,t}\Big\} + \mathbb{Y}^j_{s,u} - F^j_{s,u} + \mathbb{Y}^j_{u,t} - F^j_{u,t}.
    \end{equation}
    Now the additivity of $F^j$ is obtained by modulo the Chen's relation in \eqref{Chen's relation} for $\mathbb{Y}^j$.
\end{proof}

\begin{remark}
This is not the first paper trying to classify non-geometric rough paths. Such classifications are typically achieved by embedding d-dimensional non-geometric rough paths into geometric ones in an extended state space, but often within the framework of branched rough paths (see, e.g., \cite{Kel12}, \cite{HK15}, \cite{BC19}). Our construction, however, is distinct. It is purely combinatorial, not in the branched framework, and provides a direct isomorphism between the space of non-geometric rough paths and a corresponding set of renormalisation functions.
\end{remark}

\begin{remark}
    One cannot hope to extend the result above to $d-$dimensional non-geometric rough paths naively, since one cannot hope to describe the non-trivial crossed integrals via polynomials of the integrands and integrators. Also, unlike the $1$-dimensional polynomial interpretations, adding a renormalisation function $G^2$ onto $\XX^2$ for an $\mathbb{R}^d-$valued rough path $\XXX=(1, X, \XX^2,...,\XX^k)$ does not guarantee the well-posedness of the modification of the crossed terms in $\XX^n$ for $n\geq 3$; a necessary and sufficient condition for the well-posedness is that $(X, G^2)$ admits a geometric rough path lift. In other words, the framework of embedding non-geometric rough paths into geometric ones in \cite{HK15} become necessary to classify $d-$dimensional ones. However, we mention that it shall be natural to extend similar results to classify $d-$dimensional reduced rough paths for $\alpha\in (0,1]$ (This was defined in \cite{FH14} Definition 5.3 for the scenario of $\alpha\in (\frac{1}{3},\frac{1}{2}]$, but shall generalise naturally to $\alpha\in (0,1]$). 
\end{remark}


\subsection{An It\^o-type formula}\label{subsec:Ito}
We fix an $\alpha-$H\"older continuous path $X\in C^{\alpha}([0, T], \mathbb{R})$ and consider $\XXX$, a rough path lift associated with $X$. Moreover, let $G^2,...G^k$ be the renormalisation terms as stated in Proposition \ref{Rough Paths via Bell Polynomials}. An advantage of making explicit the renormalisation terms is that rough integrals against any rough path then admit a clean It\^o-type change of variable formula. To address this, we first discuss the integrability of $f(t, X_t)$, i.e., the controlledness as in Definition \ref{Controlled Paths}. In short, the controlledness depends on the regularity of $f$, as well as the regularities of the perturbations $G^2,..., G^k$ encoded in the rough path. 

We begin with the geometric case, i.e., $G^i=0$ for all $i=2,...,k.$
\begin{example}[Geometric case]\label{Function of Geometric Rough Paths as Controlled Paths}
    Let $\XXX$ be the geometric rough path associated with $X$, i.e.,
    \begin{equation}
        \XX^{i}_{s,t} = \frac{1}{i!} \cdot(X_t - X_s)^i \quad \quad \forall i = 1,...k;\; \;\forall s,t\in [0,T].
    \end{equation}
    Moreover, let $f\in C^{1,k}([0,T]\times \mathbb{R}, \mathbb{R})$ and $Y^{(i)}_t:=D^{(i-1)}_xf(t, X_t)$ for $i = 1,... k$, where $D^{(i-1)}_xf$ denotes the $(i-1)-$th derivative of $f$ in $x$. Then, we have:
\begin{equation}
    \begin{aligned}
        Y^{(i)}_{s,t} - \sum_{j=i+1}^k Y^{(i)}_s\cdot \XX_{s,t}^{j-1} & = D_x^{(i-1)}f(t, X_t) - D_x^{(i-1)}f(s, X_t)+D_x^{(i-1)}f(s, X_t) \\
        &\quad\quad\quad\quad\quad- \sum_{j=i+1}^k D_x^{(j-1)}f(s, X_s)\cdot \frac{(X_t-X_s)^{j-i}}{(j-i)!}\\
        & = D_tD_x^{(i-1)}f(u, X_t)\cdot (t-s) \\
        &\quad\quad\quad\quad\quad+ \frac{1}{(k+1-i)!}\cdot D_x^{(k)}f(s, X_v)\cdot (X_t-X_s)^{(k+1-i)\alpha}\\
        & \lesssim |t-s|+ |t-s|^{(k+1-i)\alpha}\lesssim  |t-s|^{(k+1-i)\alpha} \quad\quad \forall i=1,...,k
    \end{aligned}
\end{equation}
where the second equality is by the Lagrange mean value theorem and $u,v \in[s,t]$, whereas the inequalities are due to the regularity of $X$ and the fact that $(k+1-i)\alpha<1$. In other words, $Y:=(Y^{(1)},... Y^{(k)})$ is indeed controlled by $\XXX$.
\end{example}

In the geometric case, one does not need to postulate any regularity requirement on the non-geometric parts $G^2,... G^k$. This is indeed due to that geometric rough paths have a trivial non-geometric part, which admits any desired regularity condition. However, regularity requirements are necessary in the general non-geometric case to ensure controlledness.

\begin{example}[Non-geometric case]\label{Markovian Integrand}
   Let $k := \lfloor\frac{1}{\alpha}\rfloor$. Assume $G^2,... G^k\in C^{k\alpha}([0, T], \mathbb{R})$. Let $\XXX$ be a rough path associated with $X$, constructed via the ordinary complete Bell's polynomials and $G^2,... G^k$ as in Proposition \ref{Rough Paths via Bell Polynomials}. Moreover, let $f\in C^{1,k}([0,T]\times \mathbb{R}, \mathbb{R})$ and $Y^{(i)}_t:=D^{(i-1)}_xf(t, X_t)$ for $i = 1,... k$, where $D^{(i-1)}_xf$ denotes the $(i-1)-$th derivative of $f$ in $x$. Similar to above, using the convention that $G^1 = X$ we have
\begin{equation}
    \begin{aligned}
        Y^{(i)}_{s,t} - \sum_{j=i+1}^k Y^{(i)}_s\cdot \XX_{s,t}^{j-1} & = D_x^{(i-1)}f(t, X_t) - D_x^{(i-1)}f(s, X_t)+D_x^{(i-1)}f(s, X_t) \\&\quad\quad\quad\quad - \sum_{j=i+1}^k D_x^{(j-1)}f(s, X_s)\cdot \frac{(X_t-X_s)^{j-i}}{(j-i)!}\\
        &\quad\quad\quad\quad- \sum_{j=i+1}^k D_x^{(j-1)}f(s, X_s)\cdot\sum_{\substack{p_1,...p_k\geq0 \\ p_1+2p_2+...+k\cdot p_k = j-i \\ p_1\neq j-i}}\prod_{m=1}^k\frac{(G^m_t-G^m_s)^{p_m}}{p_m!}\\
        & = D_tD_x^{(i-1)}f(u, X_t)\cdot (t-s) \\
        &\quad\quad\quad\quad+ \frac{1}{(k+1-i)!}\cdot D_x^{(k)}f(s, X_v)\cdot (X_t-X_s)^{(k+1-i)\alpha}\\
        &\quad\quad\quad\quad- \sum_{j=i+1}^k D_x^{(j-1)}f(s, X_s)\cdot\sum_{\substack{p_1,...p_k\geq0 \\ p_1+2p_2+...+k\cdot p_k = j-i \\ p_1\neq j-i}}\prod_{m=1}^k\frac{(G^m_t-G^m_s)^{p_m}}{p_m!}\\
        & \lesssim |t-s|+ |t-s|^{(k+1-i)\alpha} + |t-s|^{k\alpha} \lesssim  |t-s|^{(k+1-i)\alpha}  \quad\quad \forall i=1,...,k
    \end{aligned}
\end{equation}
where the second equality is by the Lagrange mean value theorem and $u,v \in[s,t]$, whereas the inequalities are due to the regularity of $X$ and $G$ and the fact that $(k+1-i)\alpha<1$ and $k+1-i\leq k$. In other words, $Y:=(Y^{(1)},... Y^{(k)})$ is indeed controlled by $\XXX$ if the non-geometric terms in $\XXX$ is extra regular.
\end{example}

\begin{remark}\label{rem:Function of Info and Path as CP}
    With almost the same proof, the above controlledness results can be extended to include some regular information paths. This will practically allow a wider class of portfolios, as will be further discussed in Subsubsection \ref{exp:Markovian Portfolios}. Let $A:[0,T]\rightarrow \mathbb{R}^{d_A}$ be another Lipschitz-continuous path, e.g., the running average $t \mapsto \int_0^t S_u^i \mathrm{~d} u$ for some (or all) $i=1, \ldots, d$. Moreover, let $f\in C^{1,k}(\mathbb{R}^{d_A}\times \mathbb{R}, \mathbb{R})$ and $Y^{(i)}_t:=D^{(i-1)}_xf(A_t, X_t)$ for $i = 1,... k$, where $D^{(i-1)}_xf$ denotes the $(i-1)-$th derivative of $f$ in $x$. Then $Y:=(Y^{(1)},... Y^{(k)})$ is indeed controlled by $\XXX$ if the perturbation terms in $\XXX$ are as regular as in Example \ref{Markovian Integrand}. We leave the verification to readers or see a proof for $\alpha\in (\frac{1}{3}, \frac{1}{2}]$ in \cite{ALP24} Section 3.2 (but in the language of $\frac{1}{\alpha}-$variation). This provides a larger set of controlled portfolios but will not be of further use in this paper.
\end{remark}

With the integrability of the Markovian-type of integrands $f(t, X_t)$, we are now able to formulate the rough It\^o formula for 1-dimensional non-geometric rough paths.

\begin{theorem}[Rough It\^o formula]\label{Ito type formula}
    Let $\alpha>0$ and $k=\lfloor\frac{1}{\alpha}\rfloor$. Let $\XXX: = (1,\XX^1,..., \XX^k)$ be a 1-dimensional $\alpha-$H\"older rough path, and $G^1,... G^k$ its renormalisation terms as in Proposition \ref{Rough Paths via Bell Polynomials} where $G^i\in C^{k\alpha}$ for all $i=2,...,k$ and $F\in C^{1, k+1}([0,T]\times \mathbb{R}, \mathbb{R})$. Assume moreover that $(D_xF(t, X_t),...,D_x^{(k)}F(t, X_t))$ is controlled by $\XXX$ (cf. the above examples). Then we have:

    \begin{equation}\label{Ito Formula}    
    \begin{aligned}
        F(t, X_t) - F(s,X_s) = \int_s^t D_u F(u,X_u)\mathrm{d}u &+ \int_s^t (D_xF(u, X_u),...,D_x^{(k)}F(u, X_u)) \mathrm{d}\XXX_u\\
        &- \sum_{i=2}^k \int_s^t D_x^{(i)} F(u, X_u)\mathrm{d} G^i_u,
    \end{aligned}
    \end{equation}
    where the middle integral is a rough integral, and the others are Young's integral.
\end{theorem}

\begin{proof}
    For $\forall v< w$, we set the local approximation of the rough integral
    \begin{equation}
        \Xi_{v,w} := \sum_{i=1}^k D^{(i)}_xF(v, X_v)\cdot \XX^{i}_{v, w} \quad \quad \forall s,t \in [0,T].  
    \end{equation}
    as in \eqref{Compensated Riemann Sum}. Using the convention $G^1 = X$, one computes by Taylor's theorem:
    \begin{equation}
        \begin{aligned}
            F(w, X_w) - F(v,X_v) &= D_uF(v, X_v)\cdot(w-v) + \sum_{i=1}^k \frac{1}{i!}\cdot D^{(i)}_xF(v, X_v)\cdot(X_w-X_v)^i \\  & \quad\quad+ O(|w-v|^{(k+1)\alpha})\\
            & =  D_uF(v, X_v)\cdot(w-v) + \Xi_{v,w} - \sum_{i=2}^k D_x^{(i)}F(v, X_v)\cdot (G^i_w - G^i_v) \\
    &\quad\quad- \sum_{i=2}^k D_x^{(i)}F(v, X_v)\cdot \sum_{\substack{p_1,...p_k\geq0 \\ p_1+2p_2+...+k\cdot p_k = i \\ p_1\neq i,\; p_i\neq 1}}\prod_{m=1}^k
    \frac{(G^m_t - G^m_s)^{p_m}}{p_m!} \\
    &\quad\quad+ O(|w-v|^{(k+1)\alpha})\\
    & =  D_uF(v, X_v)\cdot(w-v) + \Xi_{v,w} \\ &\quad\quad- \sum_{i=2}^k D_x^{(i)}F(v, X_v)\cdot (G^i_w - G^i_v) + O(|w-v|^{(k+1)\alpha})
\end{aligned}
    \end{equation}
    where the last equality is because the regularity of $G$'s and $X$ and that in the sum at least two $p_j$'s are not zero, and we have $p_1\neq i$. Now taking $\lim_{|\mathcal{P}|\downarrow 0}\sum_{[v,w]\in \mathcal{P}}$ on both sides of the equation, we recover the desired equation \eqref{Ito Formula} where the last term is a well-defined Young's integral.
\end{proof}

As a direct consequence by setting $G^i=0$ for $i=2,...,k$, we recover the classical change of variable formula for geometric rough paths in the 1-dimensional case:

\begin{corollary}[Change of variable formula for geometric rough paths]
    Let $\alpha>0$ and $k=\lfloor\frac{1}{\alpha}\rfloor$. Let $\XXX: = (1,\XX^1,..., \XX^k)$ be a geometric 1-dimensional $\alpha-$H\"older rough path, and $F\in C^{1, k+1}([0,T]\times \mathbb{R}, \mathbb{R})$. Then we have:

    \begin{equation}\label{Ito Chain Rule}    
        F(t, X_t) - F(s,X_s) = \int_s^t D_u(u,X_u)\mathrm{d}u + \int_s^t (D_xF(u, X_u),...,D_x^{(k)}F(u, X_u)) \mathrm{d}\XXX_u
    \end{equation}
    where the first is Young's integral, and the second integral is a rough integral.
\end{corollary}

\begin{remark}[d-dimensional generalisation]\label{rem:Change of variable}
    We note that the change of variable formula also holds for $d$-dimensional geometric rough path by replacing the partial derivatives with partial gradients. We will cite a general version in Lemma \ref{lemm:Change of variable for GRP} to construct an explicit arbitrage in a geometric rough market thereafter.
\end{remark}

\subsection{The Gaussian-Hermite rough paths}\label{subsec:The unbiased Gaussian rough integrator}
In this subsection, we use the above construction to construct a canonical rough path lift for stochastic processes
\(X=\{X_t\}_{t\in[0,T]}\) whose marginals are Gaussian and whose sample paths possess arbitrarily
low H\"older regularity \(\alpha>0\). The constructed random rough path is then an unbiased rough integrator with respect to the class of certain Markovian-type integrands as in the previous subsection. Throughout this subsection, we make the following assumptions on the underlying process $X$:

\begin{assumption}\label{Assumption}
     $(X_t)_{t\in [0, T]}$ is a one-dimensional centred stochastic process with Gaussian marginals such that $X_0=0$; there exists $\alpha\in (0,1]$ such that the sample paths of $(X_t)_{t\in [0, T]}$ are almost surely $\alpha-$H\"older continuous; the function $t\mapsto \mathbb{E}[X_t^2]$ is $k\alpha$-H\"older continuous where $k:=\lfloor\frac{1}{\alpha}\rfloor$ and of bounded variation.
\end{assumption}

\begin{remark}\label{rem:scope}
We make some comments on the assumption. The $\alpha-$H\"older continuity is needed to define a rough path lift. As for the regularity assumption on the variance map, this is not demanding since taking the expectation is usually a smoothing operator. Indeed, it is easy to check that standard Brownian motion, fractional Brownian motion for any Hurst index $H\in(0,1)$, (fractional) Ornstein-Uhlenbeck processes (\cite{CKM03}), Gaussian moving averages with H\"older kernels, Riemann-Liouville processes, Volterra Gaussian processes with smooth kernels, generalised fractional Brownian motions (\cite{IPT22}, \cite{IPT25}) and many other time-homogeneous or time-changed Gaussian models all satisfy the above assumptions for a suitable $\alpha$. Hence, Assumption~\ref{Assumption} encompasses the vast majority of one-dimensional Gaussian models encountered in stochastic analysis and
mathematical finance while providing just enough structure for the
subsequent rough-path constructions.
\end{remark}

Now we are ready to define the Gaussian-Hermite rough paths. The construction is essentially a simplified version of Proposition \ref{Rough Paths via Bell Polynomials} and the well-definedness is proven therein.
\begin{definition}\label{Unbiased Rough Paths}
    Let $X$ be a stochastic process satisfying Assumption \ref{Assumption} with suitable $\alpha$, $k$ therein. We define polynomials:
    \begin{equation}
        H_n(a_1, a_2):= \sum_{\substack{p_1, p_2\geq0 \\ p_1+2p_2=n}}\frac{a_1^{p_1}\cdot a_2^{p_2}}{p_1!\cdot p_2!}= P^{(k)}_n(a_1,a_2,0,...,0)  \quad\quad\text{for} \; n=0,...,k,
    \end{equation}
    where $P^{(k)}_n$ is the $n-$th ordinary complete Bell Hermite polynomial of degree $k$ as in Definition \ref{BellPoly}. Moreover, we set $G(t):=-\frac{1}{2}\mathbb{E}[X_t^2]$ and then define the Hermite rough path associated with $X$ as:
    \begin{equation}
    \begin{aligned}
        \XXX_{s,t}(\omega):&= (\XX^0_{s,t}(\omega), \XX^1_{s,t}(\omega), ..., \XX^k_{s,t}(\omega))\quad\quad\forall s,t\in [0,T] ;\\
        \XX^i_{s,t}(\omega):& = H_i(X_{s,t}(\omega), G_{s,t})\quad\quad\forall s,t\in [0,T] \;\text{and} \; \text{for } i=0,1,...,k.
    \end{aligned}
    \end{equation}
\end{definition}

We give an immediate example with the fractional Brownian motions.
\begin{example}[Hermite fractional Brownian rough paths]
   Let $(B^H_t)_{t\in [0, T]}$ be a $1-$dimensional fBm with Hurst parameter $H\in (\frac{1}{3}, 1)$ and $G(t):=-\mathbb{E}[(B^H_t)^2]=- t^{2H}$.  We define:
    \begin{equation}
        \mathbb{B}^H_{s,t}:=H_2(B^H_{s,t}, G_{s,t}) = \frac{1}{2} (B^H_t - B^H_s)^2 - \frac{1}{2}(t^{2H}-s^{2H})\quad\quad \forall s,t\in \Delta_{[0,T]}.
    \end{equation}
    Then $\mathbf{B}^H:=(B^H, \mathbb{B}^H)$ is a well-defined $\min\{H, \frac{1}{2}\}-$H\"older continuous rough path lift of $B^H$. For $H=\frac{1}{2},$ the It\^o Brownian rough path is recovered. For $H\in (\frac{1}{3}, \frac{1}{2})$, we recover the construction in \cite{QX18} in the context of arbitrage-free pricing of Markovian-type contingent claims. A similar construction was made in \cite{Fer21} for multidimensional fBm with $H\in (\frac{1}{4},\frac{1}{3}]$ in the language of branched rough paths.
\end{example}

A key property of the Gaussian-Hermite rough path is that the rough integral of Markovian-type controlled paths against it again has zero mean. More concretely, we define:
\begin{definition}[Unbiased rough integrator]\label{def:Unbiased RI}
    Let $\XXX$ be an adapted random rough path on a filtered probability space $(\Omega, \mathcal{F}, (\mathcal{F}_t)_{t\in [0,T]}, \mathbb{P})$ and $\mathcal{H}$ be a class of adapted random piecewise a.s. by-$\XXX$-controlled paths. Moreover, let ${\mathfrak{T}}_T$ be any set of stopping times bounded by the terminal time $T$ and containing $T$. We say that $\XXX$ is an $\mathcal{H}\times {\mathfrak{T}}_T-$unbiased rough integrator, if 
    \begin{equation}
        \mathbb{E}\Big[\int_0^\tau Y_t\mathrm{d} \XXX_t\Big]=0\quad\quad \forall Y\in \mathcal{H},\; \forall \tau\in {\mathfrak{T}}_T.
    \end{equation}
    We also use the notation $\mathcal{H}-$unbiased rough integrator for ${\mathfrak{T}}_T =\{T\}$ for simplicity.
\end{definition}

Note that for each level $n\le k$, the polynomials $H_n$ are selected so that $\XX^n_{s,t}$ coincides with the
$n$-th probabilist's Hermite polynomial evaluated at the centred random variable along with some variance term $(X_{s,t},-\frac{1}{2}(\mathbb E[X_{t}^2]-\mathbb E[X_{s}^2]))$. This choice guarantees by setting $s=0$ and thus $X_s=0$:
\begin{equation}
\mathbb E\bigl[\XX^n_{0,t}\bigr] \;=\;0,
\qquad n=1,\dots,k,\quad t\in [0, T],
\end{equation}
due to the orthogonality between Gaussian distributions and the Hermite polynomials. This already corresponds to a very preliminary version of unbiasedness of the rough integrator, since from Proposition \ref{prop:Sig as CP} and its proof, we have the identity
\begin{equation}
    \XX_{0,t}^n = \int^t_0 (\XX^{(n-1)}_{0,s},..., \XX_{0,s}^0, 0,...,0)\mathbb{\XXX}_s.
\end{equation}
Moreover, due to the fact that the Hermite polynomials are orthogonal with respect to Gaussian measures, it should be expected that the above constructed rough paths, as rough integrators, should be unbiased also with respect to any Markovian-type integrands, as in example \ref{Markovian Integrand}. Indeed, the general unbiasedness for the rough integral against a Gaussian-Hermite rough path can be formulated as:

\begin{theorem}
    Let $X$ be a stochastic process satisfying Assumption \ref{Assumption} with suitable $\alpha$, $k$ therein, and $\XXX$ its Hermite rough path lift as in Definition \ref{Unbiased Rough Paths}. For $F\in C_b^{1, k+1}([0, T]\times \mathbb{R}, \mathbb{R})$, we have:
    \begin{equation}\label{eq:Unbiased Gaussian}
        \mathbb{E}\Big[\int_s^t (D_xF(u, X_u), D_x^{(2)}F(u, X_u),..., D_x^{(k)}F(u, X_u))\mathrm{d}\XXX_u\Big]=0 \quad\quad\forall 0\leq s\leq t\leq T.
    \end{equation}
    In particular, $\XXX$ is an $\mathcal{H}_X^{\text{Mar};b}-$unbiased rough integrator for 
    \begin{equation}\label{eq:Theorem 4.2 eq1}
        \mathcal{H}_X^{\text{Mar};b}:= \Big\{ (D_xF(t,X_t))_{t\in [0,T]}: F\in C_b^{1, k+1}([0, T]\times \mathbb{R}, \mathbb{R})\Big\},
    \end{equation}
    where $(D_xF(t,X_t))_{t\in [0,T]}$ is understood as a controlled path via its space derivatives as the integrand in \eqref{eq:Unbiased Gaussian}.
\end{theorem}
\begin{proof}
 By the Ito-type change of variable formula in Theorem \ref{Ito type formula}, it suffices to show
     \begin{equation}\label{(30)}
        \mathbb{E}[F(t, X_t) - F(s,X_s)]=\mathbb{E}\Big[\int_s^t D_uF(u, X_u)\mathrm{d}u \Big] + \mathbb{E}\Big[\int_s^t \frac{1}{2}\cdot D_{xx}F(u, X_u)\mathrm{d}\mathbb{E}[X_u^2]\Big]
    \end{equation}
    Let $\varphi$ be the density function of a standard normal distribution. We now compute the LHS of \eqref{(30)} with the short-hand notation $f(t):=\mathbb{E}[X^2_t]^{\frac{1}{2}}$:
    \begin{equation}\label{(31)}
        \begin{aligned}
            \mathbb{E}[F(t, X_t) - F(s,X_s)] &= \int_{\mathbb{R}} F(t,x)\cdot \frac{1}{f(t)}\cdot \varphi\Big(\frac{x}{f(t)}\Big)\mathrm{d}x - \int_{\mathbb{R}} F(s,x)\cdot \frac{1}{f(s)}\cdot \varphi \Big(\frac{x}{f(s)}\Big)\mathrm{d}x\\ & = \int_{\mathbb{R}} (F(t, f(t) x) - F(s, f(s) x))\cdot \varphi(x)\mathrm{d}x\\ & = \int_{\mathbb{R}}\int_s^t D_uF(u, f(u)x) \cdot \varphi(x)\mathrm{d}u\mathrm{d}x \\
            &\quad\quad\quad\quad+ \int_{\mathbb{R}}\int_s^t x\cdot D_xF(u, f(u)x) \cdot \varphi(x)\mathrm{d}f(u)\mathrm{d}x
        \end{aligned}
    \end{equation}
    where the first equation is due to $X_s\sim \mathcal{N}(0, f(s)^2)$; the second equation is a change of variable, whereas the third equation is due to the chain rule for Riemann-Stieltjes integral. It is obvious that the first term in the last line of \eqref{(31)} coincides with the first term of the RHS of \eqref{(30)}. Now it suffices to deal with the second term in the last line of \eqref{(31)}. Indeed,
    \begin{equation}\label{eq:Theorem 4.2 eq3}
        \begin{aligned}
            \int_{\mathbb{R}}\int_s^t x\cdot D_xF(u, f(u)x) \cdot \varphi(x)\mathrm{d}f(u)\mathrm{d}x & = \int_s^t \int_{\mathbb{R}}  D_xF(u, f(u)x) \cdot x\varphi(x)\mathrm{d}x \mathrm{d}f(u)\\ & = \int_s^t \int_{\mathbb{R}} f(u)\cdot D_{xx}F(u, f(u)x) \cdot \varphi(x)\mathrm{d}x \mathrm{d}f(u)\\
            & =  \int_{\mathbb{R}}  \int_s^tf(u)\cdot D_{xx}F(u, f(u)x) \cdot \mathrm{d}f(u)\varphi(x)\mathrm{d}x  \\
            & =\mathbb{E}\Big[\int_s^t \frac{1}{2}\cdot D_{xx}F(u, X_u)\mathrm{d}\mathbb{E}[X_u^2]\Big]
        \end{aligned}
    \end{equation}
    where the first equation is due to the Fubini theorem; the second is by an integration by parts of the inner integral and the fact that $x\varphi(x) = -\varphi'(x)$; the third is by Fubini theorem for a Lebesgue–Stieltjes integral, and the last equation is by the chain rule of the Riemann-Stieltjes integral. The desired equation \eqref{(30)} is thus proven.
\end{proof}

The same proof can be carried over to polynomial-induced (unbounded) controlled paths due to the integrability of polynomial growth functions against the Gaussian density:

\begin{theorem}\label{thm:Unbiasedness for polynomials}
    Let $X$ be a stochastic process satisfying Assumption \ref{Assumption} with suitable $\alpha$, $k$ therein, and $\XXX$ its Hermite rough path lift as in Definition \ref{Unbiased Rough Paths}. For any polynomial $P$, we have:
    \begin{equation}\label{eq:Unbiased Gaussian for Poly}
        \mathbb{E}\Big[\int_s^t (P(X_u), DP(X_u),..., D^{(k-1)}P(X_u))\mathrm{d}\XXX_u\Big]=0 \quad\quad\forall 0\leq s\leq t\leq T.
    \end{equation}
    In other words, $\XXX$ is an $\mathcal{H}_X^{\text{Pol}}-$unbiased rough integrator for
    \begin{equation}
        \mathcal{H}_X^{\text{Pol}}:= \Big\{ (P(X_t), DP(X_t),..., D^{(k-1)}P(X_t)):  \; \forall P \;\mathrm{polynomial}\Big\}.
    \end{equation}
\end{theorem}

\begin{remark}\label{rem:Unbiased GRP not for Sig}
    We make the point that the above unbiasedness properties are true really only for Markovian-type integrands. Indeed, assume the easiest non-Markovian (i.e., path-dependent) controlled path $X_{s,\cdot}$ for some $s>0$, defined by the reference path itself. Then we have by Proposition \ref{prop:Sig as CP}:
    \begin{equation}
        \mathbb{E}\Big[\int_s^t (X_{s,u}, 1,0,...,0) \mathrm{d}\XXX_u\Big] =\mathbb{E}[\XX^2_{s,t}]= \frac{1}{2}(\mathbb{E}[(X_t-X_s)^2] -\mathbb{E}[X_t^2] + \mathbb{E}[X_s^2]).
    \end{equation}
    It is straightforward to check that it is $0$ for all $s,t\in [0,T]$ if and only if $X$ has uncorrelated increments. If we additionally require that the underlying process $X$ is Gaussian, it means that $X$ has to be a martingale, i.e., a deterministic time-change of the standard Brownian motion. We shall make more precise such arguments in Section \ref{sec:Unbiasedness}.
\end{remark}


\section{One-dimensional \texorpdfstring{$\mathcal{H}\times \mathfrak{T}_T$}{H×T\_T}–Unbiased Rough Integrators}\label{sec:Unbiasedness}



We just constructed unbiased rough integrators (cf. Definition \ref{def:Unbiased RI}) with respect to Markovian-type controlled paths for 1-dimensional processes with Gaussian marginals via the Hermite polynomials. In this section, we aim to answer the inverse question: Given a 1-dimensional stochastic process, does there exist a rough path renormalisation such that it is an unbiased integrator for a large set of integrands $\mathcal{H}$? If so, what should the rough path perturbation be? As will be illustrated in Theorem \ref{RoughKYTheorem}, the answer to these questions leads to the answer to what rough markets can admit the No Controlled Free Lunch property for different classes of portfolios.

Again, we fix a filtered probability space $(\Omega, \mathcal{F}, (\mathcal{F}_t)_{t\in [0, T]}, \mathbb{P})$ where $\mathcal{F}= \mathcal{F}_T$. Assume moreover that $(\mathcal{F}_t)_{t\in [0, T]}$ satisfies the usual conditions of right-continuity and completeness. We consider the following class of random rough paths as potential noises in the price process.

\begin{definition}[1-dimensional rough noises]\label{def:rough noise}
    Let $\alpha\in(0,1]$, $k:=\lfloor \frac{1}{\alpha} \rfloor$ and $X$ a 1-dimensional centered stochastic process with $X_0=0$ almost surely. We call $\XXX:=(1, X,\XX^2,...,\XX^k)$ an $\alpha-$H\"older rough noise if the following conditions are satisfied:
    \begin{itemize}
        \item (Adapted RP Lift) $\XXX$ is an $(\mathcal{F}_t)_{t\in [0, T]}-$adapted rough path lift of the path $X$ (cf. Definition \ref{def:Rough Market Model});
        \item (Deterministic renormalisation) The renormalisation terms $G^2(\omega),... G^k(\omega)$ encoded in the rough path $\XXX(\omega)$, as in Proposition \ref{Rough Paths via Bell Polynomials}, are deterministic.
    \end{itemize}
    We call $\XXX$ renormalised by bounded variation paths if $G^2,... G^k$ are of bounded variation.
\end{definition}

\begin{remark}[Why deterministic renormalisations]\label{rem:whydetrenorm}
    It is reasonable to question the constraints of deterministic renormalisation terms, since general continuous martingales do not have deterministic quadratic variation. However, in the spirit of Dambis–Dubins–Schwarz theorem, continuous (local) martingales are essentially adapted time changes of a standard Brownian motion, where the time change is characterised by the quadratic variation clock. As will be illustrated in Subsection \ref{subsec:Random renorm}, similar clock construction via the renormalisation terms can be made for random rough paths with monotone adapted renormalisation terms. Thus, the assumption of deterministic renormalisations is essentially considered without loss of generality.
\end{remark}

Before moving forward to the classification of unbiased rough integrators with respect to different classes of integrands, we briefly explain why it suffices to consider the $1-$dimensional cases.
\begin{remark}[Why $d{=}1$ suffices]\label{rem:reduce-1d}

Let $\XXX:=(1,X, \XX^2,...,\XX^k)$ be a rough path in $\mathbb{R}^d$. For any $1$-dimensional projection $\nu\in (\mathbb{R}^d)^*$, it is then straightforward to check that 
\begin{equation}
    \XXX^{\nu}:=\Big(1, \langle X, \nu \rangle, \langle\XX^2, \nu\otimes\nu \rangle, ...,  \langle\XX^k, \nu^{\otimes^k} \rangle \Big)
\end{equation}
yields a rough path lift of the one-dimensional path projection $\langle X, \nu \rangle$. Consequently, if $\XXX$ is an unbiased rough integrator with respect to a sufficiently large class of integrands, then each of its $1-$dimensional projections has to be a $1-$dimensional rough integrator with respect to the projection of those integrands. Hence, our $d{=}1$ results give the necessary constraints for every linear view of the noise.

On the other hand, as will be revealed in Subsection \ref{subsec:HSig}, as the class of integrands turns sufficiently rich, the only admissible $1$-dimensional unbiased rough noise integrator must be, if not the same as, extremely close to the It\^o rough path lift of a standard Brownian motion, up to a deterministic time change. In the multivariate It\^o case, correlations only change the covariance matrix and do not affect the unbiasedness of It\^o integrals, so no extra constraints arise beyond the $d{=}1$ case. Thus, presenting the classification in one dimension entails no loss for our purposes.
\end{remark}


\subsection{Case 1: \texorpdfstring{$\mathcal{H}=\mathcal{H}^{\text{Pol}}$}{H = H (Pol)} and \texorpdfstring{$\mathfrak{T}_T=\{T\}$}{T\_T = \{T\}}} \label{subsec:HPol}

We fix a $1-$dimensional rough noise $\XXX$ in the rest of this section. Consider first the cone of all polynomial-induced controlled paths:
 \begin{equation}
        \mathcal{H}_{\XXX}^{\text{Pol}}:= \Big\{ \Big(P(X_t), DP(X_t),..., D^{(k-1)}P(X_t)\Big):  \; \forall P \;\text{polynomial}\Big\}.
    \end{equation}
We omit ${\mathfrak{T}}_T$ =$\{T\}$ in the notation in this subsection. We have the following characterisation of $\mathcal{H}^{\text{Pol}}-$unbiased rough noises, which can be viewed as the inverse result of the constructions of Gaussian-Hermite rough paths in Subsection \ref{subsec:The unbiased Gaussian rough integrator}.

\begin{theorem}[$\mathcal{H}^{\text{Pol}}-$unbiased rough integrator]\label{thm:Pol-unbiased rough integrators}
    Let $\XXX$ be a rough noise, renormalised by bounded variation paths $G^2,...,G^k$. Then, $\XXX$ is an $\mathcal{H}^{\mathrm{Pol}}-$unbiased rough integrator if and only if it is a Gaussian-Hermite rough path as in Definition \ref{Unbiased Rough Paths}.
\end{theorem}

The ``if" part was already indicated in Theorem \ref{thm:Unbiasedness for polynomials}. For the ``only if" part, we first prove the following lemma for the moments of the underlying noise process $X$:
\begin{lemma}\label{lemm:Uniform bounds on moments}
    Let $\XXX$ be a rough noise, renormalised by bounded variation paths $G^2,...,G^k$. If $\XXX$ is an $\mathcal{H}^{\mathrm{Pol}}-$unbiased rough integrator, then there exists a constant $C< \infty$ such that
    \begin{equation}
        \sup_{t\in[0,T]} \mathbb{E}[|X_t|^n]\leq \frac{3+(-1)^{(n+1)}}{2}\cdot n! \cdot C^n\quad\quad\forall n\in \mathbb{N}.
    \end{equation}
    For example, one can choose $C:= \max(4,\sup_{j=2,...,k} TV(G^j_{[0,T]}))$.
\end{lemma}
\begin{proof}
    We prove this fact by induction. For $n=0$, the inequality is trivially true.
    Now, assume that the inequality holds up to some even $n$. Denote by $C:=\sup_{j=2,...,k} TV(G^j_{[0,T]})$. We then compute for any $t\in [0,T]$ and the even number $n$:
    \begin{equation}
        \begin{aligned}
            \mathbb{E}[|X_t|^{n+2}]  =  \mathbb{E}[X_t^{n+2}] &= \mathbb{E}\Big[ \sum_{j=2}^{\min(n+2, k)}  \int_0^t \frac{(n+2)!}{(n+2-j)! } X_s^{n+2-j} \mathrm{d} G^j(s) \Big]\\
            & \leq  \sum_{j=2}^{\min(n+2, k)} \mathbb{E}\Big[\int_0^t \Big|\frac{(n+2)!}{(n+2-j)! } X_s^{n+2-j} \Big|\cdot \Big| \mathrm{d} G^j(s) \Big|\Big]\\
            & \leq \sum_{j=2}^{\min(n+2, k)} \int_0^t \frac{(n+2)!}{(n+2-j)! } \mathbb{E}[|X_s^{n+2-j}|]\cdot \Big| \mathrm{d} G^j(s) \Big|\\
            & \leq \sum_{j=2}^{\min(n+2, k)} \frac{(n+2)!} {(n+2-j)!}\cdot \sup_{t\in[0,T]} \mathbb{E}[|X_t|^{n+2-j}]\cdot \sup_{j=2,...,k} TV(G^j_{[0,T]})\\
            & \leq (n+2)! \sum_{j=2}^{\min(n+2, k)}  2\cdot C^{n+3-j}\leq(n+2)! \cdot C^{n+2},
        \end{aligned}
    \end{equation}
    where the first line is due to the It\^o formula in Theorem \ref{Ito type formula} and the assumption of $\mathcal{H}^{\mathrm{Pol}}-$unbiasedness. We used Fubini-Tonelli due to the non-negativity of the integrands in the third line. The last equality is due to the choice of $C\geq 4$. Now for the $(n+1)-$absolute moments, we use Cauchy-Schwarz to compute:
    \begin{equation}
        \begin{aligned}
            \mathbb{E}[|X_t|^{n+1}] & \leq \mathbb{E}[|X_t|^{n}]^{\frac{1}{2}}\cdot \mathbb{E}[|X_t|^{n+2}]^{\frac{1}{2}}\\
            & \leq \sqrt{n!(n+2)!C^{2n+2}} = (n+1)!\cdot\sqrt{\frac{n+2}{n+1}}\cdot C^{n+1} \leq 2(n+1)!\cdot C^{n+1} \quad\quad\forall t\in [0,T].
        \end{aligned}
    \end{equation}
    The desired inequality is thus proven for all $n\in \mathbb{N}$.
\end{proof}

Now we may prove Theorem \ref{thm:Pol-unbiased rough integrators}
\begin{proof}[Proof of Theorem \ref{thm:Pol-unbiased rough integrators}]
    We only need to prove the ``only if" part. The above lemma ensures the existence of the characteristic function and moment generating function of $X_t$ inside the disc of radius $\frac{1}{C}$ for any $t\in [0,T]$ in form:
    \begin{equation}\label{eq:Char Func}
        \phi_t(u):= \mathbb{E}[e^{u X_t}]=\sum_{n=0}^\infty \frac{\mathbb{E}[X_t^n]}{n!}\cdot u^n\quad\quad \forall |u|<\frac{1}{C}.
    \end{equation}
    Moreover, $\phi_t$ is analytic, and the series above converges absolutely inside that disc. Recall that the It\^o formula in Theorem \ref{Ito type formula} gives:
    \begin{equation}
    \begin{aligned}
        \mathbb{E}[X^n_t] &= \mathbb{E}\Big[ \sum_{j=2}^{\min(n, k)}  \int_0^t \frac{(n)!}{(n-j)! } X_s^{n-j} \mathrm{d} G^j(s) \Big]\\
        & =\sum_{j=2}^{\min(n, k)} \int_0^t \frac{1}{(n-j)! } \mathbb{E}[X_s^{n-j}]\mathrm{d}G^j(s)
        \quad\quad\forall t\in [0,T] \;\mathrm{and}\; n\in \mathbb{N},
    \end{aligned}
    \end{equation}
    where we used Fubini in the second line due to the uniform bound on the absolute moments in the previous lemma. Plugging in \eqref{eq:Char Func} for the moment generating function gives:
    \begin{equation}
        \begin{aligned}
            \phi_t(u) & = \sum_{n=0}^\infty u^n\cdot\Big(\sum_{j=2}^{\min(n, k)} \int_0^t \frac{1}{(n-j)! } \mathbb{E}[X_s^{n-j}]\mathrm{d}G^j(s)\Big)\\
            & = \sum_{i=0}^\infty u^i\cdot\Big(\sum_{j=2}^k \int_0^t \frac{1}{i!}\mathbb{E}[X_s^i]\cdot u^j\cdot\mathrm{d}G^j(s)\Big)\\
            & = \int_0^t \Big(\sum_{i=0}^\infty  \frac{u^i}{i!}\mathbb{E}[X_s^i]\Big)\mathrm{d} \Big( \sum_{j=2}^k G^j(s) u^j \Big)\\
            & = \int_0^t  \phi_s(u) \mathrm{d} \Big(\sum_{j=2}^k G^j(s) u^j \Big)\quad\quad \forall t\in [0,T]\;\mathrm{and}\; u\in(-\frac{1}{C}, \frac{1}{C}),
        \end{aligned}
    \end{equation}
    where the second and third equalities use Tonelli's theorem by the absolute convergence of the series and the uniform bound on the absolute moments in the previous lemma, and we also used the change in index $i:=n-j$ in the second equality. Since $G^j$'s are of bounded variation and $X_0\equiv 0$, the only solution to the above integral equation is given by:
\begin{equation}
    \phi_t(u) = \phi_0(u) \exp\Big(\sum_{j=2}^k G^j(t) u^j\Big) = \exp \Big(\sum_{j=2}^kG^j(t) u^j \Big) \quad\quad \forall t\in [0,T]\;\mathrm{and}\; u\in(-\frac{1}{C}, \frac{1}{C}).
\end{equation}
This then extends to the same function analytically on the complex disc of radius $\frac{1}{C}$. Recall that Marcinkiewicz's theorem (\cite{Mar39} Th\'eor\`eme 2) states that the exponential of a polynomial of degree strictly larger than $2$ cannot be the characteristic function of a valid probability distribution. Thus, we must have $G^j\equiv 0$ for $j=3,...,k$ and
\begin{equation}
    \phi_t(u) = \exp(G^2(t)\cdot u^2) \quad\quad \forall t\in [0,T]\;\mathrm{and}\; |u|<\frac{1}{C}.
\end{equation}
This implies that $X_t$ is normal, $G^2(t) = -\frac{1}{2}\cdot \mathbb{E}[X^2_t]$ and $G^j\equiv 0$ for $j\geq 3$. By the definition of the perturbation functions $G^2,...G^k$, this means exactly that the rough path $\XXX$ is the Gaussian-Hermite rough path associated with $X$ as constructed in Definition \ref{Unbiased Rough Paths}.
\end{proof}

We note that it is necessary to work with rough noise in the above theorem, i.e., only with deterministic renormalisation functions. The same arguments above will not apply for random rough paths with stochastic renormalisation terms, since Fubini-type theorems are then no longer applicable. General $\mathcal{H}^{\mathrm{Pol}}-$unbiased rough integrator with stochastic renormalisations can be succeeded by rough integrals against  $\mathcal{H}^{\mathrm{Pol}}-$unbiased rough noises or RDEs driven by $\mathcal{H}^{\mathrm{Pol}}-$unbiased rough noises. We give a simple example with the standard Brownian motion below.

\begin{example}
    We let $\mathbf{B}^{\text{It\^o}}$ be the It\^o rough path lift of a $1-$dimensional standard Brownian motion and let 
    \begin{equation}
        (X_t, X'_t):=\Big(\int_0^t (B_t,1)\mathrm{d}\mathbf{B}^{\text{It\^o}}, B_t \Big)\quad\quad\forall t\in [0,T]
    \end{equation}
    be the controlled path induced by the second-order increment in $\mathbf{B}^{\text{It\^o}}$ itself. Moreover, let $\XXX:=(X, \XX)$ be the rough path lift of the controlled path $(X_t, X'_t)$ via Corollary \ref{CPasRP}. Via Proposition \ref{prop:renormalisation succeeded}, the renormalisation term encoded in $\XXX$ is then given by
    \begin{equation}
        G_{\mathbf{X}}(t): = \int_0^t B^2_u \mathrm{d} G_{\mathbf{Z}}(u) = -\frac{1}{2} \int_0^t B^2_u \mathrm{d}u\quad\quad\forall t\in [0,T],
    \end{equation}
    non-deterministic, coinciding with the quadratic variation term. However, letting $P$ be any polynomial, we then have
    \begin{equation}
    \begin{aligned}
        \mathbb{E}\Big[\int_0^T (P(X_t), DP(X_t))\mathrm{d} \XXX_t \Big] & = \mathbb{E}\Big[\int_0^T (P(X_t), DP(X_t)\cdot B_t)\mathrm{d} (X_t, X'_t) \Big]\\
        & = \mathbb{E} \Big[\int_0^T (P(X_t)\cdot X_t, DP(X_t)\cdot B_t\cdot X_t + P(X_t)\cdot X'_t)\mathrm{d} \mathbf{B}^{\text{It\^o}} \Big]\\
        & = \mathbb{E}\Big[\int_0^T P(X_t)\cdot X_t \mathrm{d} B_t \Big] = 0,
    \end{aligned}
    \end{equation}
    where the first equality is due to Proposition \ref{Consistency}, the second is by Proposition \ref{Associativity}, and the last line is due to the fact that the rough integral against $\mathbf{B}^{\text{It\^o}}$ coincides with the classical It\^o calculus as well as the martingale property of classical It\^o calculus.
\end{example}

We make a final comment that so far, we cannot say anything about the increments of the underlying noise process $X$, since only Markovian-type integrands are involved. In particular, we can only conclude that $X$ has Gaussian marginals, but not yet that it has Gaussian increments.


\subsection{Case 2: \texorpdfstring{$\mathcal{H}=\operatorname{Span}\!\big(\mathcal{H}^{\mathrm{Pol}}\cup\mathcal{H}^{\mathrm{pSig}}\big)$}{H = Span(H (Pol) ∪ H (pSig))}and \texorpdfstring{$\mathfrak{T}_T=[0,T]$}{T\_T = [0,T]}}\label{subsec:HSig}

In this subsection, we enlarge the set of integrands considered for the unbiasedness of a rough integrator. We fix a 1-dimensional rough noise $\XXX:=(1,X,\XX^2,...\XX^k)$ in this subsection.

Recall from Proposition \ref{prop:Sig as CP} that $\XXX$ lift to a signature $\XXX^\infty:=(1,X,\XX^2,...)$ inductively via rough integrals, and that each signature component $\XX^n$ can again be viewed as a controlled path $\forall n\in \mathbb{N}$. We denote by $\prescript{n}{s}{Y}:=(\prescript{n}{s}{Y}^{(1)},...,\prescript{n}{s}{Y}^{(k)})$ as the piecewise controlled path induced by the path $\prescript{n}{s}{Y}^{(1)}_t: = \XX^n_{s,t}\cdot \mathbbm{1}_{t\geq s}$ as in Proposition \ref{prop:Sig as CP} and introduce the following piecewise controlled paths:
\begin{equation}
    \mathcal{H}^{\mathrm{pSig}}_{\XXX}:=\{\prescript{n}{s}{Y}: \; \forall s\in [0,T] \text{ and }\forall n\in \mathbb{N}\}.
\end{equation}

\begin{remark}
    The controlled path $\prescript{n}{s}{Y}$ can be made valid without the indicator function $\mathbbm{1}_{t\geq s}$. However, doing so will require the information of $X_s$ to know $\prescript{n}{s}{Y}_t$ for even $t<s$. In other words, the indicator function is necessary to ensure that $\prescript{n}{s}{Y}$ is adapted with respect to the underlying filtration. On the other hand, by restricting $n\neq 0$, this type of piecewise controlled path now allows minimal use of information about the path history (the historical state $X_s$) in a time subinterval, but without a jump.
\end{remark}

\begin{remark}
    Note that the current setting is really an extension of the previous subsection. Recall that in the $1-$dimensional case, the signature components are essentially polynomials of the underlying path increments with some renormalisation terms. In other words, the set of integrands is enlarged from polynomials of the current state and time to Hermite polynomials of path increments and renormalisation increments in any time subinterval, in which path independence is introduced.
\end{remark}

We now let
\begin{equation}
    \mathcal{H}:=\text{Span}( \mathcal{H}^{\text{Pol}} \cup \mathcal{H}^{\text{pSig}}) \quad\quad \text{and} \quad\quad {\mathfrak{T}}_T:=[0,T],
\end{equation}
and classify the $\mathcal{H}\times  {\mathfrak{T}}_T-$ unbiased rough integrators in this subsection. As discussed in Remark \ref{rem:Unbiased GRP not for Sig}, it is immediate that not all Gaussian-Hermite rough paths can remain unbiased with respect to $\mathcal{H}\times  {\mathfrak{T}}_T$. Indeed, such unbiased rough integrators must be very close to the Gaussian-Hermite lift of a deterministic time change of a standard Brownian motion.

\begin{definition}[Chen-Hermite almost Brownian motion]\label{def:CHABM}
    Let $(X_t)_{t\in[0,T]}$ be a $1-$dimensional stochastic process with $X_0=0$. We call $(X_t)_{t\in[0,T]}$ a Chen-Hermite almost Brownian motion if the following holds:
    \begin{itemize}
        \item (Brownian increments) $X_{s,t}\sim \mathcal{N}(0,|t-s|)$ for all $s,t \in [0,T]$;
        \item (Chen-Hermite balancing condition) $\sum_{i=1}^{n-1} \mathbb{E}[H_i(X_{s,u}, -\frac{1}{2}(u-s))\cdot H_{n-i}(X_{u,t}, -\frac{1}{2}(t-u))]=0$ for all $0\leq s\leq u\leq t\leq T$ and all $n\geq 2.$
    \end{itemize}
\end{definition}

\begin{remark}\label{rem:CHABM vs BM}
    It is obvious that a standard Brownian motion is a Chen-Hermite almost Brownian motion due to independent increments. Conversely, the Chen-Hermite balancing condition at level $n=2$ means essentially that the process has uncorrelated adjacent increments, so only joint Gaussianity of the increments is missing. We are not able to prove in this paper, but we conjecture that the Chen-Hermite balancing condition is strong enough to recover joint Gaussianity. We shall make some more comments on the Chen-Hermite almost Brownian motion in Appendix \ref{subsubsec:CHABM}.
\end{remark}

\begin{conjecture}\label{conjecture}
    Any Chen-Hermite almost Brownian motion is a $1-$dimensional standard Brownian motion.
\end{conjecture}

By allowing signature-induced piecewise controlled integrands, we have the following classification of $\mathcal{H} \times {\mathfrak{T}}_T$-unbiased rough integrators.

\begin{theorem}[$\mathcal{H} \times {\mathfrak{T}}_T$-unbiased rough integrator]\label{Thm:Main Theorem}
    Let $\XXX$ be a rough noise, renormalised by bounded variation paths $G^2,...,G^k$. Then, $\XXX$ is an $\mathcal{H} \times {\mathfrak{T}}_T$-unbiased rough integrator if and only if it is the Hermite rough path lift of a deterministic time change of a Chen-Hermite almost Brownian motion.
\end{theorem}
\begin{proof}
    We first show the ``only if" part. From Theorem \ref{thm:Pol-unbiased rough integrators}, we already know that $\XXX$ must be an unbiased Gaussian rough integrator. In other words, $X$ is centered with Gaussian marginals and 
    \begin{equation}\label{eq:Hermite for n small}
    \begin{aligned}
        G_{s,t} &:= -\frac{1}{2} (\mathbb{E}[X^2_t]-\mathbb{E}[X^2_s]) \quad\quad\forall s,t\in [0,T];\\
        \XX^n_{s,t} &:= H_n(X_{s,t}, G_{s,t})\quad\quad \forall n\leq k.
    \end{aligned}
    \end{equation}

    We first show that \eqref{eq:Hermite for n small} holds also for $n>k$. We show this by induction, and this is just a pure rough path tautology. Assume \eqref{eq:Hermite for n small} holds up to level $n-1$. Since $n>k$, then we may define \begin{equation}
        \XXX^n_{s,t}:=(1,X_{s,t}, H_2(X_{s,t}, G_{s,t}),...,H_k(X_{s,t}, G_{s,t}),..., H_n(X_{s,t}, G_{s,t})) \quad\quad\forall (s,t)\in \Delta_{[0,T]}
    \end{equation}
    as a well-defined $\frac{1}{n}-$H\"older continuous rough path by Proposition \ref{Rough Paths via Bell Polynomials}. We then have:
    \begin{equation}\label{eq:signatureidentity}
        \begin{aligned}
            H_n(X_{s,t}, G_{s,t}) &=\int_s^t \Big(H_{n-1}(X_{s,u}, G_{s,u}), H_{n-2}(X_{s,u}, G_{s,u}),...,1
\Big)\mathrm{d}\XXX^n_u\\
            &= \lim_{|\mathcal{P}|\downarrow0}\sum_{[u,v]\in \mathcal{P}}\Big(\sum_{i=1}^k H_{n-i}(X_{s,u}, G_{s,u})\cdot\XX^i_{u,v} + \sum_{i=k+1}^n  H_{n-i}(X_{s,u}, G_{s,u})\cdot H_{i}(X_{u,v}, G_{u,v}) \Big)\\
            & = \lim_{|\mathcal{P}|\downarrow0}\sum_{[u,v]\in \mathcal{P}}\sum_{i=1}^k H_{n-i}(X_{s,u}, G_{s,u})\cdot\XX^i_{u,v}\\
            & = \lim_{|\mathcal{P}|\downarrow0}\sum_{[u,v]\in \mathcal{P}}\sum_{i=1}^k \XX^{n-i}_{s,u}\cdot\XX^i_{u,v}\\
            & = \int_s^t (\XX^{n-1}_{s,u},...,\XX^{n-k}_{s,u})\mathrm{d}\XXX^n_u = \XX^n_{s,t}\quad\quad\quad\quad\forall (s,t)\in \Delta_{[0,T]},
        \end{aligned}
    \end{equation}
    where we used Proposition \ref{prop:Sig as CP} in the first and last equality. The second and fifth equalities are definitions of the corresponding rough integrals. The third equality is due to the fact that $H_i(X_{s,t}, G_{s,t})$ is $i\alpha-$H\"older continuous with $i\alpha>1$ for $i> k$, and the fourth equality is by the induction assumption.

    Now we prove that the underlying noise process $X$ has Gaussian increments. By the definition of Bell's complete polynomials and the characterization of $\XX^n$ above, we have for all $s,t\in [0,T]$ and for $z\in \mathbb{R}$ small
    \begin{equation}\label{eq:exponential of increments}
        \exp(X_{s,t}\cdot z + G_{s,t}\cdot z^2)= \sum_{n=0}^\infty \frac{\XX^n_{s,t}}{n!}\cdot z^n.
    \end{equation}
    By the Gaussian nature of $X_t$ and $X_s$ and the uniform bounds on the moments in Lemma \ref{lemm:Uniform bounds on moments}, taking expectation on both sides of \eqref{eq:exponential of increments} and an exchange of sum and expectation by Fubini-Tonelli gives:
    \begin{equation}
        \begin{aligned}
            \mathbb{E}[\exp(X_{s,t}\cdot z + G_{s,t}\cdot z^2)] &= \sum_{n=0}^\infty \frac{\mathbb{E}[\XX^n_{s,t}]}{n!}\cdot z^n \\
            & =1+ \mathbb{E}[X_{s,t}]\cdot z + \sum_{n=2}^\infty \frac{1}{n!}\cdot z^n \cdot \mathbb{E}\Big[\int_0^t\prescript{n-1}{s}{Y}_u \mathrm{d}\XXX_u\big] \\&=1 \quad\quad \forall (s,t)\in \Delta_{[0,T]} \text{ and } \forall z \text{ small},
        \end{aligned}
    \end{equation}
    where we used the $\mathcal{H}\times{\mathfrak{T}}_T$-unbiasedness of $\XXX$ and that $X$ is centered in the last line, and the identity $\XX^n_{s,t} = \int_0^t\prescript{n-1}{s}{Y}_u \mathrm{d}\XXX_u$ due to Proposition \ref{prop:Sig as CP}. Thus, the moment generating function of $X_{s,t}$ is given by $M_{s,t}(z)=\exp (-G_{s,t}\cdot z^2)$ for $z$ small. Thus, we have
    $X_{s,t}\sim \mathcal{N}(0, -2G_{s,t})$ for all $s,t\in [0,T]$. 

    Now, taking expectation on the level-$n$ Chen's relation, we have
    \begin{equation}
    \begin{aligned}
         0 = \mathbb{E}[\int_0^t \prescript{n-1}{s}{Y}_v \mathrm{d} \XXX_v] &= \mathbb{E}[\XX^{n}_{s,t}]\\
         & = \mathbb{E}[\XX^{n}_{s,u}] + \mathbb{E}[\XX^{n}_{u,t}] + \sum_{i=1}^{n-1}\mathbb{E}[\XX^i_{s,u}\cdot \XX^{n-i}_{u,t}]\\
         & = \mathbb{E}\Big[\int_0^u \prescript{n-1}{s}{Y}_v\mathrm{d} \XXX_v\Big] + \mathbb{E}\Big[\int_0^t \prescript{n-1}{u}{Y}_v \mathrm{d} \XXX_v \Big] + \sum_{i=1}^{n-1}\mathbb{E}[\XX^i_{s,u}\cdot \XX^{n-i}_{u,t}]\\
         &=  \sum_{i=1}^{n-1}\mathbb{E}[\XX^i_{s,u}\cdot \XX^{n-i}_{u,t}] \\
         & = \sum_{i=1}^{n-1} \mathbb{E}\Big[H_i(X_{s,u}, -\frac{1}{2}(u-s)) \cdot H_{n-i}(X_{u,t}, -\frac{1}{2}(t-u))\Big] \quad\quad\quad\quad\forall s,u,t\in [0,T],
    \end{aligned}
    \end{equation}
    which then recovers the desired Chen-Hermite balancing condition for $\XXX$. In particular, the level-$2$ Chen-Hermite balancing condition indicates that $X$ has uncorrelated adjacent increments. Thus, we have the monotonicity of the variance function:
    \begin{equation}\label{eq:variancemonotone}
        \text{Var}(X_{t}) = \text{Var}(X_{0,s})+ \text{Var}(X_{s,t})\geq  \text{Var}(X_{s}) \quad\forall (s,t)\in \Delta_{[0,T]}.
    \end{equation}
    Let $\sigma(t): = \inf \{ s: \text{Var}(X_s)>t\}$  be the inverse function and define $\tilde{X}_t:=X_{\sigma(t)}$ as a new process with Gaussian increments. It is then straightforward to check that $(\tilde{X}_t)_{t\in [\sigma(0), \sigma(T)]}$ is a Chen-Hermite almost Brownian motion from the properties of $X$ proven above. Together with 
    \begin{equation}
        \tilde{G}_t:= -\frac{1}{2} \mathrm{Var}(\tilde{X}_t) = -\frac{1}{2} \mathrm{Var}(X_{\sigma(t)})\quad\quad\forall t \in [0,T],
    \end{equation}
    we write:
    \begin{equation}
        \tilde{\XXX}_{s,t}:= (1,\tilde{X}_{s,t},..., H_k(\tilde{X}_{s,t}, \tilde{G}_{s,t}))\quad\quad\forall (s,t)\in \Delta_{[0,T]},
    \end{equation}
    as the associated Gaussian-Hermite rough paths. It is then straightforward to check that    $\XXX_{s,t}=\tilde{\XXX}_{\text{Var}(X_s),\text{Var}(X_t)}$ is a deterministic time change of the Gaussian-Hermite rough path lift of $\tilde{\XXX}.$

    The``if" part is simply due to the proven fact that 
    \begin{equation}
        \mathbb{E}\Big[\int_0^t \prescript{n}{s}{Y}_u \mathrm{d} \XXX_u\Big] = \mathbb{E} [H_{n+1}(X_{s,t}, G_{s,t})]\quad\quad \forall (s,t)\in \Delta_{[0,T]},
    \end{equation}
    and that Hermite polynomials are orthogonal with respect to Gaussian distributions.
\end{proof}

As discussed in Remark \ref{rem:CHABM vs BM}, if we enforce the Gaussian increments requirement, which might be stronger than the Chen-Hermite balancing condition, then we have immediately:
\begin{corollary}\label{cor:joint Gaussian increments then Bm}
    Let $\XXX$ be as above. If, in addition, the joint increments of $X$ are Gaussian, then $\XXX$ must be the It\^o rough path lift of a deterministic time change of a standard Brownian motion.
\end{corollary}
\begin{proof}
    The Chen-Hermite balancing condition at level $n=2$ gives the uncorrelatedness of disjoint increments, which is then enforced to independence by joint Gaussianity. Together with Brownian increments, we must have a standard Brownian motion. 
\end{proof}


\subsection{Case 3: \texorpdfstring{$\mathcal{H}^\#=\operatorname{Span}\!\big(\mathcal{H}^{\mathrm{Pol}}\cup \mathcal{H}^{\mathrm{AdpSig}}\big)$}{H\# = span(H (Pol) ∪ H (AdpSig))} and \texorpdfstring{$\mathfrak{T}_T=[0,T]$}{T\_T = [0,T]}}\label{subsec:Hsimple}

In the previous subsection, we needed to jointly impose Gaussian increments to enforce a standard Brownian motion in Corollary \ref{cor:joint Gaussian increments then Bm}. Indeed, this can be avoided immediately if we further enlarge the set of admissible continuous trading portfolios to encode more path-dependence. We define 
\begin{equation}
    \mathcal{H}^{\mathrm{AdpSig}}_{\XXX}:=\{\eta \cdot \prescript{n}{s}{Y}: \; \forall s\in [0,T], \; \forall \eta\in L^\infty(\mathcal{F}_s) \text{ and }\forall n\in \mathbb{N}\},
\end{equation}
and call such portfolios adaptedly scaled piecewise signature-induced portfolios. Consequently, we define
\begin{equation}
\mathcal{H}^\# :=\operatorname{Span}\!\big(\mathcal{H}^{\mathrm{Pol}}\cup \mathcal{H}^{\mathrm{AdpSig}}\big).
\end{equation}
Again, we restrict here $n\neq 0 $ to avoid any discontinuity in the trading portfolios. Since $\omega-$wise, $\eta$ is just a constant, adding such a multiplier does not break the pathwise controlledness of this class of portfolios. Then, essentially the same argument from the previous subsection can be passed to the conditional expectation level, and we obtain the following sharper result:
\begin{theorem}[$\mathcal{H}^\# \times {\mathfrak{T}}_T$-unbiased rough integrator]\label{Thm:Main Sharp Theorem}
    Let $\XXX$ be a rough noise, renormalised by bounded variation paths $G^2,...,G^k$. Then, $\XXX$ is an $\mathcal{H}^\# \times {\mathfrak{T}}_T$-unbiased rough integrator if and only if it is the It\^o rough path lift of a deterministic time change of a standard Brownian motion.
\end{theorem}

We will need the following auxiliary lemma for the conditional expectation level proof:

\begin{lemma}[Positivity recovers conditional centering]\label{lemm:condcentering}
Let $Y$ be a real random variable, and let $\mathcal{G}$ be a $\sigma-$algebra. Fix $a \geq 0$. Assume that $\mathbb{E} [e^{r|Y|}]<\infty$ for some $r>0$ and that, for every $B \in \mathcal{G}$ and every $m \geq 2$,
\begin{equation}\label{eq:condhermexp0}
    \mathbb{E}\Big[\mathbf{1}_B \, H_m\Big(Y,-\frac{1}{2} a\Big)\Big]=0 .
\end{equation}
Then
\begin{equation}\label{eq:condcharfunc}
    \mathbb{E}\left[e^{i \xi Y} \mid \mathcal{G}\right]=\exp \left(-\frac{1}{2} a \xi^2\right), \quad \forall \xi \in \mathbb{R} . 
\end{equation}
In particular, $\mathbb{E}[Y \mid \mathcal{G}]=0$.
\end{lemma}
\begin{proof}
Fix $B \in \mathcal{G}$ with $\mathbb{P}(B)>0$, and let $\mu_B$ and $m_B$ be the conditional law of $Y$ on $B$ and its mean:
\begin{equation}
    \mu_B(C):=\frac{\mathbb{P}(B \cap\{Y \in C\})}{\mathbb{P}(B)} , \quad \forall C \in \mathcal B (\mathbb R) , \quad m_B:=\int_{\mathbb R} x \mu_B(d x) . 
\end{equation}
By \eqref{eq:condhermexp0}, we have
\begin{equation}
    \int_{\mathbb R} H_m\big(x,-\frac{1}{2} a\big) \mu_B(\mathrm{d} x)=0, \quad \forall m \geq 2.
\end{equation}
Thus, again, by the exponential Taylor expansion with Hermite polynomials, the exponential-moment assumption allows us to rewrite
\begin{equation}
\begin{aligned}
    \int_{\mathbb R} \exp \Big(\lambda x-\frac{1}{2} a \lambda^2\Big) \mu_B(d x)  & = \sum_{n=0}^{\infty} \lambda^n \cdot\int_{\mathbb R} H_n \Big(x, -\frac{1}{2} a \Big) \mu_B(\mathrm{d}x) \\ &=1+\lambda m_B,
\end{aligned}
\end{equation}
which first holds for a small convergence radius and then extends by analytic continuation to the imaginary axis. Putting $\lambda=i \xi$ gives its Fourier transform
\begin{equation}\label{eq:Fourier}
    \widehat{\mu}_B(\xi)=\int_{\mathbb R} e^{i \xi x} \mu_B(\mathrm{d} x)=\left(1+i \xi m_B\right) \exp \Big(-\frac{1}{2} a \xi^2\Big) .
\end{equation}
If $a=0$, the identity for $m=2$ gives $\int_{\mathbb R} x^2 \mu_B(\mathrm{d} x)=0$, so $\mu_B=\delta_0$.

Assume $a>0$. The right-hand side of \eqref{eq:Fourier} is the Fourier transform of the signed measure $\left(1+\frac{m_B}{a} x\right) \varphi_a(x) \mathrm{d} x$ where $\varphi_a$ is the density of the normal distribution $\mathcal{N}(0, a)$ with mean $0$ and variance $a$. Then, the uniqueness of the Fourier transforms yields
\begin{equation}
    \mu_B(d x)=\left(1+\frac{m_B}{a} x\right) \varphi_a(x) \mathrm{d} x . 
\end{equation}

The density in the last display is nonnegative for all $x$ only if $m_B=0$. Therefore, $\mu_B=N(0, a)$.
Thus, for every $B \in \mathcal{G}$ with positive probability, the conditional law of $Y$ on $B$ is $N(0, a)$. This is equivalent to \eqref{eq:condcharfunc}, and the conditional mean is zero.
\end{proof}

Now we may prove Theorem \ref{Thm:Main Sharp Theorem}, which does not depend on, and can be seen as parallel to Theorem \ref{Thm:Main Theorem}.
\begin{proof}[Proof of Theorem \ref{Thm:Main Sharp Theorem}]
Assume first that $\XXX$ is $\mathcal{H}^\# \times \mathfrak{T}_T$-unbiased. By Theorem \ref{thm:Pol-unbiased rough integrators}, $\XXX$ is a Gaussian-Hermite rough path, where the renormalisation term is given by $G(t):= -\frac{1}{2} \mathrm{Var}(X_t).$
    
Fix $0 \leq s \leq t \leq T, n \in \mathbb{N}$. By the signature identity in \eqref{eq:signatureidentity} and unbiasedness, we have
\begin{equation}
    0  =\mathbb{E}\left[\int_0^t \eta \cdot \prescript{n}{s}{Y}_u \mathrm{d} \mathbf{X}_u\right]=\mathbb{E}\left[\eta \cdot \XX_{s, t}^{n+1}\right]
 =\mathbb{E}\left[\eta H_{n+1}\left(X_{s, t}, G_{s, t}\right)\right]\quad\quad \forall   \eta \in L^{\infty}\left(\mathcal{F}_s\right).
\end{equation}
Because $n \in \mathbb{N}$, this yields, for every $m \geq 2$,
\begin{equation}\label{eq:condexphermite=0}
\mathbb{E}\left[H_m\left(X_{s, t}, G_{s, t}\right) \mid \mathcal{F}_s\right]=0 .
\end{equation}
The exponential-moment assumption in Lemma \ref{lemm:condcentering} holds because $X_s$ and $X_t$ have Gaussian marginals and $\left|X_{s, t}\right| \leq\left|X_s\right|+\left|X_t\right|$. Taking \(m=2\) in \eqref{eq:condexphermite=0} gives
\begin{equation}
    0=\mathbb E\!\left[H_2(X_{s,t},G_{s,t})\mid\mathcal F_s\right]
=\frac12\mathbb E[X_{s,t}^2\mid\mathcal F_s]+G_{s,t}.
\end{equation}
Thus, applying Lemma \ref{lemm:condcentering} with
\begin{equation}
    Y=X_{s, t}, \quad \mathcal{G}=\mathcal{F}_s, \quad a=-2G_{s, t}=\mathrm{Var}(X_t) - \mathrm{Var}(X_s)\geq 0,
\end{equation}
we get
\begin{equation}
    \mathbb{E}\left[e^{i \xi X_{s, t}} \mid \mathcal{F}_s\right]=\exp \left(\xi^2 G_{s, t}\right) . 
\end{equation}
Thus $X_{s, t}$ is independent of $\mathcal{F}_s$ and has law $\mathcal{N}\left(0, \mathrm{Var}(X_t) - \mathrm{Var}(X_s)\right)$. Consequently, $X$ is a continuous process with independent Gaussian increments and a deterministic variance clock. Same arguments as in the proof of Theorem \ref{Thm:Main Theorem} imply that $\XXX$ must be a deterministic time change of the It\^o lift of a standard Brownian motion.

Conversely, the``if" part is simply due to the proven fact that 
    \begin{equation}
        \mathbb{E}\Big[\int_0^t \eta\cdot\prescript{n}{s}{Y}_u \mathrm{d} \XXX_u\Big] = \mathbb{E} [\eta\cdot H_{n+1}(X_{s,t}, G_{s,t})]\quad\quad \forall (s,t)\in \Delta_{[0,T]},\; \forall \eta \in L^\infty (\mathcal{F}_s),
    \end{equation}
    independent increments of a deterministic time-changed Brownian motion, and that Hermite polynomials are orthogonal with respect to Gaussian distributions.
\end{proof}


\subsection{Unbiased rough integrators with monotone random renormalisations}\label{subsec:Random renorm}

Now we extend the classification to random rough paths with random, adapted and monotone renormalisation paths via a DDS-type time change argument, as earlier mentioned in Remark \ref{rem:whydetrenorm}. We fix a $1-$dimensional $\alpha-$continuous random rough path $\XXX := (1, X, \XX^2,..., \XX^k)$ that is adapted to the filtered probability space $(\Omega, \mathcal{F}, (\mathcal{F}_t)_{t\in [0, T]}, \mathbb{P})$ where $\mathcal{F}= \mathcal{F}_T$ for some $\alpha \in (0,1]$ and $k:=\lfloor \frac{1}{\alpha} \rfloor$. Assume moreover that $(\mathcal{F}_t)_{t\in [0, T]}$ satisfies the usual conditions of right-continuity and completeness.

Recall from \eqref{eq:Recover perturbation} that the renormalisation term $G^j({t})$ can be recovered by $(\XX^j_{0,t}, X_{0,t}, G^{2}_{0,t}..., G^{j-1}_{0,t})$ for every $j=2,...,k$. Inductively, all renormalisation terms are adapted as well. Moreover, we assume that they are monotone decreasing and start from $0$. We consider two different cases.

\subsubsection{\texorpdfstring{$\alpha\in\big(\tfrac{1}{3},\,\tfrac{1}{2}\big]$}{alpha in (1/3, 1/2]}}

We now fix $\alpha\in (\frac{1}{3}, \frac{1}{2}]$, the random rough path $\XXX=(1, X, \XX^2)$ is then characterised by a single renormalisation term $G^2$. For any $t>0,$ we define the $t-$renormalisation clock as the following stopping time:
\begin{equation}\label{eq:renorm clock}
    \tau_t:=\inf \{s:\; G^2(s)<-\frac{1}{2}t\}.
\end{equation}
Assume moreover $G^2(T)<-\frac{1}{2}T'$ for some $T'>0$. We then consider the following time change of $\XXX:$
\begin{equation}
    \widetilde{\XXX}_{s,t}:= (1, \widetilde{X}_{s,t}, \widetilde{\XX}_{s,t}) \quad\quad\quad\forall (s,t)\in \Delta_{[0,T']}
\end{equation}
where
\begin{equation}
    \widetilde{X}_{s,t}:= X_{\tau_s, \tau_t},\;\;\;\widetilde{\XX}_{s,t}:=\XX_{\tau_s, \tau_t} = \frac{1}{2} (X_{\tau_s, \tau_t})^2+ G^2_{\tau_s, \tau_t}\quad\quad\quad \forall (s,t)\in \Delta_{[0,T']}.
\end{equation}
The renormalisation term encoded in $\widetilde{\XXX}$ becomes by monotonicity and continuity of $G^2$:
\begin{equation}
    \widetilde{G}^2_{s,t}:= G^2_{\tau_s, \tau_t} = -\frac{1}{2} (t-s) \quad\quad\quad \forall (s,t)\in \Delta_{[0,T']}
\end{equation}
is thus deterministic.

Now we fix $(Y, Y')$, a random piecewise controlled path with respect to $\XXX$, the same time change then yields a piecewise controlled path with respect to $\widetilde{\XXX}$:

\begin{equation}
    (\widetilde{Y}_t, \widetilde{Y}'_t): = (Y_{\tau_t}, Y'_{\tau_t})\quad\quad\quad \forall t\in [0, T'],
\end{equation}
and consequently, by construction of the rough integral (cf. Proposition \ref{Rough Integral}), we have
\begin{equation}\label{eq:RI time change}
    \int_0^{\tau_t}(Y_s, Y'_s) \mathrm{~d} \XXX_s = \int_0^t (\widetilde{Y}_s, \widetilde{Y}'_s ) \mathrm{~d} \widetilde{\XXX}_s.
\end{equation}

Set \(\widetilde{\mathcal F}_t:=\mathcal F_{\tau_t}\). For any class
\(\mathcal K_{\widetilde \XXX}\) of controlled paths over \(\widetilde \XXX\),
define its pull-back by
\begin{equation}
    \mathcal K_{\XXX}^{TC}
:=
\{Y:\widetilde Y_t:=Y_{\tau_t}\in \mathcal K_{\tilde{\XXX}}\}.
\end{equation}
We write
\begin{equation}
    \mathcal{H}_{\XXX}^{\mathrm{Pol;TC}},\quad \mathcal{H}_{\XXX}^{\mathrm{pSig;TC}},\quad \mathcal{H}_{\XXX}^{\mathrm{AdpSig;TC}}
\end{equation}
for the corresponding pull-back classes, and set
\begin{equation}
    \mathcal{H}_{\XXX}^{\#,\mathrm{TC}}
:=
\operatorname{Span}\bigl(\mathcal{H}_{\XXX}^{\mathrm{Pol;TC}}\cup \mathcal{H}_{\XXX}^{\mathrm{AdpSig;TC}}\bigr).
\end{equation}

Thus, for such random rough paths, we may conclude:

\begin{corollary}\label{cor:randomrenorm1}
    Let $\XXX$ be as above. It is an $\mathcal{H}^{\mathrm{Pol;TC}}\times \{ \tau_t: t\in [0, T']\}$-unbiased rough integrator if and only if it is an adapted time change of a Gaussian-Hermite rough path on $[0, \tau_{T'}]$. It is an $\mathrm{Span}(\mathcal{H}^{\mathrm{Pol; TC}}\cup \mathcal{H}^{\mathrm{pSig; TC}})\times \{ \tau_t: t\in [0, T']\}$-unbiased rough integrator if and only if it is an adapted time change of the Hermite rough path lift of a Chen-Hermite almost Brownian motion on $[0, \tau_{T'}]$.
\end{corollary}
\begin{proof}
    It is straightforward to see that $(Y, Y')\in \mathcal{H}^{\mathrm{Pol;TC}}_{\XXX}$ (or $\mathcal{H}^{\mathrm{pSig;TC}}_{\XXX}$, respectively)
    if and only if $(\widetilde{Y}, \widetilde{Y}')\in \mathcal{H}^{\mathrm{Pol}}_{\widetilde{\XXX}}$ (or $\mathcal{H}^{\mathrm{pSig}}_{\widetilde{\XXX}}$, respectively). The statement then follows from \eqref{eq:RI time change}, the fact that $\widetilde{\XXX}$ has a deterministic renormalisation term, Theorem \ref{thm:Pol-unbiased rough integrators} and Theorem \ref{Thm:Main Theorem}.
\end{proof}

More sharply, one uses the same arguments to obtain:

\begin{corollary}\label{cor:simple to BM with TC}
    Let $\XXX$ be as above. It is an $\mathcal{H}_{\XXX}^{\#,\mathrm{TC}}\times \{ \tau_t: t\in [0, T']\}$-unbiased rough integrator if and only if it is an adapted time change of the It\^o lift of a standard Brownian motion on $[0, \tau_{T'}]$. In particular, $\XXX$ is the It\^o rough path lift of a local martingale on $[0, \tau_{T'}]$.
\end{corollary}
\begin{proof}
    Same as the Proof of Corollary \ref{cor:randomrenorm1}.
\end{proof}

\begin{remark}
   If we assume $\XXX$ is an $\alpha-$H\"older continuous random rough path for some $\alpha\in (\frac{1}{3}, \frac{1}{2}]$ on the infinite time interval $[0, \infty)$, such that the monotone renormalisation term satisfies, $\lim_{t\rightarrow \infty} G^2(t) = -\infty$, the above corollaries extends to the stochastic interval $[0, \tau_{T'}]$ for any $T'>0$, following the classical DDS manner. In particular, Corollary \ref{cor:simple to BM with TC} enforces $X$ to be a local martingale on the whole time interval $[0, \infty)$. 
\end{remark}

\subsubsection{\texorpdfstring{$\alpha\in\big(0,\,\tfrac{1}{3}\big]$}{alpha in (0, 1/3]} with synchronous renormalisation clocks}\label{subsubsec:Syn}

Clearly, the above arguments do not automatically generalise to $\alpha<\frac{1}{3}$, since one has $\lfloor \frac{1}{\alpha} \rfloor - 1 =: k-1 \geq 2$ different choices of renormalisation clocks and a time change with respect to one of them does not ensure the other renormalisation terms become deterministic. Assume moreover $G^2(T)< -\frac{1}{2}T'$ for some $T'>0$. The above results extend only if $G^2,..., G^k$ share a synchronous renormalisation clock, i.e., the time-changed renormalisation 
\begin{equation}
    \widetilde{G}^j: \; [0,T'] \rightarrow \mathbb{R},\;\; t\mapsto G^j(\tau_t)\quad\quad\quad 
\end{equation}
is deterministic for all $ j=2,...,k$, where $\tau_t$ is the $t-$renormalisation clock with respect to $G^2$, as in \eqref{eq:renorm clock}. Then the same results follow from the same arguments.
\begin{corollary}\label{cor:randomrenorm2}
    Let $\XXX$ be an $\alpha-$H\"older continuous rough path for some $\alpha\in (0,1]$ with monotone renormalisation terms $G^2,... G^k$ that share a synchronous renormalisation clock as above. It is an $\mathcal{H}^{\mathrm{Pol;TC}}\times \{ \tau_t: t\in [0, T']\}$-unbiased rough integrator if and only if it is an adapted time change of a Gaussian-Hermite rough path on $[0, \tau_{T'}]$. It is an $\mathrm{Span}(\mathcal{H}^{\mathrm{Pol;TC}}\cup \mathcal{H}^{\mathrm{pSig;TC}})\times \{ \tau_t: t\in [0, T']\}$-unbiased rough integrator if and only if it is an adapted time change of the Hermite rough path lift of a Chen-Hermite almost Brownian motion on $[0, \tau_{T'}]$. It is an $\mathcal{H}_{\XXX}^{\#,\mathrm{TC}}\times \{ \tau_t: t\in [0, T']\}$-unbiased rough integrator if and only if it is an adapted time change of the It\^o rough path lift of a standard Brownian motion on $[0, \tau_{T'}]$. In particular, $\XXX$ is the It\^o rough path lift of a local martingale on $[0, \tau_{T'}]$.
\end{corollary}


\section{Market Model as Rough Paths and the No Controlled Free Lunch Condition}\label{sec:market}

We are now able to develop applications of our theory in financial modelling. In this section, we model general price processes as random rough paths, define a portfolio’s gain process via a well-defined rough integral against the price path, and clarify what we mean by No Controlled Free Lunch (NCFL) in this setting, which allows $\alpha$-h\"older continuous price paths for $\alpha>0$ arbitrarily small. We classify rough-paths-based market models with the (NCFL) condition in two different settings, the rough Bachelier-type market model and RDE market models. Throughout this section, we consider the scenario of continuous trading with an early liquidation allowed, i.e., no other jumps are allowed except a final cash out before the terminal time $T$.


\subsection{Frictionless market models as rough paths}\label{subsec:Markets as RP}
We first give meaning to stochastic continuous market models in the context of rough paths. The main difference compared with the standard semimartingale settings is that the gain process is understood as a rough integral. Throughout the rest of this paper, we fix a filtered probability space $(\Omega, \mathcal{F}, (\mathcal{F}_t)_{t\in [0, T]}, \mathbb{P})$ where $\mathcal{F}= \mathcal{F}_T$. Assume moreover that $(\mathcal{F}_t)_{t\in [0, T]}$ satisfies the usual conditions of right-continuity and completeness. We let $(S_t)_{t\in [0,T]}$ be a $(d+1)$-dimensional continuous $\mathcal{F}_t$-adapted stochastic process representing the prices of the underlying assets. If not otherwise specified, we use the convention that $S^0=1$, the riskless asset has a constant price.

\begin{definition}[Continuous rough market model]\label{def:Rough Market Model}
We call $\SSS:=(1, \Ss^1, \Ss^2, ...,\Ss^k)$ a rough market model associated with the the price process $S$ if the following conditions are satisfied:
\begin{itemize}
    \item (Adaptedness) $\SSS$ is an adapted two-parameter $T^{(k)}(\mathbb{R}^{d+1})$-valued stochastic process, i.e., a map from $\Delta_T\times \Omega$ to $T^{(k)}(\mathbb{R}^{d+1})$ such that $\SSS_{s,t}$ is $\mathcal{F}_t-$measurable for any $0\leq s\leq t\leq T$;
    \item  (RP Lift) There exists some $\alpha>\frac{1}{k+1}$ such that $\SSS$ is almost surely an $\alpha-$H\"older rough path such that $\Ss^1_{s,t} = S_t-S_s$ for any $0\leq s, t\leq T$.
\end{itemize}
In the following, we also use the self-explanatory notation $\SSS:=(1, S, \Ss^2, ...,\Ss^k)$ for such models.
\end{definition}

\begin{definition}[Controlled portfolios]\label{def:Controlled Portfolios}
    Consider the above rough market model. A (piecewise) controlled trading portfolio therein is understood as a random $\mathcal{F}_t$-adapted $\mathbb{R}^{d+1}\cong \mathcal{L}(\mathbb{R}^{d+1}, \mathbb{R})$-valued piecewise controlled path $Y(\omega):=(Y^{(1)}(\omega), Y^{(2)}(\omega),..., Y^{(k)})$. The value process of the trading portfolio $Y$ is defined as:
    \begin{equation}
        V_t(Y):= \sum_{n=0}^{d} Y^{(1); n}_t\cdot S^n_t.
    \end{equation}
\end{definition}

\begin{remark}
    Here, $Y^{(1)}_t(\omega)$ denotes the portfolio holdings across the traded assets at time $t$ and thus, the financial interpretation of the value process is self-evident. The controlledness requirements encode that the trading rule reacts to the recent path of the market. The higher order controlled components, $Y^{(i)}$ for $i=2,...,k$, on the one hand, play the role of ``sensitivities" quantifying how holdings adjust to market increments, while on the other hand, are necessary from a mathematical perspective to supply the well-definedness of the gain process as a rough integral, and hence the self-financing condition, as in \eqref{eq:Gain Process} and \eqref{eq:SF} below, which shall coincide with the self-financing condition in semimartingale settings in spirit. We will list a few typical classes of controlled portfolios in Subsubsection \ref{exp:Markovian Portfolios}, illustrating that the framework covers a broad range of strategies.
\end{remark}

Beyond continuous trading, we also allow a one-time cash-out before the terminal time $T$.
\begin{definition}[Admissible strategies]
    Consider the above rough market model. An admissible strategy is understood as a piecewise controlled trading portfolio $Y$ as above, together with an exit time $\tau$, which is a stopping time. The holding is then understood as 
    \begin{equation}
        Y^{\tau; (1);i}_t:= Y^{(1);i}_{t}\cdot \mathbbm{1}_{t<\tau} \quad\quad\mathrm{for}\; i=1,...,d 
    \end{equation}
    and 
    \begin{equation}
        Y^{\tau; (1);0}_t:= Y^{(1);0}_{t}\cdot \mathbbm{1}_{t<\tau} + \sum_{n=0}^{d} Y^{(1); n}_\tau\cdot S^n_\tau\cdot \mathbbm{1}_{t\geq\tau}.
    \end{equation}
    Then, the value process of the trading portfolio $(Y, \tau)$ is defined in the following self-evident way:
    \begin{equation}
        V_t(Y):= \sum_{n=0}^{d} Y^{(1); n}_{\tau\wedge t}\cdot S^n_{\tau\wedge t}.
    \end{equation}
\end{definition}

\begin{remark}
    We make the point that, except for the final liquidation at time $\tau$, when all risky assets are sold, no other jumps are allowed in the admissible strategies.
\end{remark}

\begin{definition}[Self-financing Condition]\label{def:Self-financing}
    Let $(Y,\tau)$ be an admissible trading strategy in a rough market model $\SSS$ as above. The gain process of the portfolio is defined as the corresponding rough integral:
    \begin{equation}\label{eq:Gain Process}
    G_t(Y):=\int_0^{\tau\wedge t} Y_u\mathrm{d} \mathbf{S}_u=
    \lim_{|\mathcal{P}| \downarrow 0}
    \sum_{[u,v]\in \mathcal{P}} Y^{(1)}_u S_{u,v} +Y^{(2)}_u \mathbb{S}^{2}_{u,v}+... +Y^{(k)}_u \mathbb{S}^{k}_{u,v},
\end{equation}
where $\mathcal{P}$ is a stochastic partition of $[0, \tau\wedge t]$ and we used the canonical isomorphism $\mathcal{L}(\mathbb{R}^{d+1}, \mathbb{R}) \cong \mathbb{R}^{d+1}$ in writing the rough integral. $Y$ is called self-financing if 
\begin{equation}\label{eq:SF}
    V_t(Y)- V_0(Y) = G_t(Y) \quad\quad \forall t\in [0,T].
\end{equation}
\end{definition}

\begin{remark}
    When $S$ is a (continuous) semimartingale and we equip it with the It\^o lift, the rough gain $\int_0^{\tau\wedge t} Y_u\,\mathrm d\mathbf S _u$ coincides with the classical It\^o integral $\int_0^{\tau\wedge t} Y^{(1)}_u\,\mathrm dS _u$ (cf. \cite{ALP24} Proposition 4.7); thus our self-financing condition above is essentially a rough path generalisation of the standard one.
\end{remark}

\begin{example}[Constant riskless asset price]\label{example:S^0=1}
    Consider the scenario of $(d+1)$ underlying assets with $S^0\equiv 1$. It is then essential to make sure that $S^0$ plays no role in the gain process. Indeed, let $\tilde{\SSS}$ be any rough path lift of the risky assets price process $(S^1,..., S^d)$. We shall lift $\tilde{\SSS}$ to a rough path lift of $(S^0,..., S^d)$ in the following way: for any word $w=w_1\cdots w_n$ with letters in the alphabet $\{0,...,d\}$, we set
    \begin{equation}
        \langle\SSS_{s,t}, w\rangle:= 
        \begin{cases}
            \langle\tilde{\SSS}_{s,t}, w\rangle \quad \quad \text{if} \quad w_i\neq 0,\quad \forall i=1,...,n;\\
            0 \quad\quad\text{otherwise}.
        \end{cases}
    \end{equation}
    In other words, we ensure that every postulated iterated integral involving a constant coordinate is trivial and the rough path lift is uniquely determined by the path formed by those non-constant coordinates.

    Now let $Y(\omega):=(Y^{(1)}(\omega), Y^{(2)}(\omega),..., Y^{(k)})$ be any trading portfolio in the rough market model $\SSS$. We define $\tilde{Y}(\omega):=(\tilde{Y}^{(1)}(\omega), \tilde{Y}^{(2)}(\omega),..., \tilde{Y}^{(k)})$ by
    \begin{equation}
        \mathcal{L} ((\mathbb{R}^{{d}})^{\otimes ^{i}}, \mathcal{L}(\mathbb{R}^d, \mathbb{R}))    \ni \tilde{Y}^{(i)}(\omega):=  \iota^* \circ Y^{(i)}(\omega)\circ\iota^{\otimes i} \quad\quad\forall i=1,...,k
    \end{equation}
    where 
    \begin{equation}
        \iota: \mathbb{R}^d\rightarrow \mathbb{R}^{d+1}, \quad (x_1,...,x_d)\mapsto (0, x_1,...,x_d)
    \end{equation}
    and $\iota^*$ is its dual map. $\tilde{Y}$ is essentially the portfolio on the risky assets. Indeed, it is straightforward to check that $\tilde{Y}$ is controlled by $\tilde{\SSS}$ with exactly the same error terms and that the two corresponding gain processes are the same:
    \begin{equation}
        G_t(Y):=\int_0^t Y_u\mathrm{d} \mathbf{S}_u = \int_0^t \tilde{Y}_u\mathrm{d} \tilde{\mathbf{S}}_u =: G_{t}(\tilde{Y}),
    \end{equation}
    indicating that the integral against the constant price path $S^0$ does not contribute to the gain process. We leave the justification of this fact to our readers since the intuition is rather clear.
\end{example}

Before giving some examples of the controlled portfolios as consequences of Subsection \ref{subsec:Examples of CP}, we make some first comments on the rough settings.

First, we make the point that we use the notation ``rough" not to mean that the sample paths of the price process that we will deal with are necessarily less regular than semimartingales (although we mostly care about such cases), the terminology ``rough" is used to rather mean that we are using rough paths and rough integral for the mathematical interpretation of necessary financial terminologies in the markets. Indeed, even for ``smoother" price processes, we can still use rough integrals for such interpretations, which allows renormalisations that can be hand-made and go beyond classical It\^o, Stratonovich, Young or Riemann-Stieltjes integrals. As we shall compare in Subsection \ref{subsec:RDE models}, the rough settings cover and indeed extend the classical semimartingale models for more general market noises.

Second, note that the first term in the compensated Riemann-Stieltjes sum in \eqref{eq:Gain Process} has a clear economic meaning as the gain in discrete settings, while the other terms remain rather mysterious. We do not argue or try to mislead our readers with the financial interpretation of any of the additional terms in the rough paths lift, controlled paths, or rough integrals. (See \cite{PP16} Section 4, \cite{ACLP23} Section 2, \cite{ALP24} Section 2 for some financial interpretation for the rough integral in the case of $\alpha\in (\frac{1}{3}, \frac{1}{2}]$ though.) At this point, we cannot connect them with any reasonable financial interpretations in our more general settings and would rather believe there are none. We use them solely to ensure the mathematical convergence of the compensated Riemann-Stieltjes sum, and thus a well-defined gain process in more general settings. However, one need not impose any interpretation on them and should think about the main contribution of this paper as follows: as a mathematical methodology to extend stochastic market models, the rough path renormalisation techniques do not carry the No Free Lunch condition outside of the classical It\^o-semimartingale setting. 

Finally, note that we do not put any classical admissibility constraints on (piecewise) controlled portfolios, e.g. an upper bound on maximal loaning, which are usually required in classical semimartingale models (cf. \cite{DS94}, \cite{DS06} Chapter 5). Indeed, requiring such constraints will not be a technical issue in the context of rough paths and can be achieved easily using a proper stopping time and piecewise controlled paths, e.g., stopping investment in risky assets when the maximum loaning is achieved. However, by not allowing jumps and requiring that the portfolio locally follow the stock price closely via the controlledness constraints in \eqref{eq:Controlled Path}, the value process would not behave dramatically for (piecewise) controlled portfolios. This is also justified by the fact (cf. Proposition \ref{Rough Integral}) that the gain process, as a rough integral, is again controlled by the price process, which means that it will behave crazily only if so does the price process. More specifically, classical ill-posed portfolios like the doubling strategy or the suicide strategy (cf. \cite{DS06} Section 8.2) are mathematically excluded from the setting. Thus, a discussion on the no-arbitrage condition within some class of controlled portfolios makes sense, and is desired, even without any classical admissibility constraints. We shall make clear the setting in the following subsection.


\subsubsection{Examples of controlled portfolios}\label{exp:Markovian Portfolios}
With the examples of controlled paths in Subsection \ref{subsec:Examples of CP}, we briefly realise them as different classes of trading portfolios in rough market models. We will use them to characterise the (NCFL) condition (cf. Definition \ref{def:NCFL}) soon. We fix an $(d+1)-$dimensional $\alpha-$H\"older continuous price path $S(\omega)_{[0,T]}$ and a rough path lift $\SSS(\omega):=(1,S(\omega),\Ss^2(\omega),...,\Ss^k(\omega))$, where $k:=\lfloor \frac{1}{\alpha} \rfloor$. We drop $\omega$ in the notation for convenience in this subsection.
\begin{example}[Markovian-type portfolios]
    Assume $\SSS$ is geometric, or non-geometric but $\alpha>\frac{1}{3}$, or non-geometric with $\alpha<\frac{1}{3}$ but with sufficiently regular renormalisation terms encoded as in Example \ref{Markovian Integrand}. Example \ref{exp:Markovian Portfolios}, \ref{Function of Geometric Rough Paths as Controlled Paths} and \ref{Markovian Integrand} and Remark \ref{rem:Function of Info and Path as CP} illustrate that sufficiently regular functions of the current price and current information of some other regular information paths can be lifted to well-defined controlled paths. They are thus well-defined as controlled portfolios. We call them Markovian-type portfolios, due to their nature that the current holding depends only on the current information, but not on the path history. They are referred to as functionally generated portfolios in \cite{ALP24} Section 3.
\end{example}

Markovian-type portfolios are restrictive, in particular, if we want to allow price processes beyond semimartingales (say, with memories) in the rough path settings. As introduced in Section \ref{subsec:Examples of CP}, there are mainly two ways to include path-dependence in controlled portfolios.

\begin{example}[Signature-based portfolios]
    Example \ref{exp:Sig as CP} indicates that the extended signature components of a rough path are again naturally controlled paths. In particular, linear combinations of the signature components are then well-defined as controlled portfolios. Such constructions encode the history of the path if there are at least two assets with non-constant price processes. Otherwise, one can still simply consider the signature of the time-augmented process. Moreover, the gain process of such portfolios can again be computed via the extended path signature components.
\end{example}

\begin{proposition}
    Let $\SSS$ be a rough market model, and $Y:=(Y^{(1)},..., Y^{(k)})$ be a controlled portfolio  with the holdings on each asset given by a linear combination of signature components of the price process:
    \begin{equation}
        Y^{(1);n}_t:= \sum_{l=1}^{L_n}  a_l\cdot \langle \XXX_{0,t}, w(l)\rangle \quad \forall n=0,...,d,
    \end{equation}
    with some words $w(l)$ of finite lengths and $a_l\in \mathbb{R}$, and the other terms in $Y$ are given as in \eqref{eq:Sig as CP} and \eqref{eq:Sig Portfolio int}. Then the gain process is
    \begin{equation}
        G_t(Y)=\sum_{n=0}^d \sum_{l=1}^{L_n} a_l\cdot  \langle \XXX_{0,t}, w(l)n\rangle
    \end{equation}
    where $w(l)n$ is the word obtained by adding the letter $n$ after $w(l)$.
\end{proposition}
\begin{proof}
    One can simply write 
    \begin{equation}
        Y^{(1)} = \sum_{n=0}^d  Y^{(1);n}_t \otimes e_{n} = \sum_{n=0}^d  \sum_{l=1}^{L_n}  a_l\cdot \langle \XXX_{0,t}, w(l)\rangle \otimes e_{n}.
    \end{equation}
    Then the desired gain process follows from \eqref{eq:Gain Process} and \eqref{eq:Sig Portfolio int}.
\end{proof}

Notably, let $\mathbf{B}$ (or $\tilde{\mathbf{B}}$, respectively) be the It\^o/geometric rough path lift of a standard Brownian motion (with time augmentation, respectively). The above construction recovers the It\^o/Stratonovich signature of a standard Brownian motion (with time augmentation, respectively).

\begin{remark}[Universal approximation of signature-based portfolios]
    We mention that for the dimension $d>1$, path signature uniformly approximates an arbitrary function on the path space due to the Stone-Weierstrass theorem. See e.g. \cite{HBS24} and \cite{CPS25} for studies in this direction. This indicates that signature-based portfolios do offer a very universal set of strategies for hedging and replicating.
\end{remark}

\begin{example}[Path-dependent functionally generated portfolios]
   As described in Example \ref{exp:Functional as CP}, for $\alpha\in (\frac{1}{3}, \frac{1}{2}]$, sufficiently regular non-anticipative functionals of a path can also be viewed as a controlled path in a slightly weaker sense, thus as controlled portfolios.
\end{example}

\subsection{A Kreps-Yan type theorem}\label{subsec:RKYTheorem}
In this section, we clarify what we mean by arbitrage and No Controlled Free Lunch (NCFL) in our settings, and provide an analogy of the Kreps-Yan Theorem (c.f. Theorem 5.2.2 in \cite{DS06}) in the rough paths context. We continue with the settings from the previous subsection.

\begin{definition}[Arbitrage]\label{def:Arbitrage}
    A self-financing strategy $(Y, \tau)$ in a rough market model $\SSS$ is called an arbitrage if
    \begin{equation}
        \mathbb{P}[G_T(Y)\geq 0]=1\quad \quad \text{and}\quad \quad \mathbb{P}[G_T(Y)> 0]>0.
    \end{equation}
    In other words, given no initial capital, we are able to at least break even at terminal time $T$ with probability $1$ and with a positive probability of making a profit.
\end{definition}

Let $\mathcal{D}_{\mathbf{S}}$ be the set of a.s. piecewise by-$\mathbf{S}$-controlled paths and let $\mathfrak{T}$ be the set of all stopping times no later than $T$. Note that $\mathbf{Y}_t \hat{\in} \mathcal{F}_t$ and $G_T(\mathbf{Y})\hat{\in} \mathcal{F}.$ We let $\mathcal{H}\subseteq \mathcal{D}_{\mathbf{S}}$ be a linear and convex subset containing $0$ and ${\mathfrak{T}}_T\subseteq \mathfrak{T} $ any subset containing the trivial stopping time $T$. Denote by:
\begin{equation}
    K^p(\mathcal{H}, {\mathfrak{T}}_T):= \Big\{X=\int_0^{\tau} Y_t\mathrm{d} \mathbf{S}_t \text{  for some } Y \in \mathcal{H} \text{ and }  \tau\in {\mathfrak{T}}_T\Big\}\cap L^p(\Omega, \mathcal{F}, \mathbb{P})
\end{equation}
for any $p \in (1, +\infty]$, as the cone of $L^p$-contingent claims that can be replicated some self-financing admissible strategy $(Y, \tau)$ with zero initial values, and thus will be priced as $0$ by arbitrage. Moreover, define
\begin{equation}
    C^p(\mathcal{H},  {\mathfrak{T}}_T): = \Big\{X-k \text{ for some }X\in K^p(\mathcal{H},  {\mathfrak{T}}_T) \text{ and } k\in L^p_+(\Omega, \mathcal{F}, \mathbb{P})\Big\}.
\end{equation}
as the cone of $L^p$-contingent claims dominated by some $X\in K^p(\mathcal{H})$.

\begin{definition}[No controlled free lunch]\label{def:NCFL}
    For $p\in (1,+\infty]$, $\mathbf{S}$ is said to satisfy the condition of no controlled free lunch (NCFL) of order $p$ within $\mathcal{H}\times  {\mathfrak{T}}_T$ if the closure $\overline{C}^p(\mathcal{H},  {\mathfrak{T}}_T)$ of $C^p(\mathcal{H},  {\mathfrak{T}}_T)$, taken with respect to the weak-star topology on $L^p(\Omega, \mathcal{F}, \mathbb{P})$, satisfies
    \begin{equation}
        \overline{C}^p(\mathcal{H},  {\mathfrak{T}}_T)\cap L^p_+(\Omega, \mathcal{F}, \mathbb{P})=\{0\}.
    \end{equation}
    $\mathbf{S}$ is said to satisfy (NCFL) of order $p$ if it satisfies (NCFL) of order $p$ within $\mathcal{D}_{\mathbf{S}}\times \mathfrak{T}$.
    A contingent claim $k \in L^p_+(\Omega, \mathcal{F}, \mathbb{P})$ is called an $L^p$-free lunch.
\end{definition}

\begin{remark}
    The weak-star topology on $L^p(\Omega, \mathcal{F}, \mathbb{P})$ above is rather a short-hand notation, and shall be understood in the following sense for different values of $p$: For $p=\infty$, we indeed mean the weak-star topology $\sigma(L^\infty, L^1)$, whereas for $p\in (1, \infty) $, we actually mean the usual weak topology on $L^p(\Omega, \mathcal{F}, \mathbb{P})$, which formally coincides with the weak-star topology on $L^p(\Omega, \mathcal{F}, \mathbb{P})$ as a dual space.
\end{remark}

\begin{remark}
    If we set $p=\infty$, we are back in the settings of Kreps-Yan. However, unlike simple strategies in Kreps-Yan's original settings, the payoff of a controlled portfolio is usually not bounded. Thus, we have to introduce the intersection with $L^p$ to embed them into a proper Banach space. More generally, one can place the gain cone in an Orlicz-type space without modifying too much what we shall prove next.
\end{remark}

Now, we may state the rough Kreps-Yan type theorem:

\begin{theorem}[Rough Kreps-Yan Theorem]\label{RoughKYTheorem}
    Let $p\in (1, +\infty],\;q\in [1,+\infty)$ with $\frac{1}{p}+ \frac{1}{q}=1$ and $\mathcal{H}\subseteq \mathcal{D}_{\mathbf{S}}$ a linear and convex subset containing $0$ and $T\in{\mathfrak{T}}_T\subset {\mathfrak{T}}$ any set of stopping times. A rough market $\mathbf{S}$ satisfies \textnormal{(NCFL)} of order $p$ within $\mathcal{H}\times {\mathfrak{T}}_T$, if and only if there exists an equivalent measure $\mathbb{Q}$ such that 
    \begin{equation}\label{eq:Rough Kreps-Yan}
        \frac{\mathrm{d}\mathbb{Q}}{\mathrm{d}\mathbb{P}}\in L^q_+ \quad \text{and} \quad \mathbb{E}_{\mathbb{Q}}[\int_0^{\tau} Y_t \mathrm{d} \mathbf{S}_t] =0 \quad\forall Y\in \mathcal{H},\;\forall \tau\in {\mathfrak{T}}_T,
    \end{equation}
    i.e., $\SSS$ is an $\mathcal{H}\times {\mathfrak{T}_T}-$unbiased rough integrator under $\mathbb{Q}$.
\end{theorem}

\begin{proof}
    The proof of the theorem requires only minimal modification to the proof of Theorem 5.2.2 in \cite{DS06}, and we repeat it for the convenience of readers in Appendix \ref{App:Kreps-Yan}.
\end{proof}
\begin{remark}
    Note that the rough integral in \eqref{eq:Rough Kreps-Yan} is defined pathwise and thus does not depend on the equivalent change of measure. 
\end{remark}

Note that the definition of our (NCFL) condition is made ad hoc, in the sense that we have freedom of choosing the class of admissible strategies $\mathcal{H}\times {\mathfrak{T}}_T$. The above theorem means that the no-arbitrage condition (NCFL) is determined by the existence of a triple $(\mathcal{H}\times {\mathfrak{T}}_T, \mathbb{Q}, \SSS)$, consisting of a class of controlled portfolios $\mathcal{H}\times {\mathfrak{T}}_T$, an equivalent measure $\mathbb{Q}$ and an $\mathcal{H}\times {\mathfrak{T}}_T$-unbiased rough integrator $\SSS$ under $\mathbb{Q}$.  For the sake of generality of allowed trading portfolios in a market, we would like $\mathcal{H}\times {\mathfrak{T}}_T$ as large as possible. However, the obvious dilemma is that the larger the $\mathcal{H\times {\mathfrak{T}}_T}$, the less likely we can find an equivalent measure $\mathbb{Q}$ such that $\SSS$ becomes an $\mathcal{H}\times {\mathfrak{T}}_T-$unbiased rough integrator under $\mathbb{Q}$. One main question that we study in this paper is whether for different classes of $\mathcal{H}\times {\mathfrak{T}}_T$, there exist $\mathcal{H}\times {\mathfrak{T}}_T-$unbiased rough integrators, and if so, how to classify them. We developed answers to such questions in Section \ref{sec:Unbiasedness}.

As a direct application of Theorem \ref{thm:Pol-unbiased rough integrators}, Theorem \ref{Thm:Main Theorem}, and Theorem \ref{Thm:Main Sharp Theorem}, we have the following immediate classification of (NCFL) condition for Bachelier-type rough market models induced by a rough noise.

\begin{corollary}[(NCFL) in Bachelier-type rough noise market models]\label{cor:NCFL Bachelier noise}
    Let $p\in (1, +\infty],\;q\in [1,+\infty)$ with $\frac{1}{p}+ \frac{1}{q}=1$ and let $\SSS$ be a $2-$dimensional rough market with the riskless asset price given by $S^0\equiv 1$ and the risky asset price realised as a 1-dimensional rough noise $\tilde{\SSS}$, as described in Example \ref{example:S^0=1}. Then,
    \begin{itemize}
        \item This market admits the $\mathrm{(NCFL})$ condition of order $p$ within (Markovian-type) polynomial-induced portfolios if and only if there exists an equivalent measure $\mathbb{Q}$ with $L^q$ Radon-Nikodym derivative such that $\tilde{\SSS}$ is a Gaussian-Hermite rough path under $\mathbb{Q}$. 
        \item This market admits the $\mathrm{(NCFL})$ condition of order $p$ within (Markovian-type) polynomial-induced portfolios, and piecewise signature-induced portfolios if and only if there exists an equivalent measure $\mathbb{Q}$ with $L^q$ Radon-Nikodym derivative such that $\tilde{\SSS}$ is the Hermite rough path lift of a deterministic time change of a Chen-Hermite almost Brownian motion under $\mathbb{Q}$.
        \item This market admits the $\mathrm{(NCFL})$ condition of order $p$ within (Markovian-type) polynomial-induced portfolios, and adaptedly scaled piecewise signature-induced portfolios if and only if there exists an equivalent measure $\mathbb{Q}$ with $L^q$ Radon-Nikodym derivative such that $\tilde{\SSS}$ is the It\^o rough path lift of a deterministic time change of a standard Brownian motion under $\mathbb{Q}$.
    \end{itemize}
\end{corollary}
\begin{proof}
    Recall that the renormalisation terms encoded in $\tilde{\SSS}$ are purely pathwise determined. Thus, if $\tilde{\SSS}$ is a rough noise under the physical measure $\mathbb{P}$, i.e., admits deterministic renormalisation terms under $\mathbb{P}$, then it remains a rough noise under any equivalent measure $\mathbb{Q}$. The statements then follow from the rough Kreps-Yan theorem and our classification of $1-$dimensional unbiased rough noises in Section \ref{sec:Unbiasedness}.
\end{proof}

Similarly, if we allow random renormalisation terms in the risky asset rough path $\tilde{\mathbf{S}}$, Corollary \ref{cor:randomrenorm1}, Corollary \ref{cor:simple to BM with TC}, and Corollary \ref{cor:randomrenorm2} apply and we have the following classification of (NCFL) condition for Bachelier-type rough market models.

\begin{corollary}[(NCFL) in Bachelier-type general rough market models]\label{cor:NCFL TC}
     Let $p,q, \mathbf{S}$ and $\widetilde{\mathbf{S}}$ be as above in Corollary \ref{cor:NCFL Bachelier noise}, except for now letting $\tilde{\SSS}$ be with adapted monotone renormalisation terms (cf. Proposition \ref{Rough Paths via Bell Polynomials}). Moreover, assume either $\alpha\in (\frac{1}{3}, \frac{1}{2}]$, or $\alpha\leq \frac{1}{3}$ and the renormalisation terms encoded in $\tilde{\SSS}$ share a synchronous clock (cf. Subsubsection \ref{subsubsec:Syn}). Then,
     \begin{itemize}
         \item This market admits the $\mathrm{(NCFL})$ condition of order $p$ within (Markovian-type) polynomial-induced portfolios with arbitrary exiting times if and only if there exists an equivalent measure $\mathbb{Q}$ with $L^q$ Radon-Nikodym derivative such that $\tilde{\SSS}$ is a time-changed Gaussian-Hermite rough path until a stopping time under $\mathbb{Q}$.
         \item This market admits the $\mathrm{(NCFL})$ condition of order $p$ within (Markovian-type) polynomial-induced portfolios, and piecewise signature-induced portfolios, if and only if there exists an equivalent measure $\mathbb{Q}$ with $L^q$ Radon-Nikodym derivative such that $\tilde{\SSS}$ is the Hermite rough path lift of an adapted time change of a Chen-Hermite almost Brownian motion until a stopping time under $\mathbb{Q}$.
         \item This market admits the $\mathrm{(NCFL})$ condition of order $p$ within (Markovian-type) polynomial-induced portfolios, and adaptedly scaled piecewise signature-induced portfolios with if and only if there exists an equivalent measure $\mathbb{Q}$ with $L^q$ Radon-Nikodym derivative such that $\tilde{\SSS}$ is the It\^o rough path lift of an adapted time change of a standard Brownian motion until a stopping time under $\mathbb{Q}$.
     \end{itemize}
\end{corollary}
\begin{proof}
    The same arguments as above apply together with Corollary \ref{cor:randomrenorm1} and Corollary \ref{cor:randomrenorm2}. In particular, we now consider stochastic exiting times using the renormalisation clocks as defined in \eqref{eq:renorm clock}.
\end{proof}

\begin{remark}
    In particular, we point out that the necessity of Gaussian marginal and increments of the market noise is no longer a model assumption in our theory, but rather a direct (NCFL) consequence for rough markets.
\end{remark}

We give an example pointing out the necessity of allowing the change of measure in the above theorem and corollaries with the mixed fractional Brownian motions (\cite{Che01}).
\begin{example}
    Let $X_t:=B^t+B^H_t$ be a $1-$dimensional mixed fractional Brownian motion where $(B_t)_{t\in [0,T]}$ is a $1-$dimensional standard Brownian motion and $(B^H_t)_{t\in [0,T]}$ is an independent $1-$dimensional fractional Brownian motion with Hurst parameter $H\in (\frac{3}{4}, 1]$. By Theorem \ref{Thm:Main Theorem}, there is no $\mathcal{H}\times {\mathfrak{T}}_T-$ unbiased rough integrator with $(X_t)_{t\in[0,T]}$ as the underlying noise process. The best we can achieve is to lift it to a Gaussian-Hermite rough path with its variance function $G(t):=-\frac{1}{2}\mathrm{Var}(X_t) = -\frac{1}{2}(t+ t^{2H}),$ such that it becomes unbiased with respect to Markovian-type integrands. However, it is well-known that $(X_t)_{[0,T]}$ is a semimartingale for $H\in (\frac{3}{4}, 1]$. More specifically, it is equivalent to a standard Brownian motion in that case (cf. \cite{Che01} Theorem 1.7) under an equivalent measure $\mathbb{Q}$. Thus, we might construct the rough noise lift
    \begin{equation}
        \XXX_{s,t}:=(1,X_{s,t}, H_2(X_{s,t}, -\frac{1}{2}(t-s)))\quad\quad \forall s,t\in [0,T]
    \end{equation}
    together with its deterministic quadratic variation path as the renormalisation term. Then, the law of $\XXX$ under the equivalent measure $\mathbb{Q}$ coincides with the It\^o rough path lift of a standard Brownian motion. Thus, it still admits the (NCFL) condition.
\end{example}

Before moving forward to the next subsection, giving a classification of the (NCFL) condition for rough-path extensions of classical SDE models, we give two final remarks on our rough path setting.

\begin{remark}
Recall from the definition of piecewise controlled paths that we do not allow jumps in the portfolios, except for the exit time. Except for the reason we explained in the previous subsection, that this will exclude classical ill-posed strategies, e.g. the doubling strategy and suicide strategy, there is one more reason for that. If we allow jumps in trading, since the $k-$tuple $(1,0,...,0)$ is always a controlled path for any rough path $\XXX$, simple strategies will also be included in the setting, and the gain process will be nothing else than the classical discrete gain process. This drives us back to the original Kreps-Yan settings, where if \eqref{eq:Rough Kreps-Yan} holds for all simple strategies, the equivalent measure $\mathbb{Q}$ will turn a martingale measure for the price process $S$ (cf. Theorem 5.2.2 in \cite{DS06}), and we will never be able to break through the semimartingale settings. Since the aim of this paper is to detect under maximal constraints, whether rough path renormalisation carries no-arbitrage conditions outside the semimartingale settings, we would like to exclude simple strategies in the setting to avoid obvious answers to this question.
\end{remark}

\begin{remark}
    We make some final comments on why we chose (NCFL) as the no-arbitrage condition in our settings, rather than develop a rough path version of the Fundamental Theorem of Asset Pricing. Since we are in a frictionless setting, we compare with the original version of the Fundamental Theorem of Asset Pricing for semimartingales in \cite{DS94}, where No Free Lunch with Vanishing Risk (NFLVR), a stronger no-arbitrage condition compared with (NFL), was considered. Roughly speaking, one replaces the weak-star topology in (NFL) with the strong topology for (NFLVR) such that an arbitrage can be characterised via a sequence of replicable contingent claims. Indeed, in the proof in \cite{DS94}, there were certain stochastic tricks for semimartingales and stochastic integrals that cannot be naively carried over in the random rough paths setting. The only additional part of the original proof, compared with the Kreps-Yan theorem, is to argue that under the (NFLVR) condition, the strong closure of the gain cone of contingent claims dominated by some replicable claim coincides with the weak-star closure. This extremely technical part relies strongly on the Memin theorem (cf. \cite{Mem80}), which states that the space of stochastic integrals against a fixed semimartingale, viewed as stochastic processes, is closed with respect to the Emery semimartingale metric. However, it is unrealistic to expect such results to hold for general stochastic processes in the random rough path setting since, on the one hand, there would be no natural replacement of the Emery semimartingale metric nor an analogy of a Doob decomposition for general stochastic processes, while on the other hand, there is no ready notion of predictable tightness in the rough path setting and the rough path topology (cf. \cite{FH14} Chapter 2) is too strong to manage the convergence of a sequence.
\end{remark}


\subsection{Rough differential equation models}\label{subsec:RDE models}
In practice, one usually models the stock price as the solution of a stochastic differential equation
describing certain dynamics, where the driving noises are some possibly correlated Brownian motions. Thus, it is natural to discuss the possibility of extending those models to RDE models driven by general rough noises. Now we consider a price process that is driven by an RDE of form
\begin{equation}\label{eq:RDE model RDE}
    \mathrm{d} Y_t=f\left(Y_t\right) \mathrm{d} \XXX _t, \quad Y_0=y_0
\end{equation}
as in Definition \ref{def:RDE}. Before lifting the price process to random rough paths and looking into the gain process, we list some classical market models that can be described as RDE models.

\begin{example}[The Black-Scholes model]
    Let $X_t:=(t, B_t)$ where $(B_t)_{t\in [0,T]}$ is a standard Brownian motion, we define the second order increment of the rough path $\XXX$ via
    \begin{equation}\label{eq:time-augmented Ito lift for BM}
        \XX_{s,t}:= 
        \begin{pmatrix}
            \frac{1}{2}(t-s)^2 & \int_s^t (u-s)\mathrm{d}B_{u}\\
            \int_s^t B_{s,u} \mathrm{d}u & \int_s^t B_{s,u} \mathrm{d}B_u = \frac{1}{2}B_{s,t}^2 - \frac{1}{2}(t-s)
        \end{pmatrix}
        \quad\quad\quad
    \forall s,t\in \Delta_{[0,T]},
    \end{equation}
    where the integrals against the Brownian motions are It\^o integrals. Moreover, by setting $f(t,x): = (\mu(t)\cdot x, \sigma(t)\cdot x)$ with some deterministic drift and volatility functions, \eqref{eq:RDE model RDE} essentially becomes the RDE version of the Black-Scholes SDE. In particular, the solution in the sense of SDE and RDE coincides by \cite{FH14}, Theorem 9.1. 
\end{example}

\begin{example}[Local volatility models]
    With the same choice of $\XXX$ as above, and let $f(t,x): = (\mu(t)\cdot x, \sigma(t,x)\cdot x)$ with some deterministic drift and volatility functions, local volatility models in Black-Scholes form are recovered. By deleting the time-augmentation in $\XXX$ and setting $f(t,x)=\sigma(t,x)$ for some volatility function, one recovers the local volatility models in the Bachelier form.
\end{example}

\begin{example}[Stochastic volatility models]\label{exp:SVM}
    It is possible to encode stochastic volatility into the rough path $\XXX$. We let $X_t:=(t, B^1_t, B^2_t)$ where $(B^1_t)_{t\in [0,T]}$ and $(B^2_t)_{t\in [0,T]}$ are two possibly correlated standard Brownian motions. One can lift $X$ to It\^o-type rough paths via the It\^o integrals in a similar manner as in \eqref{eq:time-augmented Ito lift for BM}. Moreover, for any $t,x,v$ define $f(t, x,v)\in \mathcal{L}(\mathbb{R}^3, \mathbb{R}^2)\cong  \mathbb{R}^{2\times3}$ via
    \begin{equation}
        f(t, x, v):=\begin{pmatrix}
            \mu(t)\cdot x & v\cdot x & 0\\
            \lambda(t) & 0& \sigma(t)\cdot v 
        \end{pmatrix}.
    \end{equation}
    One readily recovers stochastic volatility models of the form 
    \begin{equation}
        \begin{cases}
            \mathrm{d}S_t = \mu(t)\cdot S_t \mathrm{d} t + V_t\cdot S_t \mathrm{d}B^1_t\\
             \mathrm{d}V_t = \lambda(t) \mathrm{d} t + \sigma(t)\cdot V_t\mathrm{d}B^2_t.
        \end{cases}
    \end{equation}
    The Bachelier form of these models can also be described via RDEs in a similar manner as above.
\end{example}

\begin{example}[Rough volatility models]\label{exp:RVM}
    One can also encode rough volatility (cf. \cite{GJR18} and \cite{BFG16}) in the rough path $\XXX$. Since in rough volatility settings, the Hurst parameter of the fractional Brownian motion is usually set to be around $ 0.1$, let us assume in this example that $\alpha >0$ can be arbitrarily small. For simplicity, we consider the Bachelier form. We let $X_t:=(B^H_t, B_t)$, where $B^H$ is a fractional Brownian motion with Hurst parameter $H$ and $B$ is a standard Brownian motion. A joint It\^o rough path lift of $X$ purely via It\^o calculus was constructed in \cite{BFGJ24} via some algebraic extension on the tensor algebra. Here, we give an intuitive sketch of their constructions inductively.

    Let $k:=\lfloor\frac{1}{\alpha} \rfloor$ and $w=w_1\cdots w_n$ be a word with letters in $\{1,2\}$ of length $n\leq k$. For $n=1$, set 
    \begin{equation}
        \langle \XXX_{s,t}, 1\rangle = B_{s,t} \quad \quad \text{and} \quad\quad \langle \XXX_{s,t}, 2\rangle = B^H_{s,t} \quad\quad\forall s,t\in \Delta_{[0,T]}.
    \end{equation}
    For a word containing only the letter $2$, we simply set $\langle \XXX_{s,t}, 2\cdots 2\rangle$ to be the corresponding level of a $1$-dimensional rough path lift of $B^H$. As for the crossed integrals, we set inductively
    \begin{equation}
        \langle \XXX_{s,t}, w1\rangle:=\int_s^t \langle \XXX_{s,u}, w\rangle \mathrm{d}B_u \quad\quad\forall s,t\in \Delta_{[0,T]}
    \end{equation}
    via It\^o calculus; let $i$ be the largest index such that $w_i = 1$, we set
    \begin{equation}
    \begin{aligned}
        \langle \XXX_{s,t}, w2\rangle:=  &\langle \XXX_{s,t}, w_1\cdots w_i\rangle \cdot \langle\XXX_{s,t}, w_{i+1}\cdots w_n 2 \rangle \\ &- \int_s^t \langle\XXX_{s,u}, w_{i+1}\cdots w_n 2 \rangle \mathrm{d}\langle \XXX_{s,u}, w_1\cdots w_i\rangle \quad\quad\forall s,t\in \Delta_{[0,T]},
    \end{aligned}
    \end{equation}
    where the integral is viewed as an It\^o integral against a standard Brownian motion since $w_i =1$. The construction is essentially via a postulated integration by parts between the two components $B$ and $B^H$. It shall be expected that $\XXX$ is a valid rough path, since all our constructions are via valid integrals or some integration by parts property. We leave the verification of this fact to our readers, or see \cite{BFGJ24} Theorem 1.5  and Proposition 2.1 for an algebraic proof in a more general setting and a concrete formula of the construction above. One can then use \eqref{eq:RDE model RDE} but for $\alpha>0$ smaller to rewrite a large class of rough (even possibly path-dependent) volatility models. We refer to \cite{BFGJ24} Section 2.1 for a detailed discussion on the examples. See also \cite{FT24} for a partial rough path space construction for adapting rough volatility models in the context of rough paths.
\end{example}

Obviously, the above constructions shall work for general noises, as long as the noise for the price process and the noise for the volatility process together admit a well-defined random rough path lift. In practice, people model the noise for the price process with Brownian motion because on the one hand, it has the nice martingale property, which is strong enough to build no-arbitrage theorems; while on the other hand, It\^o calculus makes the computation with it fairly reachable. In the rough path/RDE setting, we are indeed able to make the computation reachable for a much wider class of proper noises, e.g. continuous martingales and non-Brownian Gaussian processes. We are thus interested in the question of which are the proper noises for the price process in the RDE models, for the sake of no-arbitrage, understood as the (NCFL) condition \`a la Kreps-Yan.


\subsubsection{The gain process}
To answer the question above, we need to study the gain process of a controlled portfolio in an RDE model. However, a technical barrier here is that the solution of an RDE, as the price process, is not a rough path immediately, and we do not have an integration theory against controlled paths a priori to define the gain process. Luckily, this is not an issue for $\alpha\in (\frac{1}{3}, \frac{1}{2}]$. We review some consistency results between controlled paths and rough paths in Appendix \ref{app:Consistecy}, where we shall be able to define the gain process properly. Similar to classical stochastic calculus, we are able to write the gain process as a rough integral against the underlying rough noise driving the market dynamics from a rough integral against the price process itself.

For simplicity, we consider the generalised RDE framework of the examples above in Bachelier form. Let $X_t:=(M^1_t, M^2_t)$ be some $2-$dimensional market noise process, where $M^1_t$ ($M^2_t$ respectively) is understood as the noise process in the price dynamics (volatility dynamics, respectively), and $\XXX:=(X, \XX)$ be a random rough path lift of $X$. Let $f\in \mathcal{L}(\mathbb{R}^2, \mathbb{R}^2)$ be a $C^2$ function of the form
    \begin{equation}
        f(s,v):=\begin{pmatrix}
            f_1(s,v) & 0\\
            0 & f_2(s,v)
        \end{pmatrix}
    \end{equation}
describing the market dynamics and $Y_t:=(S_t, V_t)$ together with some Gubinelli derivative $Y'_t \in \mathbb{R}^{2\times 2}$ be the solution to the RDE \eqref{eq:RDE model RDE}. Then $(S_t, S'_t)$ is a by-$M^1$-controlled path thanks to \eqref{1-18} and Proposition \ref{Rough Integral} by setting $S'_t:=f_1(S_t, V_t)$. By noticing that $\mathbf{M}^1:=(M^1, \mathbb{M}^1:=\langle \XX, 11\rangle)$ is a rough path lift of the price process noise $M^1$, we are now able to lift $(S, S')$ to a rough market model.

\begin{definition}\label{def:RDE model}
    In the above setting, we call the rough market model $\SSS:=(S, \Ss)$ the corresponding RDE model, where $\Ss$ is defined via:
    \begin{equation}
        \Ss_{s,t}:= \lim _{|\mathcal{P}| \downarrow 0} \sum_{[u,v]\in \mathcal{P}} (S_{s,u} \otimes S_{u, v}+S_u^{\prime} \otimes S_u^{\prime} \mathbb{M}^1_{u, v}) \quad\quad \forall s,t\in \Delta_{[0,T]}
    \end{equation}
    as in Corollary \ref{CPasRP}. 
\end{definition}

\begin{remark}\label{rem:multi-dim RDE model}
    Multi-dimensional RDE price processes driven by multi-dimensional market noises can be obtained in a similar manner.
\end{remark}

Given any controlled portfolio $(Y, Y')$ in the rough market model $\SSS$, we first check that the path $Y$ is also controlled by the underlying noise $M^1$. Indeed,
\begin{equation}\label{eq:Transitivity of CP}
    \begin{aligned}
        |Y_{s,t} -Y'_sS'_sM^1_{s,t}| &\approx |Y_{s,t}-Y'_s(S_{s,t}+O(|t-s|^{2\alpha}))|\\
        & \leq|Y_{s,t}-Y'_sS_{s,t}| +O(|t-s|^{2\alpha}) \lesssim O(|t-s|^{2\alpha}) \quad\quad\forall s,t\in \Delta_{[0,T]},
    \end{aligned}
\end{equation}
where we used the controlledness of $(S, S')$ in the first line and the controlledness of $(Y, Y')$ in the last inequality. Then, the gain process is then given by
\begin{equation}\label{eq:RDE model gain process}
\begin{aligned}
    G_t(Y) &=\int_0^t (Y_s,Y'_s)\mathrm{d}\SSS_s \\
    & = \int_0^t (Y_s,Y'_sS'_s)\mathrm{d}(S_s, S'_s)\\
    & = \int_0^t (Y_sf_1(S_s, V_s), Y'_sS'_sf_1(S_s, V_s)+Y_sf'_1(S_s, V_s))\mathrm{d}\mathbf{M}^1_s \quad\quad \forall t\in [0,T],
\end{aligned}
\end{equation}
where the second line is understood as the integral of a controlled path against another controlled path, as explained in Proposition \ref{1.2.9}. In the above equation, the second equality is due to Proposition \ref{Consistency} and the definition of an RDE that $S$ must satisfy \eqref{1-18}, while the third equality is due to Proposition \ref{Ass}.

\begin{remark}
    In RDE models, one can slightly extend the notion of controlled portfolios defined in Definition \ref{def:Controlled Portfolios}. Indeed, let $(Z,Z')$ be any path controlled by the noise process, including both the price and volatility noise. One might consider it a legit controlled portfolio, whose gain process is defined by
    \begin{equation}
        G_t(Z) := \int_0^t (Z_s,Z'_s)\mathrm{d}(S_s, S'_s)
    \end{equation}
    as the integral of a controlled path against another controlled path, as explained in Proposition \ref{1.2.9}. By similar arguments as in \eqref{eq:Transitivity of CP}, paths that are controlled by either the price process or the volatility process automatically fall into this larger class of portfolios. Then, similar to \eqref{eq:RDE model gain process}, the gain process can be written as a rough integral against the whole rough noise $\XXX$. This allows making full use of market information. Similar ideas were also discussed in \cite{DJ23}.
\end{remark}


\subsubsection{(NCFL) for RDE models}

Now we may discuss the NCFL condition in RDE models. Thanks to Theorem \ref{RoughKYTheorem}, the (NCFL) condition of an RDE model can be immediately reduced to an unbiasedness condition for the underlying rough noise driving the price process:

\begin{corollary}[NCFL for RDE models driven by rough noises]\label{cor:NCFL RDE}
    Let $p\in (1, +\infty],\;q\in [1,+\infty)$ with $\frac{1}{p}+ \frac{1}{q}=1$ and $\SSS$ be a Bachelier RDE model as above. Let $\mathcal{H}$ be a linear and convex set of controlled portfolios containing $0$, and $\mathfrak{T}_T$ be a set of stopping times bounded by $T$ and containing $T$. Then $\SSS$ satisfies (NCFL) of order $p$ within $\mathcal{H}\times\mathfrak{T}_T$, if and only if there exists an equivalent measure $\mathbb{Q}$ such that 
    \begin{equation}
        \frac{\mathrm{d}\mathbb{Q}}{\mathrm{d}\mathbb{P}}\in L^q_+ \quad \text{and} \quad \mathbb{E}_{\mathbb{Q}}[\int_0^\tau (Y_tf_1(S_t, V_t), Y'_tS'_tf_1(S_t, V_t)+Y_tf'_1(S_t, V_t)) \mathrm{d} \mathbf{M}^1_t] =0 \quad\forall Y\in \mathcal{H},\; \forall \tau \in \mathfrak{T}_T.
    \end{equation}
\end{corollary}

The (NCFL) condition can then be directly classified as a consequence of Theorem \ref{thm:Pol-unbiased rough integrators}, Theorem \ref{Thm:Main Theorem}, and Theorem \ref{Thm:Main Sharp Theorem}. We formulate only a short corollary here as a direct consequence of the strongest Theorem \ref{Thm:Main Sharp Theorem}.

For brevity, write
\begin{equation}
\Phi(Y)
:=
\left(
Y f_1(S,V),
Y'S'f_1(S,V)+Yf'_1(S,V)
\right).
\end{equation}
Thus, the condition above is precisely the unbiasedness of
\(\mathbf M^1\) against the induced class \(\Phi(\mathcal H)\).
Recall that
$\mathcal{H}_{\mathbf M^1}^{\#}
=
\operatorname{Span}
\left(
\mathcal{H}_{\mathbf M^1}^{\mathrm{Pol}}\cup \mathcal{H}_{\mathbf M^1}^{\mathrm{AdpSig}}
\right).$ Taking \(\mathfrak{T}_T=[0,T]\), Theorem \ref{Thm:Main Sharp Theorem} gives the following immediate consequence.

\begin{corollary}[Brownian collapse for RDE models]\label{cor:RDE Brownian collapse}
Let $p,q, \SSS$ be as above. Assume that the price-driving rough noise \(\mathbf M^1\) is a rough noise in the sense of Definition \ref{def:rough noise} with deterministic renormalisation terms. Consider a portfolio class \(\mathcal H\) such that $\Phi(\mathcal H)=\mathcal{H}^\#_{\mathbf M^1}$. Then $\SSS$ satisfies NCFL of order $p$ within $\mathcal{H}\times [0,T] $ if and only if there exists an equivalent measure $\mathbb Q$, with
$\frac{d\mathbb Q}{d\mathbb P}\in L^q_+$, under which
$\mathbf M^1$ is the Itô rough path lift of a deterministic time
change of Brownian motion.
\end{corollary}

Similarly, corollaries for RDE market models can be made for $\mathbf{M}^1$ with monotone random renormalisation terms as well, as consequences of Corollary \ref{cor:randomrenorm1}, Corollary \ref{cor:simple to BM with TC}, and Corollary \ref{cor:randomrenorm2}. We omit the details.

\begin{remark}
    The above discussions extend canonically to RDE models where price and volatility have common noises. Then, the (NCFL) condition will be equivalent to the unbiasedness of the rough noises driving the price process in a similar manner as above.
\end{remark}

\begin{remark}
    One can derive the same results for RDE models in Black-Scholes form, where the noise process $\tilde{X}_t:=(t,M^1_t, M^2_t)$ includes also a time augmentation. We sketch the arguments. Let $\tilde{\XXX}:=(\tilde{X}, \tilde{\XX})$ be some rough path lift. Indeed, the gain process of a controlled portfolio can be written as a rough integral against $\tilde{\mathbf{M}}^1$, the rough path lift of $\tilde{M}^1_t:=(t, M^1_t)$ encoded in $\tilde{\XXX}$ and then an analogy of Corollary \ref{cor:NCFL RDE} can be obtained for $\tilde{\mathbf{M}}^1$. Moreover, one can show that the rough integral against $\tilde{\mathbf{M}}^1$ is essentially a rough integral against $\mathbf{M}^1$ plus a Riemann-Stieltjes integral against $t$. By encoding the drift into the noise, one can rewrite the gain process as a rough integral against a single rough noise $\mathbf{M}^{1,drift} = (M^{1,drift}, \mathbb{M}^{1,drift})$. In particular, if $M^1$ is additionally Gaussian with some mild condition, one is able to find an equivalent measure $\mathbb{Q}$ such that $\text{Law}_{\mathbb{P}}(\mathbf{M}^1) = \text{Law}_{\mathbb{Q}}(\mathbf{M}^{1,drift})$ via the seminal Cameron-Martin theorem. In the case of a standard Brownian motion, this is done via Girsanov's transform. Thus, independent of the drift, the NCFL for the RDE model in Black-Scholes form is again equivalent to the unbiasedness of $\mathbf{M}^1$, the rough noise driving the price process, under some equivalent change of measure. We refer to \cite{QX18} for a discussion in the case $W^1=B^H$, a fractional Brownian motion with $H\in (\frac{1}{3}, \frac{1}{2}]$.
\end{remark}


\section{A Pathwise Arbitrage in Geometric Rough Markets}\label{sec:Arbitrage}

So far, we have studied the no-arbitrage of a rough market in the sense of (NCFL) a la Kreps-Yan. A free lunch in such settings lies in the weak-star closure of the set of contingent claims that are replicable by controlled portfolios. It is not clear whether a free lunch exists in this set itself, or, in other words, whether a controlled arbitrage in the sense of Definition \ref{def:Arbitrage} exists. We could not construct an arbitrage for every rough market that is not included in Theorem \ref{Thm:Main Theorem}. However, a classical arbitrage portfolio from \cite{HPS84} would apply if the rough path lift of the price process $S =(S^0,...,S^d)$ is geometric. Previously, this strategy was proven to yield arbitrage opportunities in markets whose sample paths are of bounded variation and bounded $p$-variation for some $p\in [1,2)$, as demonstrated in \cite{HPS84} and \cite{Sal98}, respectively. We adapt it to the rough path setting in this section. The validity of this arbitrage relies essentially on a chain rule of the integral used in the gain process.

We fix a $(d+1)-$dimensional $\alpha-$H\"older continuous price path $S =(S^0,...,S^d)$ throughout this section and let $k:=\lfloor \frac{1}{\alpha} \rfloor$.

\begin{theorem}\label{WGArbitrageTHM}
    Given a rough market model $\SSS:= (1, S, \Ss^2,...,\Ss^k)$ associated with the rough price process $S$, there exists a controlled arbitrage in the sense of Definition \ref{def:Arbitrage}, if $\SSS$ is geometric in the sense of Definition \ref{Geometric Rough Paths}.
\end{theorem}

For simplicity of notation, we assume without loss of generality that $S^e_0=1$ for any $e=0,...,d$. Otherwise, we can rescale the investment in each asset. We note that we do not fix $S^0$ to be constant or of bounded variation here.

For any $p\in \mathbb{R}$, we set 
\begin{equation}
    M^p_t:=\left\{\frac{1}{d+1} \sum_{e=0}^d\left(S^e_t\right)^p\right\}^{1 / p}
\end{equation}
to be the $p$-th order mean of the prices and define the $p$-portfolio $Y:=(Y^{(1)},..., Y^{(k)})$ with $Y^{(1)}:=(Y^{(1);0},...,Y^{(1);d})$ where
\begin{equation}
    Y^{(1);e}_t=M^p_t \frac{\left(S^e_t\right)^{p-1}}{\sum_{m=0}^d\left(S^m_t\right)^p}=: F(S_t)^e
\end{equation}
as a smooth function of the underlying price process $S$, and $Y^{(i)}_t:=D^{(i-1)}F(S_t)$ as described in Example \ref{Function of Geometric Rough Paths as Controlled Paths} for $i=2,...,k$.

We now first prove that the above portfolio is self-financing in the sense of Definition \ref{def:Self-financing}. 

\begin{proposition}\label{Prop:Arbitrage is self financing}
    In the above settings, we have for any $t$:
    \begin{equation}
        M^p_t = M_0^p + \int_0^t Y_s \mathrm{d} \SSS_s.
    \end{equation}
    In particular, the portfolio is self-financing.
\end{proposition}

We need a change of variable formula for $d-$dimensional geometric rough paths as noted in Remark \ref{rem:Change of variable}. The result is well-known, so we only cite it.

\begin{lemma}\label{lemm:Change of variable for GRP}
    Let $X:[0,T] \rightarrow \mathbb{R}^d$ be an $\alpha$-H\"older continuous path and $\XXX:= (1, X, \XX^2,..., \XX^k)$ be a geometric rough path lift with $k:=\lfloor \frac{1}{\alpha}\rfloor$. Now for any smooth function $F: \mathbb{R}^d \rightarrow \mathbb{R}$, we have 
    \begin{equation}
F\left(X_t\right)-F\left(X_s\right)=\int_s^t (D F\left(X_r\right), D^{(2)}F\left(X_r\right),..., D^{(k)}F\left(X_r\right)  )\mathrm{d} \XXX_r, \quad\quad\forall 0\leq s,t \leq T,
\end{equation}
where the integrand is a controlled path.
\end{lemma}
\begin{proof}
    See e.g. \cite{Kel12} Theorem 5.2.16 by setting the bracket extension therein to be trivial.
\end{proof}

\begin{proof}[Proof of Proposition \ref{Prop:Arbitrage is self financing}]
    First, set $G: \mathbb{R}^{d+1}\rightarrow \mathbb{R},\; (x_0,...,x_{d+1})\mapsto (\sum_m (x_m)^p)^{\frac{1}{p}}.$ Then one has $DG(x_0,...,x_d)= \{(\sum_m (x_m)^p)^{\frac{1}{p}-1}\cdot (x_e)^{p-1}\}_{e=0,...,d}$, which is again a smooth function. Thus, with some abuse of notation on rough integrals, we have
    \begin{equation}
        \int_0^t Y_s \mathrm{d} \SSS_s= (\frac{1}{d+1})^{\frac{1}{p}}\int_0^T DG(S_s^0,...,S_s^d) \mathrm{d} \SSS_s= (\frac{1}{d+1})^{\frac{1}{p}} (G(S_t)- G(S_0))= M^p_t-M^p_0
    \end{equation}
    where the middle equality is due to Lemma \ref{lemm:Change of variable for GRP}. Moreover, it is straightforward to verify that the portfolio value process is exactly $M^p_t$ and therefore the portfolio is self-financing. 
\end{proof}

To construct an arbitrage based on such portfolios, we cite the following inequality about the $p$-order mean:
\begin{lemma}[$p$-th order mean]\label{PowerMeanIneq}
    For a finite number of non-negative numbers $a_1$,..., $a_n$ and $p\in \mathbb{R}$, we define their $p-$th order mean as 
    \begin{equation}
        M^p(a):=\left\{\frac{1}{n} \sum_{k=1}^n\left(a_k\right)^p\right\}^{1 / p}.
    \end{equation}
    For $-\infty< p<q< \infty$, we have $M^p(a)\leq M^q(a)$, where equality holds only when $a_1=a_2=\cdots =a_n$.
\end{lemma}
\begin{proof}
    See \cite{HLP59} Theorem 16.
\end{proof}

Now, we consider the following strategy: Let $q>p$. We buy $\frac{1}{d+1}$ shares of each security at time zero and manage this portfolio according to the $q$-portfolio, and sell $\frac{1}{d+1}$ shares of each security at time zero and manage these short positions according to the $p$-portfolio. This yields a profit of $M^q(t)-M^p(t)>0$ (cf. Lemma \ref{PowerMeanIneq}) at any time $t$ with no investment. Formally, the portfolio is a controlled path as the difference between two controlled paths. Moreover, since the gain is positive at any time, this portfolio is indeed an arbitrage, which finishes the proof of Theorem \ref{WGArbitrageTHM}. In particular, we point out that this is a pathwise arbitrage strategy.

Finally, we derive a Lipschitz-H\"older version of the results in \cite{HPS84} and \cite{Sal98} as a corollary of Theorem \ref{WGArbitrageTHM}.

\begin{corollary}
    If the sample paths of the price process $S$ is $\alpha$-H\"older continuous for some $\alpha\in (\frac{1}{2}, 1]$, and the integral used in the gain process is the corresponding Riemann-Stieltjes integral against $S$, the above constructed strategy is an arbitrage in the market.
\end{corollary}
\begin{proof}
    Since $S$ is $\alpha$-H\"older continuous for some $\alpha\in (\frac{1}{2}, 1]$, it is also of bounded $p-$variation for some $p\in [1,2)$, one might define $\Ss_{s,t}:= \int_s^t S_{s,u}\mathrm{d}S_u$ as the Riemann-Stieltjes integral and then $\SSS:=(S,\Ss)$ is a well-defined rough path. We notice that the following compensation term in the rough integral $\int_0^T (Y, Y') \mathrm{d} \SSS$
    \begin{equation}
        Y'_{s}\Ss_{s,t}\sim O(|t-s|^{2\alpha})\quad\quad \forall 0\leq s\leq t\leq T
    \end{equation}
    since $\Ss$ is $2\alpha$-H\"older. Since $2\alpha>1$, we have 
    \begin{equation}
        \lim _{|\mathcal{P}| \downarrow 0} \sum_{[s, t] \in \mathcal{P}}Y_s^{\prime} \mathbb{S}_{s, t} = 0
    \end{equation}
    and thus 
    \begin{equation}
        \int_0^T (Y, Y') \mathrm{d} \SSS = \lim _{|\mathcal{P}| \downarrow 0} \sum_{[s, t] \in \mathcal{P}}\left(Y_s S_{s, t}+Y_s^{\prime} \mathbb{S}_{s, t}\right) = \lim _{|\mathcal{P}| \downarrow 0} \sum_{[s, t] \in \mathcal{P}} Y_s S_{s, t} = \int_0^T Y \mathrm{d} S,
    \end{equation}
    i.e., the rough integral coincides with the Riemann-Stieltjes integral. Thus, the gain process coincides with the one we define above. Since Riemann-Stieltjes integral admits a chain rule/ integration by parts, the rough path $(S, \Ss)$ is indeed geometric by the discussions in Remark \ref{rem:GRP}. By our discussion above, our strategy is an arbitrage in such markets.
\end{proof}

The above construction also has the following immediate consequence for RDE models:

\begin{corollary}
   Let $\alpha \in (\frac{1}{3}, \frac{1}{2}]$ and $\SSS$ be an (d+1)-dimensional RDE model for $d>0$ as in Definition \ref{def:RDE model} and Remark \ref{rem:multi-dim RDE model}. If the rough path lift of the price process $M^1$ encoded in the rough market noise $\XXX$ is geometric, then the above construction yields a controlled arbitrage in the rough market model $\SSS$.
\end{corollary}
\begin{proof}
    As discussed in Definition \ref{def:RDE model}, $\SSS$ is the rough path lift of $(S,S')$ as a controlled path by $(M^1, \mathbb{M}^1)$ via Proposition \ref{CPasRP}, where $(M^1, \mathbb{M}^1)$ is the rough path lift of the price process $M^1$ encoded in the rough market noise $\XXX$. If $(M^1, \mathbb{M}^1)$ is geometric, then so is $\SSS$ by Proposition \ref{WGpreserved}. The statement then follows from Theorem \ref{WGArbitrageTHM}.
\end{proof}

\begin{remark}
    We refer to \cite{HPS84} Section II for some economic heuristics for such portfolios. Roughly speaking, as $q \rightarrow \infty$, the associated strategy becomes closer to holding only the highest-priced security (equal division when there is more than one security at the maximum). This fact lends credibility to such a strategy of ``going with the leader," while a chain rule provides, in mathematics, the desired self-financing condition of such portfolios. 
\end{remark}


\section{Conclusion and Further Work}\label{sec:Conclusion}

In this paper, we conducted a fundamental investigation into the capacity of rough path theory to serve as a foundation for continuous, frictionless, No Controlled Free Lunch (NCFL) financial market models. It turns out that by allowing different types of portfolios and exiting times, the restriction on the admissible rough noise becomes severe. Our theory demonstrates again that the semimartingale property is not an artefact of Itô's calculus, but rather a robust and fundamental consequence of the absence of arbitrage in a frictionless environment. The power of rough path theory, in this context, lies not in creating new arbitrage-free models but in providing a definitive proof of their non-existence. We note that simple strategies were avoided in developing our theory. In the context of rough noise, i.e., random rough paths with deterministic renormalisation terms (cf. Proposition \ref{Rough Paths via Bell Polynomials}), a brief summary is given by:

\vspace{1em}

\noindent\resizebox{\textwidth}{!}{%
\begin{tabular}{|p{3.5cm}|p{3cm}|p{6.5cm}|p{2.5cm}|}
\hline
\textbf{Admissible Portfolios ($\mathcal{H}$)} & \textbf{Exiting Times ($\mathfrak{T}_T$)} & \textbf{Required Structure of Rough Noise $\XXX$ (Unbiased Integrator)} & \textbf{Key Result} \\
\hline
Polynomial-induced Markovian-type  portfolios ($\mathcal{H}^{\text{Pol}}$) & No early exiting $\mathfrak{T}_T:=\{T\}$ & \textbf{Gaussian-Hermite Rough Path:} $X$ has Gaussian marginals; lift is defined by Hermite polynomials with $G^2(t) = -\frac{1}{2}\text{Var}(X_t)$. & Theorem \ref{thm:Pol-unbiased rough integrators} \\
\hline
Polynomial-induced + Signature-induced portfolios $\mathrm{Span}$($\mathcal{H}^{\text{Pol}} \cup \mathcal{H}^{\text{pSig}}$) & Any deterministic exit time $\mathfrak{T}_T:=[0,T]$ & \textbf{Hermite Lift of a Time-Changed Chen-Hermite Almost Brownian Motion:} $X$ has Brownian-like increments and satisfies the Chen-Hermite balancing condition. & Theorem \ref{Thm:Main Theorem} \\
\hline
Polynomial-induced + adaptedly scaled Signature-induced portfolios ($\mathcal{H}^\# := \mathrm{Span}(\mathcal{H}^{\text{Pol}} \cup \mathcal{H}^{\text{AdpSig}}$))  & Any deterministic exit time $\mathfrak{T}_T:=[0,T]$ & \textbf{Itô Lift of a Time-Changed Standard Brownian Motion:} $X$ has independent Gaussian increments. & Theorem \ref{Thm:Main Sharp Theorem} \\
\hline
\end{tabular}
}

\vspace{1em}
\noindent Similar classification results were made for random rough paths with monotone random renormalisation terms (and a synchronous clock) in Subsection \ref{subsec:Random renorm}.


Building on this foundational work, several promising avenues for future research emerge. First, as conjectured in Conjecture \ref{conjecture}, it would be interesting to distinguish a Chen-Hermite almost Brownian motion and a $1-$dimensional standard Brownian motion. Second, it is still worth discussing whether rough paths and rough integrals allow more general rough noises in modelling markets with frictions. For instance, it would be interesting to discuss no-arbitrage in markets with delayed information in the framework of delayed rough integrals introduced in \cite{NNT08}. Third, our current framework contains only continuous random rough paths with deterministic renormalisation terms and cannot deal with jumps yet. In the c\`adl\`ag scenario, other normal martingales, e.g. an It\^o lift of the compensated Poisson process $X_t:= N_t-t$, should also be admissible, where clearly the current constraints that renormalisation terms must be deterministic must be removed, since the pathwise quadratic variation is $N_t$, a pure jump process. Developing a similar theory for random c\`adl\`ag rough paths (cf. \cite{FS17}) shall yield a more comprehensive answer to the possibility of frictionless market modelling with rough paths.  Finally, as mentioned, our framework is not enough to explain RDEs driven by non-geometric $\alpha-$H\"older continuous rough paths for $\alpha\leq \frac{1}{3}$. The barrier is that the consistency between controlled paths and rough paths cannot be restored in that case. Passing to the \emph{branched} setting (cf. \cite{Gub10}) restores compatibility and, simultaneously, enlarges the admissible renormalisation class for non-geometric lifts of a fixed path. On the market side, extending to branched drivers is natural for the gain process of signature-based portfolios in signature-based models (cf. \cite{CGS23}): the gain process then contains tree-level terms such as $\int_0^t \mathbb{X}^{(n)}_{0,u}\,\mathrm{d}\mathbb{X}^{(m)}_u$, which are not canonical in our current tensor algebra setting but arise as basis elements of the Connes–Kreimer algebra in the setting of \cite{Gub10}. A key next step is to analyse the (NCFL) condition in this enlarged class and to classify  \emph{unbiased branched/c\`adl\`ag rough integrators} relative to the same families of admissible strategies.

\appendix

\section{Some proofs}\label{app:Proofs}

We collect some proofs in this Appendix. 
\subsection{Proof of Proposition \ref{Rough Integral}}\label{App:Rough Integral}
We begin with the proof of Proposition \ref{Rough Integral}, which requires Gubinelli's seminal sewing lemma \cite{Gub04}. We cite the version of \cite{FH14} Lemma 4.2 and refer the readers to a proof therein as well. Let $W$ be a Banach space. We introduce the space $C_2^{\alpha, \beta}([0, T], W)$ of functions $\Xi$ from the 2 -simplex $\{(s, t): 0 \leq s \leq t \leq T\}$ into $W$ such that $\Xi_{t, t}=0$ and such that
$$
\|\Xi\|_{\alpha, \beta} \:{=} \|\Xi\|_\alpha+\|\delta \Xi\|_\beta<\infty,
$$
where $\|\Xi\|_\alpha=\sup _{s<t} \frac{\left|\Xi_{s, t}\right|}{|t-s|^\alpha}$ as usual, and also
$$
\delta \Xi_{s, u, t} :{=} \Xi_{s, t}-\Xi_{s, u}-\Xi_{u, t}, \quad\|\delta \Xi\|_\beta :{=} \sup _{s<u<t} \frac{\left|\delta \Xi_{s, u, t}\right|}{|t-s|^\beta} .
$$

\begin{lemma}[\textbf{Sewing lemma}]\label{Sewing Lemma}
    Let $\alpha$ and $\beta$ be such that $0<\alpha \leq 1<\beta$. Then, there exists a unique linear map $\mathcal{I}: C_2^{\alpha, \beta}([0, T], W) \rightarrow C^\alpha([0, T], W)$ such that $(\mathcal{I} \Xi)_0=0$ and
\begin{equation}\label{1-14}
    \left|(\mathcal{I} \Xi)_{s, t}-\Xi_{s, t}\right| \leq C|t-s|^\beta .
\end{equation}
where $C$ only depends on $\beta$ and $\|\delta \Xi\|_\beta$. $\mathcal{I}$ is called the sewing map. (The $\alpha$-H\"older norm of $\mathcal{I} \Xi$ also depends on $\|\Xi\|_\alpha$ and hence on $\left.\|\Xi\|_{\alpha, \beta}.\right)$
\end{lemma}
\begin{proof}
    See \cite{FH14} Lemma 4.2.
\end{proof}

Now we may prove Proposition \ref{Rough Integral}.

\begin{proof}[Proof of Proposition \ref{Rough Integral}]
    We apply Gubinelli's sewing lemma to the local approximations
    \begin{equation}\label{Local Approximation}
        \Xi_{s,t} := \sum_{i=1}^k Y^{(i)}_{s} \XX^{i}_{s, t} \quad \quad \forall s,t \in [0,T].  
    \end{equation}
We compute:
\begin{equation}
\begin{aligned}
     |\delta \Xi_{s,u,t}| & = |\sum_{i=1}^k (Y^{(i)}_{s}\XX^i_{s,t} - Y^{(i)}_{s}\XX^i_{s,u} - Y^{(i)}_{u}\XX^i_{u,t})| \\
     &= |\sum_{i=1}^k (Y^{(i)}_{s}(\XX^i_{u,t} + \sum_{j=1}^{i-1} \XX^{i-j}_{s,u}\otimes \XX_{u,t}^j) - Y^{(i)}_{u}\XX^i_{u,t})|  \\
     &= |\sum_{i=1}^k -Y^{(i)}_{s,u} \XX_{u,t}^{i} + \sum_{i=1}^k \sum_{j = 1}^{i-1} Y^{(i)}_s (\XX_{s,u}^{i-j}\otimes \XX^j_{u,t})|\\
     &= |\sum_{i=1}^k -Y^{(i)}_{s,u} \XX_{u,t}^{i} + \sum_{i=1}^{k-1} \sum_{j = i+1}^{k} Y^{(j)}_s (\XX_{s,u}^{j-i}\otimes \XX^i_{u,t})| \\
     &=  |\sum_{i=1}^k -(Y^{(i)}_{s,u} - \sum_{j = i+1}^k Y^{(j)}_s \XX_{s,u}^{j-i})\XX_{u,t}^{i}|\lesssim |t-s|^{(k+1)\alpha} \quad\quad \forall s<u<t\in [0,T],
\end{aligned}
\end{equation}
where the second equality is due to Chen's relation, the fourth equation is a change of index $(i,j)$ to $(j,i)$, and the last inequality is due to the rough path regularity that $\XX_{u,t}^{i}\lesssim |t-u|^{i\alpha} \leq |t-s|^{i\alpha}$ and the controlledness that $Y^{(i)}_{s,u} - \sum_{j = i+1}^k Y^{(i)}_s\cdot \XX_{s,u}^{j-i}\lesssim |u-s|^{(k+1-i)\alpha} \leq |t-s|^{(k+1-i)\alpha}$. We note that $Y^{(j)}_s \XX_{s,u}^{j-i}$ in the last line is characterized as an element in $\mathcal{L}((\mathbb{R}^d)^{\otimes ^{i}}, V)\cong \mathcal{L}((\mathbb{R}^d)^{\otimes ^{i-1}}, \mathcal{L}(\mathbb{R}^d, V))$ via
\begin{equation}
     v\mapsto Y^{(j)}_s (\XX_{s,u}^{j-i}\otimes v) \quad\quad \forall v\in (\mathbb{R}^d)^{\otimes ^i},
\end{equation}
where we recall that $ Y^{(j)}_s\in \mathcal{L}((\mathbb{R}^d)^{\otimes ^{j-1}}, \mathcal{L}(\mathbb{R}^d, V))\cong \mathcal{L}((\mathbb{R}^d)^{\otimes ^{j}}, V)$ and that $\XX_{s,u}^{j-i} \in (\mathbb{R}^d)^{\otimes^{j-i}}$, with which it is straightforward to check that
\begin{equation}
    Y^{(j)}_s (\XX_{s,u}^{j-i}\otimes \XX^i_{u,t}) = (Y^{(j)}_s \XX_{s,u}^{j-i})\XX_{u,t}^{i}
\end{equation}
in writing the fifth equality. Now, since $(k+1)\cdot\alpha>1$, Gubinelli's sewing lemma applies and the limit on the RHS of \eqref{Compensated Riemann Sum} does exist as $\mathcal{I}\Xi$ where $\mathcal{I}$ is the sewing map in the above lemma. 

Now we prove that the $k-$tuple $(\int_0^\cdot Y_u\mathrm{d}\XXX_u, Y^{(1)},..., Y^{(k-1)})=: (Z^{(1)},..., Z^{(k)})$ is again a $V$-valued controlled path. By taking $\beta=(k+1)\alpha$, it follows from the error bound \eqref{1-14} in the sewing lemma and the fact that $\XX^k\in C^{k\alpha}$:
\begin{equation}
\begin{aligned}
    |Z^{(1)}_{s,t}- \sum_{j=2}^k Z^{(j)}_{s,t}\XX^{j-1}_{s,t}| &= |(\mathcal{I}\Xi)_{s,t}- \sum_{i=1}^{k-1} Y^{(i)}_{s,t}\XX^{i}_{s,t}|\\
    & \leq |(\mathcal{I}\Xi)_{s,t}- \sum_{i=1}^{k} Y^{(i)}_{s,t}\XX^{i}_{s,t}| + |Y^{(k)}_{s,t}\XX^{k}_{s,t}| \\
    &\lesssim |t-s|^{(k+1)\alpha} + |t-s|^{k\alpha}\lesssim |t-s|^{k\alpha}\quad\quad \quad\forall 0<s,t<T.
\end{aligned}
\end{equation}
Similarly, for any $i=2,...,k-1$, we compute:
\begin{equation}
    \begin{aligned}
    |Z^{(i)}_{s,t}- \sum_{j=i+1}^k Z^{(j)}_{s,t}\XX^{j-i}_{s,t}| &= |Y^{(i-1)}_{s,t}- \sum_{j=i}^{k-1} Y^{(j)}_{s,t}\XX^{j+1-i}_{s,t}|\\
    & \leq |Y^{(i-1)}_{s,t}- \sum_{j=i}^{k} Y^{(j)}_{s,t}\XX^{j+1-i}_{s,t}| + |Y^{(k)}_{s,t}\XX^{k+1-i}_{s,t}| \\
    &\lesssim |t-s|^{(k+2-i)\alpha} + |t-s|^{(k+1-i)\alpha}\lesssim |t-s|^{(k+1-i)\alpha}\quad\quad \quad\forall 0<s,t<T,
\end{aligned}
\end{equation}
where we used \eqref{eq:Controlled Path} for the $(i-1)$-th component of the controlled path $Y$ and the fact that $\XX^{k+1-i} \in C^{(k+1-i)\alpha}$ in writing the error bounds. The controlledness of $Z$ is thus proven.
\end{proof}

\subsection{Proof of Proposition \ref{prop:Sig as CP}}\label{App:Sig as CP}
\begin{proof}[Proof of Proposition \ref{prop:Sig as CP}]
    For $n\leq k$, we prove \eqref{eq:Sig as CP} by induction on $n$. Recall the convention that $\langle \XXX, \emptyset \rangle \equiv 1$. \eqref{eq:Sig as CP} is obviously true for $n=1$. Assume \eqref{eq:Sig as CP} holds for all words of length $m<n$. Given any word $w$ of length $n$, then 
    \begin{equation}
         \Big(\langle \XXX_{s,\cdot}, w^{-1} \rangle, \langle \XXX, w^{-2} \rangle \otimes e_{w_{n-1}},\cdots, \; \langle \XXX, \emptyset \rangle \otimes e_{w_1}\otimes\cdots \otimes e_{w_{n-1}},0,...,0\Big)
    \end{equation}
    is a controlled path since $w^{-1}$ is of length $n-1$. Then, after taking the tensor product with $e_{w_n}$, one may compute the following rough integral
    \begin{equation}
    \begin{aligned}
         \int _s^t & \Big(\langle \XXX_{s,u}, w^{-1} \rangle \otimes e_{w_{n}}, \langle \XXX_{s,u}, w^{-2} \rangle \otimes e_{w_{n-1}}\otimes e_{w_{n}},\cdots, \; \langle \XXX_{s,u}, \emptyset \rangle \otimes e_{w_1}\otimes\cdots \otimes e_{w_{n}},0,...,0\Big) \mathrm{d} \XXX_u\\
        & =\lim_{|\mathcal{P}|\downarrow 0} \sum_{[u,v]\in \mathcal{P}} \sum_{i=1}^{n} \langle \XXX_{s,u}, w_{1}\cdots w_{n-i} \rangle \cdot \langle \XXX_{u,v}, w_{n-i+1}\cdots w_{n}\rangle\\
        & = \lim_{|\mathcal{P}|\downarrow 0} \sum_{[u,v]\in \mathcal{P}} \langle \XXX_{s,v}, w \rangle - \langle \XXX_{s,u}, w \rangle\\
        & = \langle \XXX_{s,t}, w \rangle \quad\quad\quad \quad\quad\quad\quad\quad\quad\quad\quad\quad\forall 0\leq s,t\leq T.
    \end{aligned}
    \end{equation}
    where the first equality is the definition of the rough integral, the second is by Chen's relation, and the last is just summing up the telescoping sum. By Proposition \ref{Rough Integral}, we have the controlledness of the integral path again.
    
    That we may extend $\XXX$ via \eqref{eq:Sig Portfolio int} follows the same argument inductively for $n>k$.
\end{proof}

\subsection{Proof of Lemma \ref{lem:Bell}}\label{App:Bell}
We verify the explicit formula of complete Bell polynomials in Lemma \ref{lem:Bell}.

\begin{proof}[Proof of Lemma \ref{lem:Bell}]
    We compute from the definition:
    \begin{equation}
        \begin{aligned}
             \exp(\sum^k_{m=1}a_m\cdot x^m)& = \sum_{n=0}^\infty \frac{1}{n!}(\sum^k_{m=1}a_m\cdot x^m)^n\\
             & = \sum_{n=0}^\infty \frac{1}{n!}\sum_{\substack{p_1,..., p_k\geq0 \\ p_1+\cdots+p_k=n}}\frac{n!}{p_1!\cdot \cdots \cdot p_k!}\prod_{m=1}^k (a_m\cdot x^m)^{p_m}\\
             & = \sum_{n=0}^\infty \sum_{\substack{p_1,..., p_k\geq0 \\ p_1+\cdots+p_k=n}}\frac{1}{p_1!\cdot \cdots \cdot p_k!}(\prod_{m=1}^k a_m^{p_m})\cdot x^{\sum_{m=1}^k m\cdot p_m}\\
             & = \sum_{n=0}^\infty \sum_{\substack{p_1,..., p_k\geq0 \\ p_1+2p_2+\cdots+k\cdot p_k=n}} (\prod_{m=1}^k \frac{a_m^{p_m}}{p_m!})\cdot x^n,
        \end{aligned}
    \end{equation}
    where we collect the terms by the total degree $n = \sum_{m=1}^k m\cdot p_m$ in the last line. The lemma then follows from the Definition \ref{BellPoly}.
\end{proof}

\subsection{Proof of Theorem \ref{RoughKYTheorem}}\label{App:Kreps-Yan}
We now prove Theorem \ref{RoughKYTheorem}. Let us first state a functional analytical lemma.

\begin{lemma}\label{Lemma6.1}
    Assume that $(NCFL)$ of order $p$ within $\mathcal{H}$ holds. For any fixed $f\in L^p_+\backslash \{0\}$, there exists $g\in L^q_+ \subset (L^p)^*$, viewed as a linear functional, such that $\mathbb{E}[h\cdot g] \leq 0$ for any $h \in \overline{C}^p(\mathcal{H}\times {\mathfrak{T}}_T)$, and such that $\mathbb{E}[f\cdot g]> 0$.
\end{lemma}
\begin{proof}
    A direct application of the Second Separation Theorem (c.f. Theorem II.9.2 in \cite{SW99}) to the closed convex set $\overline{C}^p(\mathcal{H}\times {\mathfrak{T}}_T)$ and the compact set $\{f\}$ gives us some $\alpha< \beta$ and $g\in L^q \cong (L^p)^*$ such that $g|_{\overline{C}^p(\mathcal{H}\times {\mathfrak{T}}_T)}\leq \alpha$ and  $g(f)> \beta$. Notice that $0\in \overline{C}^p(\mathcal{H}\times {\mathfrak{T}}_T)$ so $\alpha \geq 0$. Moreover, since $\overline{C}^p(\mathcal{H}\times {\mathfrak{T}}_T)$ is a cone, $\alpha$ has to be non-positive and thus $\alpha=0$ and the lemma is proven. 
\end{proof}

Now we are ready to prove the Kreps-Yan type theorem.

\begin{proof}[Proof of Theorem \ref{RoughKYTheorem}]
     We first assume there exists such an equivalent measure $\mathbb{Q}$. Then $\mathbb{E}_{\mathbb{Q}}[X]\leq0$ for every $X\in C^p(\mathcal{H}\times {\mathfrak{T}}_T)$ by assumption. Moreover, since the map $X\mapsto \mathbb{E}_{\mathbb{Q}}[X]=\mathbb{E}_{\mathbb{P}}[X\cdot\frac{\mathrm{d\mathbb{Q}}}{\mathrm{d}\mathbb{P}}]$ is continuous with respect to the weak star topology by the definition of the weak star topology and the assumption that $\frac{\mathrm{d}\mathbb{Q}}{\mathrm{d}\mathbb{P}}\in L^q_+ \subseteq(L^p)^*$, this then extends to $\overline{C}^p(\mathcal{H}\times {\mathfrak{T}}_T)$. This then implies the desired $(NCFL)$ property.\\
    Now we assume the $(NCFL)$ condition of order $p$ within $\mathcal{H}\times {\mathfrak{T}}_T$ holds. We shall show the existence of a desired equivalent measure by an exhaustion argument. Denote by $\mathcal{G}$ the set of all $g \in L_{+}^q\subseteq (L^q)^*$ such that $g|_{\overline{C}^p(\mathcal{H}\times {\mathfrak{T}}_T)} \leq 0$. Since $0 \in \mathcal{G}$, $\mathcal{G}$ is non-empty. Now let $\mathcal{S}$ be the family of (equivalence classes of) subsets of $\Omega$ formed by the supports $\{g>0\}$ of the elements $g \in \mathcal{G}$. Note that $\mathcal{S}$ is closed under countable unions, as for a sequence $\left(g_n\right)_{n=1}^{\infty} \in \mathcal{G}$, we may find strictly positive scalars $\left(\alpha_n\right)_{n=1}^{\infty}$, such that $\sum_{n=1}^{\infty} \alpha_n g_n \in \mathcal{G}$. Thus, let $\left(g_n\right)_{n=1}^{\infty} \in \mathcal{G}$ be a sequence such that $\mathbb{P}[\{g_n>0\}]\rightarrow \sup_{g \in \mathcal{G}} \mathbb{P}[\{g>0\}]$, we may choose $g_0:=\sum_{n=1}^{\infty} \alpha_n g_n \in \mathcal{G}$ for some proper choices of $\alpha_n>0$ (e.g., $ \alpha_n:=
2^{-n} /(1+\left\|g_n\right\|_q)
$). By the $\sigma$-additivity of $\mathbb{P}$, $g_0$ satisfies:
\begin{equation}
    \mathbb{P}\left[\left\{g_0>0\right\}\right]=\sup_{g \in \mathcal{G}} \mathbb{P}[\{g>0\}].
\end{equation}
We now claim that $\mathbb{P}\left[\left\{g_0>0\right\}\right]=1$, which readily shows that $g_0$ is strictly positive almost surely. Indeed, if $\mathbb{P}\left[\left\{g_0>0\right\}\right]<1$, then we could apply Lemma \ref{Lemma6.1} to $f=\mathbbm{1}_{\left\{g_0=0\right\}}$ to find $g_1 \in \mathcal{G}$ with
\begin{equation}
    \mathbb{E}_{\mathbb{P}}\left[f g_1\right]=\int_{\left\{g_0=0\right\}} g_1(\omega) d \mathbb{P}(\omega)>0.
\end{equation}
Hence, $g_0+g_1$ would be an element of $\mathcal{G}$ such that $\mathbb{P}\left[\left\{g_0+g_1>0\right\}\right]>\mathbb{P}\left[\left\{g_0>0\right\}\right]$, a contradiction. Normalise $g_0$ so that $\left\|g_0\right\|_1=1$ and let $\mathbb{Q}$ be the measure on $\mathcal{F}$ with Radon-Nikod\'ym derivative $\frac{d \mathbb{Q}}{d \mathbb{P}}=g_0$. From $g_0\leq 0$ on $\overline{C}^p(\mathcal{H}\times {\mathfrak{T}}_T)$ one implies $g_0 = 0$ on the linear subspace $K^p(\mathcal{H}\times {\mathfrak{T}}_T)$. Thus, $\mathbb{Q}$ satisfies the desired properties.
\end{proof}


\section{Some comments on the Chen-Hermite almost Brownian motion}\label{subsubsec:CHABM}

We now collect a few immediate consequences and cautions of the Chen-Hermite almost Brownian motion (cf. Definition \ref{def:CHABM}). Let us first provide some evidence for Conjecture \ref{conjecture}. First, the Gaussian moment bounds and Kolmogorov's continuity criterion yield a continuous modification with $\frac{1}{2}^-$Hölder, i.e., same sample path regularity as a standard Brownian motion. Moreover, the increment law $X_{t}-X_{s}\sim\mathcal N(0,|t-s|)$ together with the level-$2$ Chen-Hermite balancing condition implies that a Chen-Hermite almost Brownian motion must have the same variance and covariance structure as a standard Brownian motion. As illustrated in Corollary \ref{cor:joint Gaussian increments then Bm}, if joint Gaussian increments are enforced, we must have a standard Brownian motion. In other words, it suffices to prove the following conjecture:

\begin{conjecture}
    Let $(X_t)_{t\in[0,T]}$ be a $1-$dimensional continuous stochastic process with Gaussian single increments such that the Chen-Hermite balancing condition is satisfied for all adjacent time intervals, i.e.,
    \begin{equation}
        \sum_{i=1}^{n-1} \mathbb{E}[H_i(X_{s,u}, -\frac{1}{2}(u-s))\cdot H_{n-i}(X_{u,t}, -\frac{1}{2}(t-u))]=0, \quad\quad\forall0\leq s\leq u\leq t\leq T, \quad \forall n\geq 2.
    \end{equation}
    Then $(X_t)_{t\in[0,T]}$ must be a Gaussian process.
\end{conjecture}

We could not construct a counterexample to a much weaker version of the conjecture above by requiring the balancing condition for only $n=2$, i.e., only demanding uncorrelated adjacent increments. However, we note that the conjecture is wrong if we do not impose continuity on the process $X$. We now give a counterexample via a Sarmanov-type copula in the discrete setting.

\begin{example}
    Pick any two bounded, mean-zero functions $h, g \in L^2(\mathbb{R}, \varphi)$, where $\varphi$ is the standard normal density function. For a small $\varepsilon$, define a joint density on $\mathbb{R}^2$ by
\begin{equation}
    f_{\varepsilon}(x, y)=\varphi(x) \varphi(y)[1+\varepsilon(h(x) g(y)-h(y) g(x))].
\end{equation}
The construction is essentially adding some anti-symmetric perturbation to the joint density of two independent standard normals. In particular, the boundedness of the perturbation functions g and h ensures the well-definedness of the joint density function for $\varepsilon$ small.

One checks that $X, Y \sim \mathcal{N}(0,1)$ , $X+Y \sim \mathcal{N}(0,2)$, they are not jointly Gaussian but satisfy the Chen-Hermite balancing condition at all levels $n\geq 2$:
\begin{equation}
    \sum_{i=1}^{n-1} \mathbb{E}\Big[H_i\Big(X, -\frac{1}{2}\Big)\cdot H_{n-i}\Big(Y, -\frac{1}{2}\Big)\Big]=0, \quad\quad \forall n\geq 2.
\end{equation}

\begin{proof}
 The marginal Gaussianity is proven by partial integrating the joint pdf wrt $x$ ($y$ respectively), and then the mean-zero properties of the chosen perturbation functions $g$ and $h$ ensure the perturbation terms vanish after partial integration. The Gaussianity of $X+Y$ is obtained by computing the characteristic function:
\begin{equation}
    \begin{aligned}
\qquad 
\phi_{X+Y}(t) & =\mathbb{E}\left[e^{i t(X+Y)}\right] \\
& =\iint e^{i t(x+y)} \varphi(x) \varphi(y)[1+\varepsilon(h(x) g(y)-h(y) g(x))] \mathrm{d} x \mathrm{d} y \\
& =\left(\int e^{i t x} \varphi(x) \mathrm{d} x\right)^2+\varepsilon(I-I) =  (\phi_{X}(t))^2 = e^{-t^2}, 
\end{aligned}
\end{equation}
where  $I:=\iint e^{i t(x+y)} \phi(x) \phi(y) h(x) g(y) \mathrm{d} x \mathrm{d} y$ is integrable; consequently, $X+Y\sim \mathcal{N}(0,2)$. Now we check the Chen-Hermite balancing condition at level $n$ for any $n\geq 2$:
\begin{equation}
\begin{aligned}
    \sum_{i=1}^{n-1} \mathbb{E}\Big[H_i\Big(X, -\frac{1}{2}\Big)\cdot & H_{n-i}\Big(Y, -\frac{1}{2}\Big)\Big] = \int_{\mathbb{R}}\int_{\mathbb{R}} \sum_{i=1}^{n-1} H_i\Big(x, -\frac{1}{2}\Big)\cdot H_{n-i}\Big(y, -\frac{1}{2}\Big)\cdot \varphi(x)\varphi(y) \mathrm{d}x \mathrm{d}y \\ & + \int_{\mathbb{R}}\int_{\mathbb{R}} \sum_{i=1}^{n-1} H_i\Big(x, -\frac{1}{2}\Big)\cdot H_{n-i}\Big(y, -\frac{1}{2}\Big)\cdot  \varepsilon(h(x) g(y)-h(y) g(x)) \varphi(x)\varphi(y) \mathrm{d}x \mathrm{d}y\\ 
    & = 0+ \int_{\mathbb{R}}\int_{\mathbb{R}} \sum_{i=1}^{n-1} H_i\Big(x, -\frac{1}{2}\Big)\cdot H_{n-i}\Big(y, -\frac{1}{2}\Big)\cdot  \varepsilon(h(x) g(y)-h(y) g(x)) \varphi(x)\varphi(y) \mathrm{d}x \mathrm{d}y\\
    & = 0
\end{aligned}
\end{equation}
where the first line is from integrating the independent coupling part of the joint distribution, and the last equality is due to the whole integrand being anti-symmetric in $x$ and $y$. Now we finally show $(X,Y)$ is not jointly Gaussian. The perturbation in the joint density is antisymmetric in $(x, y)$, so $f_{\varepsilon}(x, y) \neq f_{\varepsilon}(y, x)$ for $\varepsilon \neq 0$. Any bivariate normal with identical marginals is symmetric; hence, this joint law is not Gaussian.     
\end{proof}

Now, we may consider a discrete process $(X_n)_{n=0,1,2}$, with 
\begin{equation}
    X_0:=0,\quad X_1:=X,\quad X_2=X+Y.
\end{equation}
The above construction ensures that the process has Gaussian increments, the adjacent increments satisfy the Chen-Hermite balancing condition for every $n\geq 2$, but the process fails to be a Gaussian process.

\end{example}

\section{Consistency between Rough Paths and Controlled Paths}\label{app:Consistecy}

Now fix $\alpha\in (\frac{1}{3}, \frac{1}{2}]$ and  $\XXX:=(X,\XX)$ an $\alpha$-H\"older rough path in $V:=\mathbb{R}^d$. We use the notation $(Y,Y')$ for controlled paths, where $Y'$ is usually referred to as the Gubinelli derivative in literature. We list some consistency results between controlled paths and rough paths, which are used in Subsection \ref{subsec:RDE models}. Most of the statements are standard, and the proof can be found in \cite{FH14} Chapter 7 or \cite{All21}. If not, a proof will be given.

\begin{proposition}\label{1.2.9}
    Let $\left(Y, Y^{\prime}\right)$ an $\mathcal{L}(W, \bar{W})-$valued controlled path  and  $(Z, Z')$ an $W$-valued controlled paths for some Euclidean spaces $\bar{W}$ and $W$. Then we have a well-defined integral via:
    \begin{equation}\label{CPintCP}
        \int (Y, Y') \mathrm{d} (Z, Z') := \lim _{|\mathcal{P}| \downarrow 0} \sum_{[s, t] \in \mathcal{P}} (Y_s Z_{s, t}+Y_s^{\prime} Z_s^{\prime} \mathbb{X}_{s, t}).
    \end{equation}
    where we identify $Y_s'Z_s'$ as a linear map via the canonical isomorphism $\mathcal{L} (V, \mathcal{L}(W, \bar{W})) \times \mathcal{L} (V, W) \cong \mathcal{L} (V\otimes V, \bar{W})$. Moreover, the integral is again controlled by $\XXX$ with Gubinelli derivative $YZ$.
\end{proposition}
\begin{proof}
    See \cite{FH14} Remark 4.12.
\end{proof}

In particular, associating a controlled path with its own iterated integral in this sense allows us to translate it into a rough path:

\begin{corollary}[Controlled paths as rough paths]\label{CPasRP}
    Let $(Y, Y')$ be a path controlled by $\XXX$. Then $\mathbf{Y}:=(Y, \mathbb{Y})$ is also an $\alpha$-H\"older rough path, where
\begin{equation}
    \mathbb{Y}_{s, t} :=\int_s^t Y_{s, r} \otimes d Y_r {:=} \lim _{|\mathcal{P}| \downarrow 0} \sum_{[u,v]\in \mathcal{P}} (Y_{s,u} \otimes Y_{u, v}+Y_u^{\prime} \otimes Y_u^{\prime} \mathbb{X}_{u, v}).
\end{equation}
In particular, by setting $(Y, Y'):=(X, Id)$ we recover $\mathbf{Y}=\XXX$. Note that the integral is exactly in the sense of Proposition \ref{1.2.9}.
\end{corollary}
\begin{proof}
    See \cite{FH14} Section 7.1.
\end{proof}

In particular, the construction preserves the geometric property:

\begin{proposition}[Geometricity is preserved]\label{WGpreserved}
    Let $\XXX$, $(Y, Y')$ and $\mathbf{Y}$ be as above. Moreover, assume that $\XXX$ is geometric. Then $\mathbf{Y}$ is again geometric.
\end{proposition}
\begin{proof}
    Set $\Xi_{u,v} := Y_{s,u} \otimes Y_{u, v}+Y_u^{\prime} \otimes Y_u^{\prime} \mathbb{X}_{u, v}$ as above. By the geometricity of $\XXX$ we have:
    \begin{equation}
        2\mathrm{Sym}(Y_u^{\prime} \otimes Y_u^{\prime} \mathbb{X}_{u, v}) =Y_u^{\prime} \otimes Y_u^{\prime} 2\mathrm{Sym}(\mathbb{X}_{u, v}) = Y_u^{\prime} \otimes Y_u^{\prime} (X_{u,v}\otimes X_{u,v}). 
    \end{equation}
    Thus, for any $i,j=1,...,n$ we have
    \begin{equation}\label{WGLocApprox}
    \begin{aligned}
        \Xi_{u,v}^{i,j}+ \Xi_{u,v}^{j,i} &= Y^i_{s,u} Y^j_{u, v}+ Y^j_{s,u} Y^i_{u, v}+ Y'^i_u X_{u,v}\cdot Y'^j_u X_{u,v}\\
        &= Y^i_{s,u} Y^j_{u, v}+ Y^j_{s,u} Y^i_{u, v}+ Y^i_{u,v}\cdot Y^j_{u,v}+ O\left(|v-u|^{3 \alpha}\right)\\
        &= -(Y^i_s Y^j_{u,v} + Y^j_s Y^i_{u,v}) + (Y^i_v Y^j_v - Y^i_u Y^j_u) + O\left(|v-u|^{3 \alpha}\right)
    \end{aligned}
    \end{equation}
    where the second equality is by the simple observation that $Y^i$ is again controlled by $X$ with Gubinelli derivative $Y'^i$ for any component $Y^i$, and the third equality is a simple computation. Now summing them up along any $\mathcal{P}$ and letting $|\mathcal{P}|\downarrow 0$, we obtain as desired $\mathbb{Y}^{i,j}_{s,t}+\mathbb{Y}^{j,i}_{s,t} = Y^i_{s,t} Y^j_{s,t}.$
\end{proof}

More concretely, the renormalisation terms in controlled-paths-lifted rough paths can be computed explicitly if the controlled paths are from a rough integral.
\begin{proposition}\label{prop:renormalisation succeeded}
    Let $\mathbf{X}=(X, \mathbb{X}) $ be as above but now let $d=1$ and let $\left(K, K^{\prime}\right)$ be a controlled path. Recall that $\left(Z, Z^{\prime}\right):=\left(\int_0^{\cdot} K_u \mathrm{~d} \mathbf{X}_u, K\right)$ is again a controlled path as in Theorem \ref{Rough Integral}. Let $\mathbf{Z}=(Z, \mathbb{Z})$ be the canonical rough path lift of $Z$, as defined in Corollary \ref{CPasRP}. Moreover, let $G_{\XXX} $ and $G_{\mathbf{Z}}$ be the renormalisation functions encoded in $\XXX$ and $\mathbf{Z}$ as in Proposition \ref{Rough Paths via Bell Polynomials}, respectively. Then

\begin{equation}
    G_{\mathbf{Z}}(t)=\int_0^t\left(K_u \cdot K_u\right) \mathrm{d}G_{\mathbf{X}}(u)
\end{equation}
where the integral on the right-hand side is a Young integral.
\end{proposition}
\begin{proof}
    See \cite{All21} Lemma 6.8.
\end{proof}

 Another important property of this embedding is that the integral of a controlled path against another controlled path (cf. Proposition \ref{1.2.9}) is the same as when viewing the integrator as a rough path and the integral in the sense of (\ref{Compensated Riemann Sum}), as in the following result:

\begin{proposition}[Consistency]\label{Consistency}
    Let $\XXX$, $(Y, Y')$ and $\mathbf{Y}$ be as above. If $\left({Z}, {Z}^{\prime}\right)$ is controlled by $\mathbf{Y}$, then we have the equality
\begin{equation}
    \int_s^t (Z_r, Z'_r) \mathrm{d} \mathbf{Y}_r = \int_s^t (Z_r, Z'_rY'_r) \mathrm{~d} (Y_r, Y'_r)
\end{equation}
where $(Z_r, Z'_rY'_r)$ is a path controlled by $X$. The left-hand side is a rough integral as in \eqref{Compensated Riemann Sum}, while the right-hand side is in the sense of (\ref{CPintCP}).
\end{proposition}
\begin{proof}
     See \cite{All21} Proposition 6.1.
\end{proof}

One also has the following associativity property for rough integrations:

\begin{proposition}[Associativity]\label{Ass}
Let $\mathbf{X}$ be as above and $\left(Y, Y^{\prime}\right),\left(K, K^{\prime}\right)$ two controlled paths defined in proper Euclidean spaces such that the two integrals in (\ref{Associativity}) exist. Moreover, let $\left(Z, Z^{\prime}\right) := \left(\int_0^{\cdot} (K_u, K'_u) \mathrm{d} \mathbf{X}_u, K\right) $ be another controlled path by Theorem \ref{Rough Integral}. Then
\begin{equation} \label{Associativity}
    \int_0 (Y_u ,Y'_u) \mathrm{~d} (Z_u, Z'_u) =\int_0 (Y_u K_u, Y'_uK_u + Y_uK'_u)\mathrm{d} \mathbf{X}_u
\end{equation}
where the integral on the left-hand side is in the sense of (\ref{CPintCP}) and the integral on the right-hand side is a rough integral as in \eqref{Compensated Riemann Sum}.
\end{proposition}
\begin{proof}
    See \cite{All21} Proposition 6.2.
\end{proof}

\begin{remark}
    We note that the results in this appendix cannot be extended to $\alpha<\frac{1}{3}$ for general (non-geometric) rough paths. One needs to employ the notion of branched rough paths (cf. \cite{Gub10}) to develop analogous techniques.
\end{remark}

\bibliographystyle{apacite}
\bibliography{main.bib}
\end{document}